\documentclass[12pt]{article}%

\usepackage[T1]{fontenc}%
\usepackage{graphicx}%
\usepackage[cm]{fullpage}%
\usepackage{paralist}%
\usepackage{amsmath}%
\usepackage{amssymb}%
\usepackage{euscript}%
\usepackage{xcolor}%
\usepackage{colortbl}%
\usepackage{hhline}%
\usepackage{xspace}%
\usepackage{styles/picins}%
\usepackage{stmaryrd}%
\usepackage{mleftright}%
\usepackage{multirow}%
\usepackage{calc}%

\usepackage{titlesec}%
\titlelabel{\thetitle. }%

\usepackage{amsmath}%
\usepackage[amsmath,thmmarks]{ntheorem}
\theoremseparator{.}

\usepackage[hypcap=true]{caption}

\usepackage{hyperref}%
\hypersetup{%
   breaklinks,%
   colorlinks=true,%
   linkcolor=[rgb]{0.45,0.0,0.0},%
   citecolor=[rgb]{0,0,0.45}
}

\IfFileExists{sariel_computer.sty}{\def\sarielComp{1}}{}
\ifx\sarielComp\undefined%
\newcommand{\SarielComp}[1]{}
\newcommand{\NotSarielComp}[1]{#1}%
\else
\newcommand{\SarielComp}[1]{#1}%
\newcommand{\NotSarielComp}[1]{}%
\fi
\SarielComp{\usepackage{sariel_colors}}%
\providecommand{\Mh}[1]{#1}%

\numberwithin{figure}{section}
\numberwithin{equation}{section}%
\numberwithin{table}{section}%

\newtheorem{theorem}{Theorem}[section] 

\newtheorem{defn}[theorem]{Definition}
\newtheorem{observation}[theorem]{Observation}

\newtheorem{lemma}[theorem]{Lemma}%
\newtheorem{proposition}[theorem]{Proposition}%

\theoremstyle{remark}%
\theoremheaderfont{\sf}
\theorembodyfont{\upshape}%
\newtheorem{remark}[theorem]{Remark}%

\newtheorem{problem}[theorem]{Problem}

\newcommand{\ProblemC}[1]{\textsf{#1}}

\newcommand{\PrelimVersions}{%
   \begin{minipage}[t]{0.95\linewidth}
       A preliminary version of this paper appeared as:
       \begin{compactenum}[(A)]
           \item%
           A.~Adamaszek and A.~Wiese.  Approximation schemes for
           maximum weight independent set of rectangles.  In
           \si{\emph{ \si{Proc.}  54\th \si{Annu. }IEEE \si{Sympos.}
                 Found. \si{Comput.} Sci.}}  (FOCS), pages 400--409,
           2013.%

           \item
           A.~Adamaszek and A.~Wiese.  %
           A {\QPTAS} for maximum weight independent set of polygons
           with polylogarithmically many vertices. %
           In \emph{\si{Proc. 25\th ACM-SIAM \si{Sympos.} Discrete
                 \si{Algs.}}}  (SODA), pages 645--656, 2014.%
           \item%
           S.~{Har-Peled}.  Quasi-polynomial time approximation scheme
           for sparse subsets of polygons. %
           In \emph{\si{\si{Proc.} 30\th \si{Annu. Sympos. Comput.}
                 Geom.}}  \textrm(SoCG), pages 120--129,
           2014.
           The full version of the paper is available from the \si{arXiv} \cite{h-qptas-13}.
       \end{compactenum}%
       \smallskip
   \end{minipage}
}

\newcommand{\myparagraph}[1]{\paragraph*{#1{}.}}%

\newcommand{\DGraph}{\mathsf{D}}%

\newcommand{\Graph}{\mathsf{G}}%

\newcommand{\LineSpanX}[1]{\mathrm{line}\pth{#1}}
\newcommand{\CandidX}[1]{C\pth{#1}}%
\newcommand{\CandidAll}{C}%

\newcommand{\TGrid}{\mathcal{G}}

\newcommand{\MDC}{\mathcal{M}}
\newcommand{\MD}[1]{\MDC\pth{#1}}

\newcommand{\Sq}{\ensuremath{\Box}\xspace}

\newcommand{\Edges}{\mathsf{E}}%
\newcommand{\EdgesX}[1]{\mathsf{E}\pth{#1}}%

\newcommand{\Vertices}{\mathsf{V}}%
\newcommand{\VerticesX}[1]{\mathsf{V}\pth{#1}}%

\newcommand{\Corridor}{\Mh{C}}

\newcommand{\eps}{{\varepsilon}}%

\newcommand{\PntSetA}{\mathsf{Q}}%

\newcommand{\PolySet}{\mathcal{P}}%
\newcommand{\XSet}{\mathcal{X}}%

\newcommand{\Sample}{\mathcal{S}}%

\newcommand{\PolySetA}{\mathcal{Q}}%
\newcommand{\PolySetB}{\mathcal{B}}%

\newcommand{\poly}{\mathrm{poly}}
\newcommand{\polylog}{\mathrm{polylog}}

\newcommand{\Family}{\EuScript{F}}%
\newcommand{\FamilyH}[1]{\EuScript{F}_{\geq #1}}%

\newcommand{\CDC}{\Mh{\EuScript{C}}}%
\newcommand{\CSet}{\EuScript{C}'}%
\newcommand{\CD}[1]{\EuScript{C}\pth{ #1}}%

\newcommand{\CDH}[2]{\mathcal{C}_{\geq #2}\pth{ #1}}%

\newcommand{\BSet}{\mathcal{B}}
\newcommand{\BSetA}{\mathcal{B}'}

\newcommand{\BOpt}{\mathcal{B}_\mathrm{opt}}

\newcommand{\RSet}{\mathcal{R}}
\newcommand{\RSetA}{\mathcal{R}'}

\newcommand{\rxL}[1]{x_{#1}}%
\newcommand{\rxR}[1]{x'_{#1}}%
\newcommand{\ryB}[1]{y_{#1}}%
\newcommand{\ryT}[1]{y'_{#1}}%

\newcommand{\rH}[1]{d{}y_{#1}}%
\newcommand{\rW}[1]{d{}x_{#1}}%

\newcommand{\GWidth}{\Delta}

\newcommand{\OIntY}[2]{\mleft(#1, #2 \mright)}%

\let\P\undefined

\newcommand{\nopt}{m_{\mathrm{opt}}}
\newcommand{\wopt}{W_{\mathrm{opt}}}
\newcommand{\woptX}[1]{W_{\mathrm{opt}}\pth{#1}}

\newcommand{\Opt}{\mathcal{O}}

\newcommand{\DefSetC}{\Mh{D}}%
\newcommand{\DefSet}[1]{\DefSetC\pth{#1}}
\newcommand{\KillSet}[1]{\Mh{K}\pth{#1}}

\newcommand{\Ex}[1]{\mathop{\mathbf{E}}\pbrc{#1}}
\newcommand{\sep}[1]{\,\left|\, {#1} \right.}
\newcommand{\sepw}[1]{\left|\, {#1} \right.}
\newcommand{\pbrc}[1]{\mleft[ {#1} \mright]}

\newcommand{\cardin}[1]{\left| {#1} \right|}%
\newcommand{\Set}[2]{\left\{ #1 \;\middle\vert\; #2 \right\}}
\newcommand{\pth}[1]{\mleft({#1}\mright)}
\newcommand{\brc}[1]{\left\{ {#1} \right\}}
\newcommand{\ceil}[1]{\left\lceil {#1} \right\rceil}

\newcommand{\HLinkShort}[2]{\hyperref[#2]{#1\ref*{#2}}}
\newcommand{\HLink}[2]{\hyperref[#2]{#1~\ref*{#2}}}
\newcommand{\HLinkPage}[2]{\hyperref[#2]{#1~\ref*{#2}%
      $_\text{p\pageref{#2}}$}}
\newcommand{\HLinkPageOnly}[1]{\hyperref[#1]{Page~\refpage*{#1}%
      $_\text{p\pageref{#1}}$}}

\newcommand{\HLinkSuffix}[3]{\hyperref[#2]{#1\ref*{#2}{#3}}}
\newcommand{\HLinkPageSuffix}[3]{\hyperref[#2]{#1\ref*{#2}%
      #3$_\text{p\pageref{#2}}$}}

\newcommand{\figlab}[1]{\label{fig:#1}}
\newcommand{\figref}[1]{\HLink{Figure}{fig:#1}}

\newcommand{\seclab}[1]{\label{sec:#1}}
\newcommand{\secref}[1]{\HLink{Section}{sec:#1}}
\newcommand{\secrefpage}[1]{\HLinkPage{Section}{sec:#1}}

\providecommand{\deflab}[1]{\label{def:#1}}
\newcommand{\defref}[1]{\HLink{Definition}{def:#1}}
\newcommand{\defrefpage}[1]{\HLinkPage{Definition}{def:#1}}

\newcommand{\itemlab}[1]{\label{item:#1}}
\newcommand{\itemref}[1]{\HLinkSuffix{(}{item:#1}{)}}
\newcommand{\itemrefpage}[1]{\HLinkPageSuffix{(}{item:#1}{)}}

\newcommand{\remlab}[1]{\label{rem:#1}}
\newcommand{\remref}[1]{\HLink{Remark}{rem:#1}}%

\newcommand{\problab}[1]{\label{prob:#1}}

\newcommand{\lemlab}[1]{\label{lemma:#1}}
\newcommand{\lemref}[1]{\HLink{Lemma}{lemma:#1}}%
\newcommand{\lemrefpage}[1]{\HLinkPage{Lemma}{lemma:#1}}

\newcommand{\proplab}[1]{\label{prop:#1}}
\newcommand{\propref}[1]{\HLink{Proposition}{prop:#1}}%

\newcommand{\obslab}[1]{\label{observation:#1}}

\newcommand{\thmlab}[1]{{\label{theo:#1}}}
\newcommand{\thmref}[1]{\HLink{Theorem}{theo:#1}}
\newcommand{\thmrefpage}[1]{\HLinkPage{Theorem}{theo:#1}}

\providecommand{\eqlab}[1]{}%
\renewcommand{\eqlab}[1]{\label{equation:#1}}

\newcommand{\Eqrefpage}[1]{\HLinkPageSuffix{Eq.~(}{equation:#1}{)}}

\renewcommand{\th}{th\xspace}
\newcommand{\cbrc}[2]{\mleft[ #1 \;\middle\vert\; #2 \mright]}

\newcommand{\Prob}[1]{\mathop{\mathbf{Pr}}\pbrc{#1}}
\newcommand{\ProbC}[2]{\mathop{\mathbf{Pr}}\cbrc{#1}{#2}}

\definecolor{newResultColor}{rgb}{0.8,0.8,0.95}%
\newcommand{\newResultC}{\cellcolor{newResultColor}}
\definecolor{TableTitleColor}{rgb}{0.9,0.9,0.9}%
\newcommand{\TableTitleC}{\cellcolor{TableTitleColor}}

\definecolor{blue25}{rgb}{0,0,0.55}%
\providecommand{\emphic}[2]{%
   \textcolor{blue25}{%
      \textbf{\emph{#1}}}%
   \index{#2}}%
\providecommand{\emphi}[1]{\emphic{#1}{#1}}

\DefineNamedColor{named}{OliveGreenX}    {rgb}{0.1934,0.395,0.340}
\providecommand{\ComplexityClass}[1]{{{\textcolor[named]{OliveGreenX}{%
      {\textsc{#1}}}}}}

\newcommand{\NP}{{\ComplexityClass{N\ensuremath{%
            \text{\textsc{P}}}}\xspace}}%
\newcommand{\P}{{\ComplexityClass{{\ensuremath{\text{\textsc{P}}}}}\xspace}}%

\providecommand{\APXHard}{{\textsf{APXHard}}%
   \index{NP!hard}\xspace}

\newcommand{\PTAS}{\textsf{PTAS}\xspace}
\newcommand{\QPTAS}{\textsf{QPTAS}\xspace}

\renewcommand{\Re}{{\rm I\!\hspace{-0.025em} R}}

\newcommand{\Nat}{{\rm I\!\hspace{-0.025em} N}}

\newcommand{\totalI}{\mathsf{t}}

\newcommand{\facesX}[1]{\mathcal{F}}%
\newcommand{\facesPX}[1]{\mathcal{F}_{\!+}^{}}
\newcommand{\facesPZ}{\facesPX{\LZ}}

\newcommand{\ArrX}[1]{\mathcal{A}\pth{#1}}

\newcommand{\bd}{\partial}

\newcommand{\cFunc}[1]{u\pth{#1}}

\newcommand{\weightX}[1]{w\pth{#1}}

\newcommand{\etal}{\textit{et~al.}\xspace}

\newcommand{\Property}{\Pi}
\newcommand{\PropertyF}[1]{\Property_{ #1}}

\newcommand{\I}{\mathcal{I}}

\newcommand{\Polygon}{\sigma}
\newcommand{\CutPolygon}{\Gamma}
\newcommand{\PolygonA}{\tau}
\newcommand{\PolygonB}{\phi}

\newcommand{\AnnaThanks}[1]{%
   \thanks{%
      Department of Computer Science (DIKU); %
      University of Copenhagen; %
      Denmark; %
      \texttt{anad\atgen{}di.ku.dk}. %
      #1%
   }%
}

\newcommand{\AndreasThanks}[1]{%
   \thanks{%
      Max-Planck-Institut f\"ur Informatik; %
      Saarbr\"ucken; Germany; %
      \texttt{awiese\atgen{}mpi-inf.mpg.de}.%
      #1%
   }%
}

\newcommand{\SarielThanks}[1]{%
   \thanks{%
      Department of Computer Science; %
      University of Illinois; %
      201 N. Goodwin Avenue; %
      Urbana, IL, 61801, USA; %
      {\tt \si{sariel}\atgen{}\si{uiuc.edu}}; %
      {\tt \url{http://sarielhp.org}.}%
      #1%
   }%
}

\newcommand{\si}[1]{#1}
\newcommand{\atgen}{\symbol{'100}}%

\newcommand{\myqedsymbol}{\rule{2mm}{2mm}}

\theoremheaderfont{\em}%
\theorembodyfont{\upshape}%
\theoremstyle{nonumberplain}
\theoremseparator{}
\theoremsymbol{\myqedsymbol}
\newtheorem{proof}{Proof:}

\newcommand{\XSays}[2]{\NONONONONO{ {$\rule[-0.12cm]{0.2in}{0.5cm}$\fbox{\tt #1:}
      } #2 \marginpar{\textcolor{red}{#1}}
      {$\rule[0.1cm]{0.3in}{0.1cm}$\fbox{\tt
            end}$\rule[0.1cm]{0.3in}{0.1cm}$} } }
\newcommand{\sariel}[1]{{\XSays{Sariel}{#1}}}

\newcommand{\andy}[1]{{\XSays{Andy}{#1}}}

\newcommand{\anna}[1]{{\XSays{Anna}{#1}}}
\newcommand{\NotHandled}[1]{}%

\DefineNamedColor{named}{RedViolet} {cmyk}{0.07,0.90,0,0.34}

\newcommand{\ds}{\displaystyle}

\newcommand{\trail}{trail\xspace}%
\newcommand{\ring}{ring\xspace}%
\newcommand{\IRectSetX}[1]{I\pth{#1}}%

\newcommand{\Ring}{Ring\xspace}%

\newcommand{\Term}[1]{\textsf{#1}}%

\newcommand{\SETH}{\Term{SETH}\xspace}

\newcommand{\minitab}[2][l]{\begin{tabular}{#1}#2\end{tabular}}
\newcommand{\HHLine}[1]{%
      \doublerulesepcolor{white}%
      \hhline{#1}%
      \doublerulesepcolor{black}%
}

\newcommand{\tpoly}{\Mh{T}}%

\newcommand{\Path}{\pi}

\newcommand{\CC}{\Mh{\xi}}%
\newcommand{\CCA}{\rho}

\newcommand{\face}{\Xi}

\newcommand{\pnt}{\ensuremath{\mathsf{p}}\xspace}

\newcommand{\segMax}{\seg_{\mathrm{max}}}

\newcommand{\edge}{\Mh{\mathsf{e}}}%
\newcommand{\edgeB}{\mathsf{e}'}%

\newcommand{\rect}{\mathsf{r}}

\newcommand{\Block}{\mathsf{b}}%

\newcommand{\RectSet}{\mathcal{R}}%
\newcommand{\FC}{\mathcal{F}}%
\newcommand{\IntRange}[1]{\left\llbracket #1 \right\rrbracket}

\newcommand{\lenX}[1]{\mleft\| #1 \mright\|}
\newcommand{\GCell}{\Box}%
\newcommand{\GCellA}{\Box'}%
\newcommand{\interX}[1]{\mathrm{int}\pth{#1}}%

\newcommand{\seg}{\mathsf{s}}%
\newcommand{\segA}{\mathsf{t}}%
\newcommand{\segB}{\mathsf{u}}%
\newcommand{\bSeg}{\mathsf{v}}

\newcommand{\SegSet}{\mathsf{S}}%

\newcommand{\LX}{\mathcal{X}}
\newcommand{\LY}{{\mathcal{Y}}}%

\newcommand{\LSet}{\mathcal{Z}}
\newcommand{\ULSet}{\cup\mathcal{Z}}

\newcommand{\constA}{%
   \ensuremath{%
      \left(\frac{1}{\eps\cdot\delta}\right)^{4}%
   } \xspace%
}
\newcommand{\constAmath}{%
   \ensuremath{%
      1/ (\eps\delta)^{4}%
   } \xspace%
}
\newcommand{\constB}{1/(\eps\delta)^{O(1)}}

\newcommand{\EdgeSet}{{\mathcal{E}}}%
\newcommand{\shortcut}{\psi}%

\newcommand{\segCY}[2]{{#1 #2}}%
\newcommand{\segOY}[2]{\interX{#1 #2}}%

\newcommand{\SetX}{X}
\newcommand{\SetY}{Y}

\newcommand{\reach}{reach\xspace}%

\newcommand{\remove}[1]{}%
\newcommand{\clX}[1]{\mathrm{cl}\pth{#1}}

\newcommand{\NInt}{\zeta} 

\newcommand{\AllRegions}{\mathcal{R}}%
\newcommand{\LAllRegions}{\mathcal{R}_L}%
\newcommand{\TAllRegions}{\mathcal{R}_T}%

\newcommand{\minipageW}[2]{%
   \begin{minipage}{\widthof{#1}}%
       #2
   \end{minipage}%
    }

\newcommand{\IncludeGraphics}[2][]{%
   \typeout{}%
   \typeout{Graphics: #2}%
   \typeout{\ includegraphics[#1]{#2}}%
   \includegraphics[#1]{#2}
   \typeout{}%
}

\begin{document}

\title{Approximation Schemes for Independent Set and Sparse Subsets of
   Polygons%
   \thanks{%
      \protect\PrelimVersions{}
   }%
}

\author{%
   Anna Adamaszek%
   \AnnaThanks{Supported by the Danish Council for Independent
      Research \si{DFF-MOBILEX} mobility grant.}%
   \and%
   Sariel Har-Peled%
   \SarielThanks{Work on this paper was partially supported by a NSF
      AF award CCF-1217462.%
   }%
   \and%
   Andreas Wiese%
   \AndreasThanks{}%
}%

\date{\today}

\maketitle

\begin{abstract}
    We present an $(1+\eps)$-approximation algorithm with
    quasi-polynomial running time for computing the maximum weight
    independent set of polygons out of a given set of polygons in the
    plane (specifically, the running time is
    $n^{O( \poly( \log n, 1/\eps))}$). Contrasting this, the best
    known polynomial time algorithm for the problem has an
    approximation ratio of~$n^{\eps}$.  Surprisingly, we can extend
    the algorithm to the problem of computing the maximum weight
    subset of the given set of polygons whose intersection graph
    fulfills some sparsity condition. For example, we show that one
    can approximate the maximum weight subset of polygons, such that
    the intersection graph of the subset is planar or does not contain
    a cycle of length $4$ (i.e., $K_{2,2}$). Our algorithm relies on a
    recursive partitioning scheme, whose backbone is the existence of
    balanced cuts with small complexity that intersect polygons from
    the optimal solution of a small total weight.

    For the case of large axis-parallel rectangles, we provide a
    \emph{polynomial} time $(1+\eps)$-approximation for the maximum
    weight independent set. Specifically, we consider the problem
    where each rectangle has one edge whose length is at least a
    constant fraction of the length of the corresponding edge of the
    bounding box of all the input elements.  This is now the most
    general case for which a \PTAS is known, and it requires a new and
    involved partitioning scheme, which should be of independent
    interest.
\end{abstract}


\section{Introduction}

In this paper we study the \emph{Independent Set of Polygons} problem.
We are given a set $\PolySet = \brc{\Polygon_1, \ldots, \Polygon_m}$
of $m$ simple polygons in the plane, with weights
$w_1,w_2, \ldots, w_m > 0$, respectively, encoded by $n$ input bits.
Our goal is to find an independent set of polygons from $\PolySet$ of
a maximum total weight. A set of polygons is \emphi{independent}, if
no two polygons from the set intersect, where we treat polygons as
open sets.

This problem and its special cases arise in various settings such as
\begin{inparaenum}[(i)]
    \item channel admission control~\cite{lno-racna-02},
    \item chip manufacturing~\cite{hm-ascpp-85},
    \item map labeling \cite{aks-lpmis-98, fmp-acggi-00, va-oamis-99},
    \item cellular networks~\cite{ccj-udg-90},
    \item unsplittable flow~\cite{aglw-ciglu-13, bsw-cfaau-11},
    \item data mining~\cite{fmmt-dmotd-01, kmp-artp-98, lsw-car-97},
    and many others.
\end{inparaenum}

A natural approach to this problem is to build an \emphi{intersection
   graph} $\Graph = (\Vertices,\Edges)$, where we have one vertex for
each input polygon and two vertices are connected by an edge if and
only if their corresponding polygons intersect. The weight of each
vertex equals the weight of its corresponding polygon. The task at
hand is to compute the maximum weight independent set in $\Graph$.  In
general graphs, even the unweighted maximum independent set problem
does not allow an approximation factor within
$\cardin{ \Vertices }^{1-\eps }$ for any $\varepsilon >0$, if
$\NP \neq \P$ \cite{z-ldeim-07}.  Surprisingly, even if the maximum
degree of the graph is bounded by $3$, no \PTAS is possible
\cite{bf-apisp-99} (assuming that $\NP \neq \P$). However, in our case
the intersection graph stems from geometric objects, and we can make
use of the exact locations of the input polygons in our
computations. As we demonstrate, this allows obtaining much better
approximation factors.

\myparagraph{Fat (convex) polygons} %
If the input objects are fat (e.g., disks or squares), \PTAS{}es are
known.  One approach \cite{c-ptasp-03,ejs-ptasg-05} relies on a
hierarchical spatial subdivision, such as a quadtree, combined with
dynamic programming techniques \cite{a-ptase-98}. This approach works
even in the weighted case.  Another approach \cite{c-ptasp-03} relies
on a recursive application of a nontrivial generalization of the
planar separator theorem \cite{lt-stpg-79, sw-gsta-98}. However, this
approach is limited to the unweighted case.

\myparagraph{Axis-parallel rectangles}

The problem turns out to be significantly harder already for the
setting of axis-parallel rectangles.  No constant factor approximation
algorithms are known in this setting, while the best known hardness
result is strong $\mathsf{NP}$-hardness~\cite{fpt-opcpn-81,
   ia-fccmc-83}.  This gap remains despite a lot of research on the
problem \cite{aks-lpmis-98, bdmr-eaatp-01, cc-misr-09, c-nmisr-04, %
   ch-aamis-12, fpt-opcpn-81, ia-fccmc-83, \si{kmp-artp-98},
   \si{lno-racna-02}, \si{n-fsbhd-00}}.  For the weighted case, there
are several $O(\log m)$ approximation algorithms
known~\cite{aks-lpmis-98,kmp-artp-98,n-fsbhd-00}, and the hidden
constant can be made arbitrarily small, since for any constant $k$
there is a $\ceil{ \log_{k}m }$-approximation algorithm due to Berman
\etal \cite{bdmr-eaatp-01}. Chan and Har-Peled \cite{ch-aamis-12}
provided an $O( \log m / \log \log m)$-approximation for the weighted
case.  For the unweighted case, an $O( \log \log m)$-approximation was
given by Chalermsook and Chuzhoy \cite{cc-misr-09}.

Some algorithms have been studied which perform better for special
cases of the problem.  There is a $4q$-approximation algorithm due to
Lewin-Eytan \etal \cite{lno-racna-02} where $q$ denotes the size of
the largest clique in the given instance. In case when the optimal
independent set has size $\beta m$ for some $\beta \leq 1$, Agarwal
and Mustafa present an algorithm which computes an independent set of
size $\Omega(\beta^{2}m)$~\cite{am-isigc-06}.

\myparagraph{Other input objects}
For the case when the input instance is a collection of $m$ line
segments, an $O\big((\nopt)^{1/2+o(1)})$-approximation was developed
by Agarwal and Mustafa \cite{am-isigc-06}, where $\nopt$ is the size
of the optimal solution.  Fox and Pach~\cite{fp-cinig-11} have
improved the approximation factor to $m^\eps$ for line segments, and
also curves that intersect a constant number of times.  Their argument
relies on the intersection graph having a large biclique if it is
dense, and a cheap separator
otherwise. 

For an independent set of unweighted pseudo-disks, Chan and Har-Peled
\cite{ch-aamis-12} provided a \PTAS. Surprisingly, their algorithm is
a simple local search strategy that relies on using the planar
separator theorem to argue that if the local solution is far from the
optimal, then there is a ``small'' beneficial exchange.

\myparagraph{The challenge}

Although the complexity of geometric independent set is
well-understood in the setting of squares, already for axis-parallel
rectangles the problem is still widely open.  In particular, the
techniques of the above approximation schemes for squares do not carry
over to rectangles. The \PTAS from~\cite{ejs-ptasg-05} requires that
every horizontal or vertical line intersects only a bounded number of
objects of the optimal solution that are relatively large in at least
one dimension. For rectangles this number can be up to $\Theta(m)$,
which is too much. For the local search techniques, one can easily
construct examples showing that for any size of exchanges (which still
yields quasi-polynomial running time), the optimum is missed by a
factor of up to $\Omega(m/(\log m)^{O(1)})$.

\begin{figure}[t]%
    \setlength{\doublerulesep}{2pt}%
    \setlength{\arrayrulewidth}{1pt}%
    \doublerulesepcolor{black}%
    \newcommand{\MPFirstColumn}[1]{%
       \begin{minipage}{3.1cm}
           \smallskip%
           #1%
       \end{minipage}
    }%
    \newcommand{\MPSecondColumn}[1]{%
       \begin{minipage}{3.4cm}
           \smallskip%
           #1%
       \end{minipage}
    }%
    \begin{tabular}{|l|l|l|l|c|rrr}
      \HHLine{|-----|}
      \multicolumn{5}{|c|}{%
      \TableTitleC%
      \textbf{Independent set}
      }\\
      \hline%
      Shape
   & %
     Attributes%
   & approximation 
   &%
     Ref
   & running time \\
      \HHLine{|-----|}
      \HHLine{|-----|}
      \multirow{3}{*}{\minitab[c]{
      \\
      Axis-parallel\\
      rectangles}} %
   & %
     Unweighted
   & %
     $\Bigl. O( \log \log m)$ 
   & \cite{cc-misr-09} %
   &
     \multirow{2}{1in}{
     \minitab[c]{$m^{O(1)}$}%
     }\\
      \HHLine{~|---|}
   & %
     Weighted
   & %
     $\Biggl. O\pth{ \frac{ \log m}{ \log \log m } }$ 
   &\cite{ch-aamis-12} %
   &
      \\
      \HHLine{~|----}
   & %
     \newResultC
     \begin{minipage}{3cm}
         \smallskip%
         $\delta$-large weighted

         rectangles with

         vertices in $\Bigl.\IntRange{N}^2$. %
     \end{minipage}
   & %
     \newResultC%
     $\Bigl. 1+\eps$%
   &%
     \newResultC%
     \thmrefpage{delta:large}%
   &%
     \newResultC%
     $(m{} N)^{\constB} $
      \\
      \HHLine{-----}
      %
      \si{Segs}/curves%
   &
     \begin{minipage}{3.4cm}
         \smallskip%
         Unweighted and at most $k =O(1)$
            
         intersection points per pair of curves.$\Bigl.$
     \end{minipage}
   &%
     $m^{\eps}$
   &%
     \cite{fp-cinig-11}
   &%
     $\Bigl. n^{O\pth{ (4/\eps)^{-2/\eps}}}$
        
      %
      \\%
      \HHLine{-----}
      \MPFirstColumn{
      \si{Seg}s / curves 

      \si{rects} / polygons%

      \smallskip
      
      }
   &
     \newResultC%
     Weighted
   &%
     \newResultC%
     $1 + \eps$
   &
     \newResultC%
     \thmrefpage{main}%
   &%
     \newResultC%
     \begin{math}
         \Bigl.  2^{\poly(\log m, 1/\eps)} \cdot n^{O(1)}
     \end{math}
      \\
      %
      %
      \HHLine{=====}
      \multicolumn{5}{|c|}{%
      \TableTitleC%
      \textbf{Sparse properties}}\\
      \hline
      \newResultC%
      \MPFirstColumn{Polygons

      }%
   &
     \newResultC%
     \MPSecondColumn{%
     \smallskip%
     Weighted
     \&  
        pairs
        intersect 
        $O(1)$
        
        times

        \medskip%

        }
   &%
     \newResultC%
     $1+\eps$%
   &%
     \newResultC%
     \thmrefpage{main:2}%
   &%
     \newResultC%
     \begin{math}
         \Bigl.  2^{\poly(\log m, 1/\eps)} \cdot n^{O(1)}
     \end{math}
      \\
      \hline
    \end{tabular}
    \caption{Summary of known and new results. Here
       $\IntRange{N} = \brc{1, \ldots, N}$.}%
    \figlab{results}%
\end{figure}

\subsection{Our results}

We present the first $(1-\eps)$-approximation algorithm to the problem
of computing the maximum weight independent set of polygons, with a
quasi-polynomial running time of
$2^{\mathrm{poly}(\log m, 1/\eps)} \cdot n^{O(1)}$.  In particular,
our algorithm works for axis-parallel rectangles, line segments, and
arbitrary 
polygons.  As mentioned above, the best known polynomial time
approximation algorithm for our setting has a ratio of
$m^\eps$~\cite{fp-cinig-11}, and even for axis-parallel rectangles the
currently best known ratios are $O(\log m/\log\log m)$ for the
weighted case~\cite{ch-aamis-12}, and $O(\log\log m)$ for the
unweighted case~\cite{cc-misr-09}.

We are not aware of any previous algorithms for the problem with
quasi-polynomial running time which would give better bounds than the
above mentioned polynomial time algorithms. Our \QPTAS rules out the
possibility that the problem is \APXHard, assuming that
$\mathsf{NP}\nsubseteq\mathsf{DTIME}(2^{\polylog(n)})$, and thus it
suggests that it should be possible to obtain significantly better
polynomial time approximation algorithms for the problem%
\footnote{%
   Indeed, if a problem is \APXHard, then a \QPTAS for it would imply
   that \ProblemC{SAT} can be solved in $2^{\polylog(n)}$ time, which
   is unlikely. Furthermore, the \emph{strong exponential time
      hypothesis} (\SETH), which is believed to be true, states that
   \ProblemC{SAT} cannot be solved in time better than $2^{cn}$, for
   some absolute constant $c$.  If \SETH is correct, then even
   $\polylog$ sized instances of \APXHard problems cannot be
   $(1+\eps)$-approximated in polynomial time.  }.
   
Then we show how to extend our \QPTAS to computing subsets of polygons
whose intersection graph complies with a given sparsity property.  In
addition, we present a \PTAS for the case of $\delta$-large rectangles
for any constant $\delta > 0$, i.e., for the case when each input
rectangle has at least one edge of length at least $\delta N$,
assuming that in the input only integer coordinates within
$\brc{0,...,N}$ occur.

We give an overview for the previous and the new results for the
problem in \figref{results}.

\subsection{Technical contribution}

\paragraph*{Recursive partitioning.}
The key technique in our \QPTAS is a new geometric partitioning
scheme.  We prove that for the polygons in the optimal solution (and
for any set of non-intersecting polygons) there exists a balanced cut
that intersects only a weighted $O(\eps/\log m)$-fraction of the
polygons and this cut can be described by only
$O(\poly(\log m, 1/\eps))$ bits. Due to the latter property there are
only quasi-polynomially many candidates for this cut, and thus we can
try all of them in quasi-polynomial time. The polygons intersecting
the cut are ``lost'', as the algorithm throw them away. Then the
algorithm calls recursively on both sides of the cut until we obtain
subproblems (described by subparts of the input area) that contains
only a few polygon from the optimal solution, which can be solved
directly by brute force.  Since the cuts are balanced, the recursion
depth is $O(\log m)$, and thus the overall running time of the
algorithm is quasi-polynomial.

\paragraph*{Cheap balanced cuts.}
In order to show that a cut with the claimed property always exists,
we use cuttings \cite{c-chdc-93, bs-ca-95} -- that yields a planar
graph, where each face intersects a relatively small fraction of the
optimal solution.  We then use a separator theorem for planar graphs,
applied to the cuttings, to get a cheap balanced partition of the area
into two pieces. To the best of our knowledge, the idea of using
planar separators together with cuttings is novel, and is one of the
key contributions of this work.  Since the input polygons are
weighted, we need weighted cuttings, and while this is an easy
extension of known techniques, this is not written explicitly in
detail anywhere. As such, for the sake of self-containment, we reprove
here the weighted version of the exponential decay lemma of Chazelle
and Friedman \cite{cf-dvrsi-90}. Our proof seems to be 
simpler than the previous proofs, and the constants are
slightly 
better, and as such the result might be of independent interest.

\paragraph*{Extensions to other sparsity conditions.}
When we ask for an independent set of the input polygons, we require
that the intersection graph corresponding to the set of computed
polygons contains only isolated vertices.  Such a graph is the
ultimate sparse graph. Using the new approach we can also obtain a
\QPTAS when other, more relaxed sparsity conditions are required from
the intersection graph of the computed polygons.  For technical
reasons, here we need to assume that every pair of input polygons
intersects a constant number of times (note that we did not need this
assumption in the independent set case).  We provide a \QPTAS for any
sparsity condition that guarantees that the intersection graph
corresponding to the set of computed objects has a sub-quadratic
number of edges. There are many conditions that fall in this category,
for instance that the intersection graph is planar, or that it does
not contain a $K_{s,t}$ as a subgraph, where $s,t$ are some
constants,. If the input polygons are pseudo-disks then we can even
compute the maximum set such that no point in the plane is covered by
more than $d$ polygons, for any given constant~$d$.

\paragraph*{\PTAS for large rectangles.}
When the input instance is a set of axis-parallel rectangles where
each rectangle has at least one large edge compared to the length of
the corresponding edge of the bounding box, we provide a different
partitioning scheme that leads to a \PTAS. It requires a novel and
rather involved construction, as we cannot use the standard tools to
facilitate it. Our partition has two levels. At the top level, we
partition the plane into a constant number of polygons with constant
complexity each. Some of the polygons of the partition correspond to a
single (potentially large) rectangle, and the others are narrow
corridors with constant complexity each.  This step incurs only a
small loss, and the number of possible partitions is polynomial, so
our algorithm can try all of them.  In the second level, we show how
to decompose each narrow polygon of the partition recursively in a way
that is dynamic programming (DP) friendly.  Then we can use the DP to
find a near-optimal solution in polynomial time.  We believe that this
new partition scheme and the associated dynamic programming algorithm
are the first step in getting a \PTAS for the general problem.

\subsection{Impact of this work}
\seclab{impact}%

This paper contains two new technical concepts: the cheap balanced
(geometric) cuts, used for our \QPTAS{}s, and the partition of the
plane into few large rectangles and narrow corridors, used for the
\PTAS for $\delta$-large rectangles.  Following the initial conference
publication of this work \cite{aw-asmwi-13, aw-qmwis-14}, both
techniques have been used to obtain other results for a variety of
geometric problems.

Using the cheap balanced cuts, Mustafa \etal \cite{mrr-qgscp-14}
showed that one can get a \QPTAS for geometric set cover of points by
weighted pseudodisks. Since the problem becomes \APXHard for fat
triangles of similar size \cite{h-bffne-09}, this is the best one can
hope for. This demonstrates that the geometric set cover and geometric
independent set problems in the plane are inherently different (as far
as approximability).

The partition into a constant number of corridors has been used by
Adamaszek and Wiese~\cite{aw-qptdg-15} as a starting point to get a
\QPTAS for the geometric two-dimensional knapsack problem. They
refined the corridor partition further to a partition into a
poly-logarithmic number of rectangular boxes that separates the
rectangles that are large in the horizontal dimension from those that
are large in the vertical dimension. Moreover, Nadiradze and
Wiese~\cite{nw-aspbr-16} used the corridor partition to obtain a
$(1.4+\eps)$-approximation algorithm in pseudo-polynomial time for the
strip-packing problem. Here, also, the partition into corridors was
used as a starting point to a more refined partition into a constant
number of rectangular boxes.

Bandyapadhyay \etal \cite{bbv-aspcd-16} used cheap balanced cuts for
designing \QPTAS{}s for the convex decomposition problem and the
surface approximation problem.  Marx and Pilipczuk \cite{mp-opapf-15}
used them in order to find faster algorithms for facility location
problems on planar graphs and in the 2-dimensional plane.

The corridor decomposition was used by Har-Peled \cite{h-sppss-16} in
a tool in designing a sublinear space algorithm for shortest path in a
polygon.

\myparagraph{Paper organization} %
In \secref{qptas}, we describe the \QPTAS for the maximum weight
independent set of polygons, where the low-level decomposition tools
needed for the algorithm are described in
\secref{decompose:cut}. Specifically, in \secref{decompose}, we
describe a canonical decomposition of the complement of the union of
disjoint polygons, and we show how to extend it to work for arbitrary
intersecting polygons.  In \secref{cuttings} we reprove the
exponential decay lemma, show how to build weak $1/r$-cuttings of
disjoint polygons of size $O( r \log r)$, and spell out the conditions
enabling one to compute smaller $1/r$-cuttings of size $O(r)$.  In
\secref{extensions}, we describe the extension to a \QPTAS for
computing the maximum weight sparse subset of polygons.  In
\secref{PTAS-large-rectangles}, we describe the \PTAS for large
axis-parallel rectangles.  We conclude in \secref{conclusions} with
some comments.


\section{A \QPTAS for independent set of polygons}
\seclab{qptas}

In this section, we present our $(1-\eps)$-approximation algorithm for
the problem of computing a maximum weight independent set of polygons
with a quasi-polynomial running time.

An instance the problem consists of a set of $m$ weighted simple
polygons $\PolySet = \{\Polygon_1, \ldots, \Polygon_m\}$, with a total
of $n$ vertices.  Let $\eps>0$ be a fixed approximation
parameter. First, we ensure that the weights of the input polygons are
integers in a polynomial range without changing the instance
significantly.  by losing at most a factor of $1-\eps$ in the weight
of an optimal solution $\Opt$, Observe that the following lemma
implies that $w(\Opt)\le m^{2}/\eps$.

\begin{lemma}%
    \lemlab{normalize:weights}%
    If there is a (quasi-)polynomial time $(1+\eps)$-approximation
    algorithm for the case that
    $\weightX{\Polygon}\in\IntRange{m/\eps} = \brc{1,\ldots, m/\eps}$,
    for each polygon $\Polygon\in\PolySet$, then there is a
    (quasi-)polynomial time $(1+\eps)^2$-approximation algorithm for
    the general case.
\end{lemma}

\begin{proof}
    We scale the weights of all polygons such that
    $\alpha = \max_{\Polygon\in\PolySet}\weightX{\Polygon}=m/\eps$.
    For the weight of the optimal solution $w(\Opt)$, we have that
    $w(\Opt) \geq \alpha \geq m/\eps$.  It follows that rounding down
    the weight of each polygon to the closest integer costs at most
    $m$ overall, which is at most $\eps w(\Opt)$.  polygons of weight
    zero after the rounding can be removed.
\end{proof}

From this point on, we assume that
$\weightX{\Polygon}\in\IntRange{m/\eps}$ for each polygon
$\Polygon\in\PolySet$.  The key ingredient of the new algorithm is
that for any independent set of polygons $\PolySet'\subseteq\PolySet$,
and in particular for the optimal solution, there exists a cheap
balanced cut. 

\begin{defn}
    \deflab{cheap:b:cut}%
    Given a set $\PolySet$ of polygons in the plane, a \emphi{cheap
       balanced cut} is a polygon $\CutPolygon$, with the following
    three properties:
    \begin{compactenum}[\quad(A)]
        \item the total weight of polygons in $\PolySet'$ that are
        intersected by $\CutPolygon$ is at most
        $\frac{\eps}{\log m}w(\PolySet')$,
        
        \item the total weight of polygons in $\PolySet'$ that are
        completely inside (resp. outside) of $\CutPolygon$ is at most
        $\frac{2}{3}w(\PolySet')$, 
        
        \item \itemlab{encoding}%
        the polygon $\CutPolygon$ can be fully encoded by a binary
        string with $\poly(\log m,1/\eps)$ bits, and 

        \item the polygon $\CutPolygon$ has $O(n)$ vertices.
    \end{compactenum}
\end{defn}

\begin{lemma}%
    \lemlab{cheap:balanced:cut}%
    For any independent set of polygons $\PolySet'\subseteq\PolySet$
    there exists a cheap balanced cut $\CutPolygon$ or there is a
    polygon $\Polygon\in\PolySet'$ such that
    $w(\Polygon)\ge\frac{2}{3}w(\PolySet')$.
\end{lemma}

The proof of \lemref{cheap:balanced:cut} is in \secref{decompose:cut}
(see \remref{proof:cheap:balanced:cut}).  Our algorithm enumerates all
polygons $\CutPolygon$ that could be cheap balanced cuts corresponding
to the (unknown) optimal solution $\Opt \subseteq \PolySet$. Since the
encoding of such a polygon is short, by \defref{cheap:b:cut}
\itemref{encoding}, there are at most $2^{\poly(\log m,1/\eps)}$ such
polygons, and we can enumerate all of them in quasi-polynomial
time. For each enumerated cut $\CutPolygon$ we call recursively on two
subproblems.  One subproblem consists of all input polygons that lie
completely inside $\CutPolygon$, the other consists of all input
polygons that lie completely outside of $\CutPolygon$. We solve these
subproblems recursively and combine the obtained solutions to a global
solution to the original problem.

A degenerate case here is that the optimal solution $\Opt'$, for the
current subproblem $\PolySet'$, contains a polygon $\Polygon\in\Opt'$,
such that $w(\Polygon)\ge\frac{2}{3}w(\Opt')$.  In the case, the cut
is defined by $\Polygon$ -- one subproblem is $\brc{\Polygon}$, and
the other subproblem is the set of all the polygons in $\PolySet'$
that do not intersect $\Polygon$.

If each step the algorithm correctly guess the cheap balanced cut
$\CutPolygon$, then the recursion has a depth of
$O(\log w(\Opt))=O(\log (m/\eps))$.  Therefore, the algorithm stops
the recursion after $O(\log (m/\eps))$ levels (naturally, the
algorithm also returns immediately if the given subproblem is empty).

\paragraph*{Running time.}
In each node of the recursion tree the algorithm enumerates at most
$2^{\alpha}$ candidates for the cheap balanced cut $\CutPolygon$,
where $\alpha = \poly(\log m,1/\eps)$. Thus, each node has at most
$2^\alpha$ children. Now for each cut $\CutPolygon$, the algorithm
partitions the polygons from the current input instance into three
groups:
\begin{compactenum}[\quad(i)]
    \item polygons intersecting $\CutPolygon$,
    \item polygons contained in the interior
    of $\CutPolygon$, and
    \item polygons  contained in the exterior of $\CutPolygon$.
\end{compactenum}%
\smallskip%
The cut polygon $\CutPolygon$ has $O(n)$ vertices, and this partition
can be computed in $n^{O(1)}$ time -- and in $O(n \log n)$ time if one
is more careful the implementation, see \remref{filtering} below.  The
algorithm then call recursively on the two subproblems defined by the
partition. The recursion depth is $h = O(\log (m/\eps))$. As such, a
recursive subproblem is encoded by
\begin{math}
    \beta = O( h \, \poly( \log m, 1/\eps)) = \poly( \log m, 1/\eps)
\end{math}
bits, and thus there are $2^\beta$ different subproblems overall. The
overall overall running time is
\begin{math}
    2^\beta 2^\alpha n^{O(1)} =2^{\poly(\log m,1/\eps)} O(n \log n).
\end{math}

\paragraph*{Approximation ratio.}

For the correct sequence of cuts, at each level of the recursion the
weight of an optimal solution changes at most by a factor of
$1-\frac{\eps}{\log m}$, since a cheap balanced cut intersects
polygons whose total weight is at most a
$\frac{\eps}{\log m}$-fraction of the optimal solution of the
respective subproblem. Thus, the obtained approximation ratio is at
least $\pth{1-\frac{\eps}{\log m}}^{O(\log (m/\eps))}=1-O(\eps)$.

\medskip

We thus obtain the following.

\begin{theorem}%
    \thmlab{main}%
    Let $\PolySet = \brc{\Polygon_1,\ldots, \Polygon_m}$ be a set of
    $m$ simple polygons in the plane, with $n$ vertices, where
    $\Polygon_i$ has weight $w(\Polygon_i)$, for $i=1,\ldots, m$.
    Then one can compute an independent set
    $\PolySet' \subseteq \PolySet$ of weight at least
    $(1-\eps)w(\Opt)$ in $2^{\poly(\log m,1/\eps)} n \log n$ time,
    where $\Opt$ is the optimal solution.
\end{theorem}

\section{Decompositions and cuttings}
\seclab{decompose:cut}

Our goal in this section is to prove \lemref{cheap:balanced:cut} --
show that for any independent set of polygons there exists a cheap
balanced cut. We show a stronger \emph{constructive} result, by
providing an algorithm that for a given set of non-overlapping
polygons computes a cheap balanced cut efficiently. Of course, in our
settings, the set of these non-overlapping polygons that form the
optimal solution is not known, so only the existence of the cheap
balanced cut is used in the analysis of the algorithm of
\secref{qptas}.

First, in \secref{decompose}, we describe a \emph{canonical
   decomposition} for the plane, guided by the polygons of the optimal
solution. Based on this, in \secref{cuttings} we show that there is an
$1/r$-cutting with small complexity. Finally, in
Section~\secref{good:separation} we prove \lemref{cheap:balanced:cut}.

Later, in \secref{extensions} we extend our reasoning for independent
set to more general settings. Those involve in particular cases in
which the optimal solution consists of polygons which might
overlap. Therefore, we present our reasoning about cuttings also for
the case of polygons that might intersect.

\subsection{Decomposing an arrangement %
   of polygons %
   into corridors}
\seclab{decompose}

\subsubsection{A canonical decomposition for disjoint polygons}
\seclab{canonical-decomposition}

Let $\PolySet = \brc{\Polygon_1, \ldots, \Polygon_m}$ be a set of $m$
non-overlapping simple polygons in the plane, of total complexity $n$.
We also have a special outside square that contains all the polygons
of $\PolySet$, which is the \emphi{frame}.  For the sake of simplicity
of exposition, we assume that all the edges involved in $\PolySet$ and
the frame are neither horizontal nor vertical. This can be ensured by
slightly rotating the axis system\footnote{In the example of
   \figref{example} we do not bother to do this, and the frame is
   axis-parallel.}.

\begin{figure*}[p]
    \centerline{%
       \begin{tabular}{c{c}c}
         \IncludeGraphics[page=1,scale=0.65]{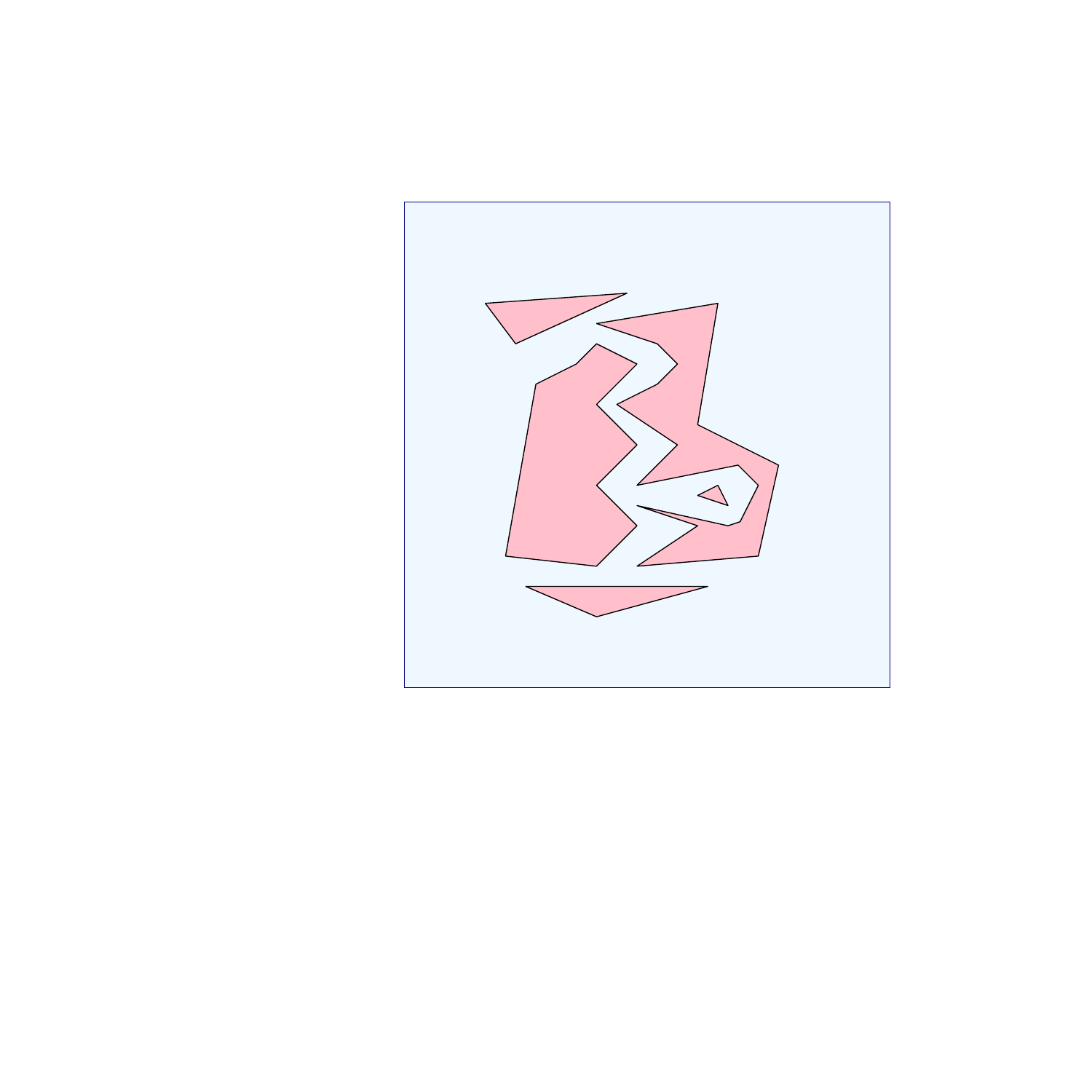}%
         &%
           \qquad\qquad\qquad%
         &
           \IncludeGraphics[page=2,scale=0.65]{figs/medial_axis}\\
         (A) The polygons of $\PolySet$, and the frame.%
         & &
             (B) Some critical squares.\\[0.2cm]%
         \IncludeGraphics[page=3]{figs/medial_axis}
         & &%
             \IncludeGraphics[page=4]{figs/medial_axis}\\
         (C) $\MDC$: The $L_\infty$ medial axis.%
         & &%
             (D) $\MDC'$: The reduced $L_\infty$ medial
             axis.\\[0.2cm]
         \IncludeGraphics[page=7]{figs/medial_axis}%
         & &
             \IncludeGraphics[page=8]{figs/medial_axis}\\
         \begin{minipage}{0.3\linewidth}
             (E) The vertices of degree $3$, their critical squares,
             and the spokes they induce.
         \end{minipage}%
         & &%
             \begin{minipage}{0.3\linewidth}
                 (F) The resulting corridor decomposition, and some
                 corridors.
             \end{minipage}%
       \end{tabular}%
    }
    
    \caption{Building up the corridor decomposition.}
    \figlab{example}
\end{figure*}

We are interested in a canonical decomposition of the complement of
the union of the polygons of $\PolySet$ inside the frame into
\emphi{cells}, that has the property that the number of cells is
$O(m)$, and every cell is defined by a constant number of polygons of
$\PolySet$. To this end, consider the \emphi{medial axis} of
$\PolySet$. To make the presentation easier\footnote{Or at least to
   make the drawing of the figures easier.}, we use the
$L_\infty$-medial axis $\MDC = \MD{\PolySet}$. Specifically, a point
$\pnt \in \Re^2$ is in $\MDC$ if there is an $L_\infty$-ball (i.e., an
axis-parallel square $\Sq$ centered in $\pnt$) that touches the
polygons of $\PolySet$ or the frame in two or more points, and the
interior of $\Sq$ does not intersect any of the polygons of $\PolySet$
or the exterior of the frame. We refer to a square $\Sq$ with the
above properties as a \emphi{critical square}.

The $L_\infty$-medial axis is a connected collection of interior
disjoint segments (i.e., it is the boundary of the Voronoi diagram of
the polygons in $\PolySet$ under the $L_\infty$ metric together with
some extra bridges involving points of the medial axis that have the
same polygon on both sides).  The medial axis $\MDC$ contains some
features that are of no use to us -- specifically, we repeatedly
remove vertices of degree one in $\MDC$ and the segments that support
them -- this process removes unnecessary tendrils. Let $\MDC'$ be the
resulting structure after this cleanup process. Note that this is
exactly the boundary of the Voronoi diagram of the input polygons.

Let $\Vertices = \VerticesX{\MDC'}$ be the set of vertices of $\MDC'$
of degree at least three. Each such vertex corresponds to a point
$\pnt \in \Re^2$ which has a critical square $\Sq_\pnt$ associated
with it. For such a square $\Sq_\pnt$, there are $k \geq 3$ input
polygons (not necessarily pairwise distinct) that it touches, and let
$\pnt_1, \ldots, \pnt_k$ be these $k$ points of contact. We refer to
the segments $\pnt \pnt_1, \pnt\pnt_2, \ldots, \pnt \pnt_k$ as the
\emphi{spokes} of $\pnt$. Since no edge of the input polygons or the
frame is axis parallel, the spokes are uniquely defined.

Let $\SegSet$ be the set of all spokes defined by the vertices of
$\Vertices$. Consider the arrangement formed by the polygons of
$\PolySet$ together with the segments of $\SegSet$. This decomposes
the complement of the union of $\PolySet$ contained inside the frame
into simple polygons. Each such polygon boundary is made out of two
polygonal chains that lie on two polygons of $\PolySet$, and four
spokes, see \figref{example} for an example. We refer to such a
polygon as a \emphi{corridor}, and we denote by $\CD{\PolySet}$ the
collection of corridors corresponding to $\PolySet$.  The set of
corridors $\CD{\PolySet}$ is the \emphi{corridor decomposition} of
$\PolySet$.


\begin{defn}
    \deflab{def:kill:set}%
    Let $\CD{\PolySet}$ be the corridor decomposition of a set of
    non-overlapping simple polygons $\PolySet$ in the plane, and let
    $\PolySetA \subseteq \PolySet$. 
    \begin{compactenum}[\quad(A)]
        \item Consider a corridor $C \in \CD{\PolySetA}$. Then, there
        exists a subset $\PolySetB \subseteq \PolySet$ of size at most
        $4$ such that $C \in \CD{\PolySetB}$. We denote the set
        $\PolySetB$ by $\DefSet{C}$, and call it a \emphi{defining
           set} of the corridor $C$.
        
        \item For a corridor $C \in \CD{\PolySetA}$, a polygon
        $\Polygon \in \PolySet$ \emphi{conflicts} with $C$, if $C$ is
        \emph{not} a corridor of
        $\CD{\DefSet{C} \cup \brc{\Polygon}}$. This happens if
        $\Polygon$ intersects $C$, or alternatively, if the presence
        of $\Polygon$ prevents the creation of the two vertices of the
        medial axis defining $C$. The set of polygons in
        $\PolySet \setminus \DefSet{C}$ that conflict with $C$ is the
        \emphi{stopping set} (or \emphi{conflict list}) of $C$, and is
        denoted by $\KillSet{C}$.
    \end{compactenum}
\end{defn}

Note that the defining set of a corridor might not be unique and that
a defining set might define several corridors and that the stopping
set will be the same, independent of the choice of the defining
set. Note also that if any pair of polygons intersects only $O(1)$
times then any defining set of constant size can define only a
constant number of corridors.

%
%

\begin{lemma}
    \lemlab{corridor:decomp}%
    For a set $\PolySet$ of $m$ disjoint simply connected polygons
    (i.e., polygons without holes) in the plane, we have that
    $\cardin{\CD{\PolySet}} = O\pth{ m }$.
\end{lemma}
\begin{proof}
    Consider the reduced median axis $\MDC'$. It can be naturally
    interpreted as a connected planar graph, where the vertices of
    degree at least three form the vertex set $\Vertices$, and two
    vertices are connected by an edge if there is a path $\pi$ on
    $\MDC'$ that connects them, and there is no vertex of $\Vertices$
    in the interior of $\pi$. Let $\Graph = (\Vertices, \Edges)$ be
    the resulting graph.
    
    Observe that the drawing of $\Graph$ has $m+1$ faces, as each face
    contains a single polygon of $\PolySet$ in its interior (except
    for the outer one, which ``contains'' the frame).  The graph
    $\Graph$ might contain both self loops and parallel edges.
    However, every vertex of $\Graph$ has degree at least $3$. As
    such, we have that $e \geq 3v/2$, where $v$ and $e$ are the number
    of vertices and edges in $\Graph$, respectively.
    
    Euler's formula in this case states that $m+1-e+v =2$ (the formula
    holds even if the graph contains loops and parallel edges). As
    such we have that $m+1 -(3v/2) +v \geq 2$, which implies that
    $2m + 2 \geq v + 4$; that is $v \leq 2m -2$. This in turn implies
    that $m+1 -e + (2m-2) \geq 2$, which implies that $e \leq 3m-3$.
    Now, clearly, every corridor corresponds to one edge of $\Graph$,
    which implies the claim.
\end{proof}

\subsubsection{A canonical decomposition for %
   intersecting polygons}

\newcommand{\IncA}[1]{
   \IncludeGraphics[page=#1,width=0.4\linewidth]{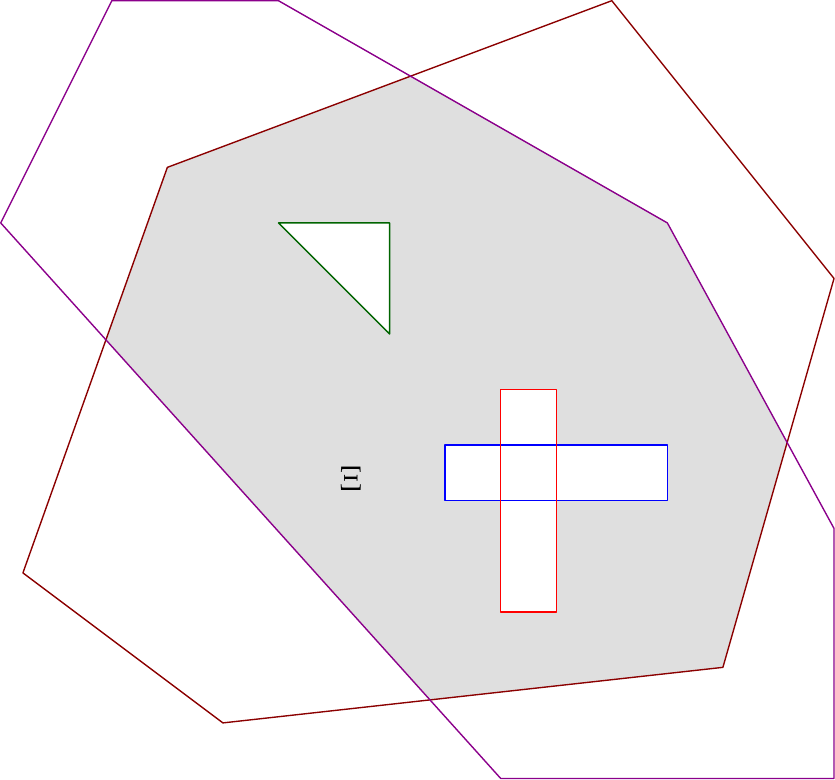}%
}
\begin{figure*}[p]
    \centerline{%
       \begin{tabular}{c{c}c}
         \IncA{1}
         &%
           \qquad\qquad%
         &%
         \IncA{2}%
         \\
         (A) The polygons defining face $\face$ (gray).
         &%
         &
           (B) The edges defining the boundary of $\face$.%
         \\[0.2cm]
         \IncA{3}%
         & &%
             \IncA{4}%
         \\
         (C) The critical squares.%
         &
         &%
           (D) The medial axis vertices generated.%
         \\[0.2cm]
         \IncA{5}%
         &%
         &%
           \IncA{6}%
         \\
         \begin{minipage}{0.4\linewidth}
             (E) The vertices of degree $3$, and their spokes.
         \end{minipage}%
         & &%
             \begin{minipage}{0.4\linewidth}
                 (F) The resulting corridor decomposition, and some
                 corridors.
             \end{minipage}%
       \end{tabular}%
    }
    
    \caption{Building up the corridor decomposition for a single face $\face$,
       for non-disjoint polygons.}
    \figlab{example:2}
\end{figure*}


Let $\PolySetA = \brc{\Polygon_1, \ldots, \Polygon_m}$ be a set of $m$
simple polygons in the plane of total complexity $n$, that are not
necessarily disjoint.  In the following, we think about each polygon
as being a (closed) \emphi{curve}, and we naturally assume that no
three curves pass through a common point.


For two curves of $\PolySetA$, an intersection point of their
boundaries is an \emphi{intersection vertex}.  Consider the
arrangement $\ArrX{\PolySetA}$ -- it is a decomposition of the plane
into \emphi{faces}, i.e., maximal connected components that avoid the
curves of $\PolySetA$.  The maximum connected portion on a curve
between two intersection vertices is an \emphi{edge}. If a curve has
no intersection points on it, then the whole curve is an edge. See
\figref{example:2} (A) and (B) for an example.

Consider a face $\face$ of the arrangement $\ArrX{\PolySetA}$. Denote
by $\totalI_\face$ the number of intersection vertices on the boundary
of $\face$, and by $k_\face$ the number of connected components of the
boundary of $\face$. Each connected component of the boundary can be
broken into several edges. To decompose the face $\face$ into
corridors, we apply a modified version of \lemref{corridor:decomp},
where we treat the connected components of the boundary of $\face$ as
polygons, and the outer connected component of the boundary as the
frame. The main modification is that during the cleanup process, we do
not delete the tendrils that rise out of the intersection vertices
(i.e., the endpoint of a medial axis edge ending at an intersection
vertex is not deleted)%
\footnote{Conceptually, the reader might think about inserting a tiny
   ``puncture'' polygon into $\face$ just next to each such
   intersection vertex, and applying \lemref{corridor:decomp} to this
   set of polygons, where each connected component of the boundary is
   a polygon.}. %
Each such tendril would give rise to one additional corridor.  An
example of the resulting decomposition into corridors is depicted in
\figref{example:2}.

Repeating the above operation for all the faces in the arrangement
$\ArrX{\PolySetA}$ results in a decomposition of the whole plane into
a collection of corridors $\CD{\PolySetA}$. Now, unlike in the setting
of disjoint input polygons, a corridor might be contained in the
interior of several polygons of $\PolySetA$.  However, we still have
the property that no polygon boundary intersects the interior of a
corridor.

\begin{lemma}
    \lemlab{corridor:decomp:ext}%
    Let $\PolySetA$ be a set of $m$ simply connected polygons in the
    plane, let $\totalI$ be the total number of intersection vertices
    in $\ArrX{\PolySetA}$, and let $\CD{\PolySetA}$ be the corridor
    decomposition of $\PolySetA$.  Then
    $|\CD{\PolySetA}| = O\pth{ m + \totalI }$, and each corridor in
    $\CD{\PolySetA}$ is defined by at most $O(1)$ polygons of
    $\PolySetA$.
\end{lemma}

\begin{proof}
    As for the total number of corridors, observe that every
    intersection vertex can contribute to at most four
    faces. Similarly, a single curve such that there is no vertex on
    it, is on the boundary of two faces. Therefore, the total number
    of edges for all faces of the arrangement is
    $O\pth{ m + \totalI }$, and we get
    $|\CD{\PolySetA}| = O\pth{ m + \totalI }$.
    
    We now need to verify that every corridor is defined by a constant
    number of polygons. Indeed, an intersection vertex is defined by
    two polygons, and an edge by three polygons. A medial-axis vertex
    is defined by three edges, which also is the defining set for a
    spoke. As such, all these entities have a constant size defining
    set.
\end{proof}

\subsubsection{Loose and tight corridors}

Let $\PolySet$ be a set of $m$ polygons, where every pair of them
intersects at most $\NInt$ times. Let
$\AllRegions = \bigcup_{\Sample \subseteq \PolySet} \CD{\Sample}$ be
the set of all corridors that are induced by some polygon of
$\PolySet$.

A corridor $\Corridor$ that is defined by a set
$\DefSetC = \DefSet{\Corridor}$ of polygons that are pairwise disjoint
is a \emphi{loose} corridor.  A corridor whose defining set contains a
pair of intersecting polygons is \emphi{tight}.

\begin{lemma}
    \lemlab{loose:tight}%
    Let $\PolySet$ be a set of $m$ polygons, where every pair of them
    intersects at most $\NInt$ times. Let $\LAllRegions$ and
    $\TAllRegions$ be the sets of all loose and tight corridors,
    respectively, induced by any subset of polygons of $\PolySet$. We
    have the following:
    \begin{compactenum}[\quad(A)]
        \item A loose corridor has a defining set of size $\leq 4$,
        and $\cardin{\LAllRegions} = O(m^4)$.

        \item A tight corridor has a defining set of size $\leq 12$,
        and $\cardin{\TAllRegions} = O\pth{m^{12}\NInt^{8}}$.
    \end{compactenum}
\end{lemma}
\begin{proof}
    It is easy to verify that a loose polygon is defined by at most
    four polygons -- indeed, two polygons define the floor and ceiling
    chains, and two additional polygons define the start and end
    vertices. In particular, $\cardin{\LAllRegions} = O(m^4)$.

    A tight corridor is defined by vertices and subcurves of
    $\ArrX{\PolySet}$. As such, to bound the number of tight
    corridors, we first bound the number of such entities in the
    arrangement $\ArrX{\PolySet}$.  A vertex of $\ArrX{\PolySet}$ is
    the intersection point (of the boundary) of two polygons
    $\Polygon, \PolygonA \in \PolySet$.  Since there are $\leq \NInt$
    intersection points between $\bd \Polygon$ and $\bd \PolygonA$, it
    follows that a vertex can be specified uniquely by these two
    polygons and an integer in $\IntRange{\NInt}$. Thus, overall,
    there are $\binom{m}{2} \NInt$ vertices in $\ArrX{\PolySet}$.

    A subcurve of $\ArrX{\PolySet}$ starts at vertex $u$, that is
    induced by two polygons $\Polygon, \PolygonA \in \PolySet$, and
    follows (say) $\PolygonA$, till it reaches a new vertex $v$ that
    is the intersection of $\bd \PolygonA$ with the boundary of some
    polygon  $\PolygonB \in \PolySet$. As such, the number of such
    subcurves in $\ArrX{\PolySet}$ is $\leq m^3 \NInt^2$, and such a
    subcurve is defined by three input polygons.

    Now, as in the loose case, a tight corridor is defined by four entities
    -- in the loose case these were four polygons, while in the tight
    case these are subcurves. It follows that a tight corridor is
    defined by at most twelve  polygons of $\PolySet$, and the number
    of such corridors is at most $\pth{m^3 \NInt^2}^4 = m^{12}
    \NInt^8$.
\end{proof}


    
    

\begin{lemma}
    \lemlab{u:good}%
    Let $\PolySet$ be a set of polygons, such that the boundary of any
    pair of them intersects at most $\NInt$ times, and let
    $\Sample \subseteq \PolySet$ be a  set  of $m$ polygons. Then,
    the number of corridors in $\CD{\Sample}$ (i.e., the complexity of
    $\CD{\Sample}$) is
    \begin{math}
        u(m) = O(m)
    \end{math}
    if the polygons of $\Sample$ are disjoint, and 
    \begin{math}
        u(m) = O( \NInt m^2 )
    \end{math}
    otherwise.
\end{lemma}

\begin{proof}
    The disjoint case is immediate from \lemref{corridor:decomp}. As
    for the more general case, observe that the arrangement
    $\ArrX{\Sample}$ has $\leq \NInt\, \binom{m}{2}$ vertices. The
    bound now follows by applying \lemref{corridor:decomp:ext} for
    each face of the arrangement.
\end{proof}


\subsection{Sampling, exponential decay, %
   and %
   cuttings}
\seclab{cuttings}

We next show how compute \emphi{$1/r$-cuttings} for a collection of
disjoint polygons, and for sparse sets of polygons. We start by
reproving the exponential decay lemma.

\subsubsection{Exponential decay}

Let $\PolySet=\{\Polygon_1,\ldots,\Polygon_m\}$ be a set of $m$
polygons in the plane, where every polygon $\Polygon_i \in \PolySet$
has assigned weight $w_i > 0 $, and let $W = \sum_{i=1}^m w_i$.  We
consider two cases: when $\PolySet$ is independent, and when the
polygons in $\PolySet$ can intersect.

Consider any subset $\Sample \subseteq \PolySet$.  It is easy to
verify that the following two conditions hold. %
\smallskip
\begin{compactenum}[\qquad\qquad(i)]
    \item 
    For any $C \in \CD{\Sample}$, we have
    $\DefSet{C} \subseteq \Sample$ and
    $\Sample \cap \KillSet{C} = \emptyset$, where $\DefSet{C}$ and
    $\KillSet{C}$ are the defining set and the conflict list of $C$,
    respectively.
    
    \item 
    If $\DefSet{C} \subseteq \Sample$ and
    $\KillSet{ C }\cap \Sample=\emptyset$, then $C \in \CD{\Sample}$.
\end{compactenum}
\smallskip

Namely, the corridor decomposition complies with the technique of
Clarkson-Shor \cite{cs-arscg-89} (see also \cite[Chapter
8]{h-gaa-11}).


\begin{defn}
    For a set $\PolySet$ of weighted polygons and a target size
    $\rho > 0$, a \emphi{$\rho$-sample} is a random sample
    $\Sample \subseteq \PolySet$, where each polygon
    $\Polygon_i \in \PolySet$ is in $\Sample$ with probability
    $\displaystyle \rho {w_i} / {W}$.
\end{defn}

\begin{defn}
    A monotone increasing function $\cFunc{\cdot} \geq 0$ is
    \emphi{polynomially growing}, if for any integer $i > 0$ we have
    that $\cFunc{i n } \leq i^{O(1)}\cFunc{ n }$.
\end{defn}

We next prove a weighted version of the exponential decay lemma --
this is a standard implication of the Clarkson-Shor technique.  The
proof is a straightforward extension of the standard proof (if
slightly simpler), and is presented here for the sake of completeness.

\begin{lemma}[Weighted exponential decay lemma]%
    \lemlab{exponential:decay}%
    Let $\PolySet = \brc{\Polygon_1,\ldots, \Polygon_m}$ be a set of
    $m$ disjoint polygons in the plane, with weights $w_1,\ldots,w_m$,
    respectively. Let $\rho \leq m$ and $1 \leq t \leq \rho/4$ be
    parameters, and let $W= \sum_i w_i$.  Consider two independent
    random $\rho$-samples $\Sample_1$ and $\Sample_2$ of $\PolySet$,
    and let $\Sample = \Sample_1 \cup \Sample_2$.  A corridor
    $C \in \CD{\Sample}$ is \emphi{$t$-heavy} if the total weight of
    the polygons of $\PolySet$ in its conflict list $\KillSet{ C }$ is
    at least $t W/\rho$. Let $\CDH{\Sample}{t}$ be the set of all
    $t$-heavy corridors of $\CD{\Sample}$.  We have that
    \begin{math}
        \displaystyle \Ex{\bigl.  \cardin{ \CDH{\Sample}{t} } } =
        O\pth{ \bigl. \rho \exp\pth{-t} }.
    \end{math}
    
    In a more general setting, when the polygons in $\PolySet$ are not
    disjoint, and the corridor decomposition of any $m'$ of them has
    complexity $\cFunc{m'}$, where $\cFunc{m'}$ is a polynomially
    growing function, we have that
    $\Ex{\Bigl. \cardin{ \CDH{\Sample}{t} } } = O\pth{ u(\rho)
       \exp\pth{-t} }$.
\end{lemma}

\begin{proof}
    The basic argument is to use double sampling. Intuitively (but
    outrageously wrongly), a heavy corridor of $\CD{\Sample}$ has
    constant probability to be present in $\CD{\Sample_1}$, but then
    it has exponentially small probability (i.e., $e^{-t}$) of not
    being ``killed'' by a conflicting polygon present in the second
    sample $\Sample_2$.
    
    For a polygon $\PolygonA \in \PolySet$, we have that
    \begin{math}
        \Prob{ \PolygonA \in \Sample_1 \sep{ \PolygonA \in \Sample}}
        =%
        \Prob{ \PolygonA \in \Sample_2 \sep{ \PolygonA \in \Sample}}%
        \geq%
        1/2.
    \end{math}
    Now, consider a corridor $C \in \CD{\Sample}$, and let
    $\PolygonA_1, \ldots, \PolygonA_b \in \PolySet$ be its defining
    set, where $b$ is some small constant. Clearly, we have that
    \begin{align*}
      &
        \Prob{ C \in \CD{\Sample_1} \sepw{ \bigl. C \in
        \CD{\Sample}}}%
        =%
        \Prob{\PolygonA_1, \ldots, \PolygonA_b \in \Sample_1 \sepw{
        \bigl. C \in \CD{\Sample}}}%
      \\%
      &=%
        \prod_{i=1}^b \Prob{\PolygonA_i \in \Sample_1 \sepw{\bigl. C
        \in \CD{\Sample}}}
        = \prod_{i=1}^b \Prob{\PolygonA_i \in {\Sample_1} \sepw{\bigl.
        \PolygonA_1, \ldots, \PolygonA_b \in
        \Sample}}\\
      &%
        = \prod_{i=1}^b \Prob{\PolygonA_i \in {\Sample_1} \sepw{\bigl.
        \PolygonA_i \in \Sample}}%
        \geq%
        \frac{1}{2^b}.
    \end{align*}
    This in turn implies that
    \begin{align}
      2^b \Prob{ \Big. \big(C \in \CD{\Sample_1} \!\big)\, \cap\,
      \big( C \in \CD{\Sample} \!\big) }%
      \geq %
      \Prob{ \Bigl.  C \in \CD{\Sample}}.%
      \eqlab{one}
    \end{align}

    Next, consider a corridor $C \in \CD{\Sample_1}$ that is
    $t$-heavy, with, say,
    $\brc{\Polygon_1, \ldots, \Polygon_k} \subseteq \PolySet$ being
    its conflict list.  Clearly, the probability that $\Sample_2$
    fails to pick one of the conflicting polygons in $\Sample_2$, is
    bounded by
    \begin{align*}
      \ProbC{ C \in \CD{\Sample}}%
      { \Bigl. C \in \CD{\Sample_1}}%
      &=%
        \Prob{ \Bigl.\forall i \in \brc{1,\ldots, k} \quad \Polygon_i
        \notin \Sample_2}%
        =%
        \prod_{i=1}^k \pth{1 - \rho\frac{w_i}{W}}%
      \\%
      &\leq %
        \prod_{i=1}^k \exp \pth{ - \rho \frac{w_i}{ W}} %
        =%
        \exp \pth{ - \frac{\rho}{W} \sum_{i=1}^k w_i } %
      \\
      &\leq%
        \exp\pth{ - \frac{\rho}{W} \cdot t \frac{W}{\rho} } =
        e^{-t}. %
    \end{align*}
    
    Let $\Family$ be the set of possible corridors that can be present
    in the corridor decomposition of any subset of polygons from
    $\PolySet$, and let $\FamilyH{t}\subseteq \Family$ be the set of
    all $t$-heavy corridors from $\Family$.  We have that
    \begin{align*}
      \Ex{\Bigl. \cardin{ \CDH{\Sample}{t} } }%
      &=%
        \sum_{C \in \FamilyH{t}} \Prob{\Bigl. C \in \CD{\Sample}}
        \leq%
        \sum_{C \in \FamilyH{t}} 2^b \Prob{ \Bigl. \pth{C \in
        \CD{\Sample_1}} \cap \pth{ C \in \CD{\Sample}}}\\%
      &\leq%
        2^b \sum_{C \in \FamilyH{t}}%
        \underbrace{ \ProbC{ \Bigl. C \in \CD{\Sample}}{ C \in
        \CD{\Sample_1}} }_{\leq e^{-t}} \Prob{ \Bigl. C \in
        \CD{\Sample_1}}
        %
         \leq%
         2^b e^{-t} \sum_{C \in \FamilyH{t}} \Prob{ \Bigl. C \in
         \CD{\Sample_1}}%
      \\%
      &%
        \leq%
        2^b e^{-t} \sum_{C \in \Family} \Prob{ \Bigl. C \in
        \CD{\Sample_1}}%
         = %
         2^b e^{-t} \Ex{\Bigl.  \cardin{\CD{\Sample_1}}}%
         = %
         2^b e^{-t} \Ex{ \Bigl.  O\pth{ \Bigl. \cardin{\Sample_1}}} = %
         O\pth{ e^{-t} \rho },
    \end{align*}
    by \lemref{corridor:decomp}, and since
    $\Ex{ \cardin{\Sample_1}} = \rho$ and $b$ is a constant.
    
    As for the second claim, by the Chernoff inequality, and since
    $\cFunc{\cdot}$ is polynomially growing, there are constants $c$
    and $c'$, such that
    \begin{align*}
      \Ex{\Bigl.  \cardin{\CD{\Sample_1}}}%
      \leq%
      \cFunc{\rho} + \sum_{i=1}^{\infty} \Prob{ \Big.\! \cardin{
      \Sample_1} \geq i \rho} \cFunc{\Big. (i+1)\rho}%
      \leq%
      \cFunc{\rho} + \sum_{i=1}^{\infty} 2^{-i} c (i+1)^{c'} \cFunc{
      \rho}%
      =%
      O\pth{ \cFunc{\rho} \bigr. }.
    \end{align*}
\end{proof}

Our proof of the exponential decay lemma is inspired by the work of
Sharir \cite{s-cstre-01}. The resulting computations seem somewhat
easier than the standard argumentation.
%

\subsubsection{Cuttings}

For a set $\PolySet$ of weighted polygons of total weight $W$, and a
parameter $r \in \Nat$, a \emphi{$1/r$-cutting} is a decomposition
$\CDC$ of the plane into corridors, such that
\begin{compactenum}[\quad(A)]
    \item the total number of regions in $\CDC$ is small, and
    \item for a $\Corridor \in \CDC$, the total weight of the polygons
    of $\PolySet$ that their boundary intersects the interior of
    $\Corridor$ is at most $W /r$.\footnote{Note that this definition
       does not bound the total weight of the polygons that fully
       contain a region of the cutting. Indeed, this quantity can be
       arbitrarily large.}
\end{compactenum}
See \cite{cf-dvrsi-90, bs-ca-95, h-cctp-00} and references therein for
more information about cuttings.

\begin{lemma}%
    \lemlab{weak:cutting}%
    Let $\PolySet$ be a set of weighted polygons of total weight $W$,
    not necessarily disjoint, such that for any subset
    $\Sample \subseteq \PolySet$, the complexity of $\CD{\Sample}$ is
    $\cFunc{\cardin{\Sample}}$, and $\cFunc{\cdot}$ is a polynomially
    growing function. Then for any parameter $r \in \Nat$ there exists
    a $1/r$-cutting for $\PolySet$ which consists of
    $ O\pth{ \cFunc{ \rho } }$ corridors, where
    $ \rho = O( r \log r ) + r \ln u(2)$. Furthermore, this cutting
    can be computed efficiently.
\end{lemma}

\begin{proof}
    Let $\Sample_1$ and $\Sample_2$ be two independent random
    $\rho$-samples of $\PolySet$ for $\rho = r ( c\ln r + \ln u(2))$,
    where $c$ is a sufficiently large constant, and let
    $\Sample = \Sample_1 \cup \Sample_2$. We claim that the corridor
    decomposition $\CD{\Sample}$ is the desired  cutting.

    
    Since $\cFunc{\cdot}$ is polynomially growing, it must be that
    $\cFunc{i} \leq u(2) i^{O(1)}$. Now, a corridor
    $C \in \CD{\Sample}$ such that the polygons on the conflict list
    of $C$ have the total weight of at least $W/r$ is $t$-heavy for
    $t = c \ln r + \ln u(2)$. By \lemref{exponential:decay}, the
    number of such corridors is in expectation
    \begin{align*}
      \Ex{\Bigl.  \cardin{ \CDH{\Sample}{t} } }%
      =%
      O\pth{ \Bigl. \cFunc{\rho} \exp\pth{-t} }%
      =%
      O\pth{ \Bigl. u(2) \rho^{O(1)} \exp\pth{-t} }%
      =%
      O\pth{ \Bigl. u(2) r^{O(1)} \exp\pth{-t} }%
      <%
      \frac{1}{r^{O(1)}},
    \end{align*}
    for a sufficiently large constant $c$.  By Markov's inequality, we have
    \begin{math}
        \Prob{ \Bigl. \cardin{ \CDH{\Sample}{t} } \geq 1 } \leq \Ex{
           \Bigl.  \cardin{ \CDH{\Sample}{t} } } \leq 1/r ^{O(1)}.
    \end{math}
    Namely, with probability $\geq 1 - 1/r ^{O(1)}$, there are no
    $t$-heavy corridors in $\CD{\Sample}$ -- that is, all the
    corridors of $\CD{\Sample}$ have conflict lists with weights
    $\leq W /r$, as desired.

    The expected of size of the decomposition $\CD{\Sample}$ is
    $O( u(\rho))$, as follows from the argument used in
    \lemref{exponential:decay}.

    Thus, the $\CD{\Sample}$ is a $1/r$-cutting with probability
    $\geq 1-1/r^{O(1)}$, and its (expected) size is $O( u(\rho))$.
\end{proof}

Note that one key property of the above lemma is the bound on the
number of regions. Our lemma above yields a weaker bound on this than
what is known for similar settings in the literature. However, it will
be sufficient for our purposes.  Note that for disjoint polygons we
have that $u(\cdot)$ is linear (see \lemref{corridor:decomp}), and
therefore the cutting has size $O(r \log r)$.

\subsubsection{Smaller cuttings}

Getting $1/r$-cuttings of size $O(u(r))$ (and thus of size $O(r)$ for
disjoint polygons)
is somewhat more challenging. However, for our purposes, any
$1/r$-cutting of size $O(r^c)$, where $c<2$ is a constant, is
sufficient (as provided by \lemref{weak:cutting}). Nevertheless, one
way to get the smaller cuttings, is by restricting the kind of
polygons under consideration.  We do not use the following lemma in
our algorithms but it might be useful for further work.

\begin{lemma}%
    \lemlab{smaller:cutting}%
    Let $\PolySet$ be a set of weighted polygons of total weight $W$,
    not necessarily disjoint, such that for any subset
    $\Sample \subseteq \PolySet$, the complexity of $\CD{\Sample}$ is
    $\cFunc{\cardin{\Sample}}$, and $\cFunc{\cdot}$ is a polynomially
    growing function. In addition, assume that every polygon in
    $\PolySet$ has $O(1)$ intersection points with any line, and the
    boundaries of every pair of polygons of $\PolySet$ have a constant
    number of intersections.  Then for any parameter $r \in \Nat$
    there exists a $1/r$-cutting of $\PolySet$ which consists of
    $ O\pth{ \cFunc{r } }$ regions, where every region is the
    intersection of two corridors. Furthermore, this cutting can be
    computed efficiently.
\end{lemma}

\begin{proof}
    In this case, $u(2) = O(1)$ since a pair of polygons intersect
    only a constant number of times. As such, the result follows by
    the standard two level sampling used in the regular cutting
    construction.  Specifically, we first take a corridor
    decomposition $\CD{\Sample}$ corresponding to a sample $\Sample$
    of size $r$. Then we fix any corridor $C \in \CD{\Sample}$ such
    that the polygons in the conflict list of $C$ are too heavy, by
    doing a second level sampling. We are using \lemref{weak:cutting}
    here. In the resulting decomposition, we have to clip every
    corridor generated in the second level, to its parent
    corridor. The assumption about every polygon intersecting any line
    at most some constant number of times implies the desired bound.
    We omit any further details -- see \si{de} Berg and Schwarzkopf
    \cite{bs-ca-95} and Chazelle and Friedman \cite{cf-dvrsi-90}.~
\end{proof}

\subsection{Structural lemma about a good %
   separating polygon}
\seclab{good:separation}

\begin{lemma}
    \lemlab{good:polygon}%
    Consider a set $\PolySet$ of $m$ weighted polygons of total
    complexity $n$, not necessarily disjoint. Let $\Opt$ be a maximum
    weight independent set of polygons in $\PolySet$, where
    $\nopt := \cardin{\Opt}$ and
    \begin{math}
        \wopt := \weightX{\Opt} = \sum_{\Polygon \in \Opt}
        \weightX{\Polygon}.
    \end{math}
    Let $r$ be a parameter. Then there exists a
    polygon $\CutPolygon$ satisfying the following
    conditions. \medskip
    \begin{compactenum}[\quad (A)]
        \item The total weight of the polygons of $\Opt$ completely
        inside (resp. outside of) $\CutPolygon$ is at least
        $\frac{1}{10} \wopt$.
        
        \item The total weight of the polygons of $\Opt$ that
        intersect the boundary of $\CutPolygon$ is
        $ O\pth{ \sqrt{ \frac{ \log r}{ r} } \wopt }$.
        
        \item The polygon $\CutPolygon$ can be fully encoded by a
        binary string of $O\pth{ \sqrt{r \log r} \log m }$ bits.
    \end{compactenum}

    \newcommand{\IncB}[1]{%
       {\IncludeGraphics[page=#1,width=0.45\linewidth]%
          {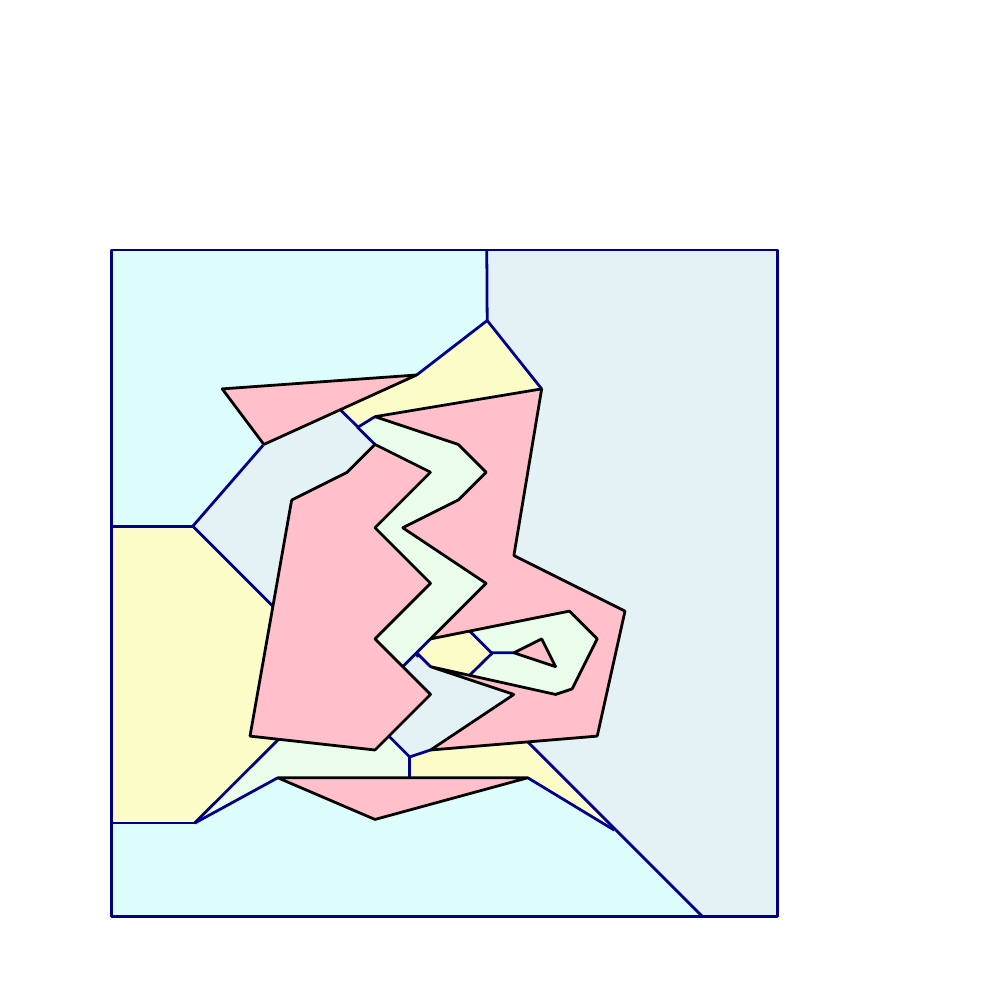}}%
    }%
    \newcommand{\MPB}[1]{%
       \begin{minipage}{0.4\linewidth}
           #1
       \end{minipage}%
    }%
    \begin{figure}
        \IncB{2}
        \hfill%
        \IncB{3}

        \MPB{(A)
           The  corridor
           decomposition
           (see 
           \figref{example}),
           and its dual graph.

        }%
        \hfill%
        \MPB{(B) Fixed and triangulated dual graph.}  \hfill%

        \IncB{4}%
        \hfill%
        \IncB{5}

        \MPB{(C) A separating cycle, and the outer boundary of the
           union of corresponding corridors, which is the separating
           polygon $\CutPolygon$.}%
        \hfill%
        \MPB{(D) The separating polygon $\CutPolygon$ in the original
           corridor decomposition.}

        \caption{Computing a balanced separator from a cutting.}
        \figlab{balanced:cut}%
    \end{figure}
\end{lemma}

\begin{proof}
    Let $\CDC$ be a $1/r$-cutting of $\Opt$, as computed by the
    algorithm of \lemref{weak:cutting}. Here, since $\Opt$ is a set of
    disjoint polygons, the complexity of the corridor decomposition of
    any subset of $\nu$ of them is $u(\nu) = O( \nu)$, by
    \lemref{u:good}.  As such, $\CDC$ is a decomposition of the plane
    into $\rho = O( r \log r)$ corridors (there are also the polygons
    that define the cutting $\CDC$ -- we treat them too as corridors
    that are part of the decomposition $\CDC$).
        
    We interpret $\CDC$ as a planar map with $O(\rho)$ faces, and
    assign every polygon $\Polygon \in \Opt$ to the corridor
    $\Corridor \in \CDC$ which contains the leftmost vertex of
    $\Polygon$.  As such, the \emph{weight} of a corridor
    $\Corridor \in \CDC$ is the total weight of the polygons that have
    been assigned to $\Corridor$. Notice that although a polygon of
    $\Opt$ might intersect several corridors, it is assigned to only
    one of them.
    
    
    Now, consider the dual graph $\DGraph$, where every corridor
    $\Corridor \in \CDC$ corresponds to a vertex in $\DGraph$, and two
    vertices are connected by an edge if the corresponding corridors
    are adjacent, see \figref{balanced:cut} (A).  The dual graph is
    connected, but potentially it might contain self-loops, parallel
    edges, and vertices of degree one. We now fix the dual graph so
    that it becomes triangulated and has none of these ``bad''
    features. To this end, we replace some of the vertices of
    $\DGraph$ by a set of vertices, as follows.
    \smallskip%
    \begin{compactenum}[\;\;(A)]
        \item We replace each vertex $u$ of degree one by two vertices
        $u_1$ and $u_2$, and the only edge $uv$ adjacent to $u$ by two
        edges $u_1v$ and $u_2v$. We also introduce an edge $u_1u_2$
        between the two new vertices. We do a similar reduction if the
        vertex is of degree two.  After this operation all the
        vertices of $\DGraph$ have a degree of at least three.
        
        \item If a vertex $u$ of degree $d$ has self loops or parallel
        edges, we replace it by $d$ new vertices $u_1, \ldots, u_d$
        that are connected in a cycle. We triangulate the inner cycle,
        and redirect the $i$\th edge of $u$ to $u_i$.
    \end{compactenum}%
    \smallskip%
    Finally, we triangulate the resulting graph (i.e., we add edges,
    that are not necessarily straight segments, till every face has
    three edges on its boundary), and let $\DGraph'$ be the resulting
    graph, see \figref{balanced:cut} (B).  Every vertex
    $v \in \VerticesX{\DGraph}$ is associated with a set of vertices
    $D(v)$ in $\DGraph'$. We assign the weight of $v$ arbitrarily to
    one of the vertices of $D(v)$, and all the other vertices of
    $D(v)$ are assigned weight $0$. It is easy to verify that
    $\cardin{\VerticesX{\DGraph'}} = O\pth{\cardin{
          \VerticesX{\DGraph}}}$.
    
    Now, $\DGraph'$ has a cycle separator $\CutPolygon'$ such that the
    total weight of the vertices inside (outside, respectively) the
    cycle $\CutPolygon'$ is at most $(3/4) \wopt$, and $\CutPolygon'$
    has at most $4\sqrt{\cardin{\VerticesX{\DGraph'}}}$ vertices --
    this follows from the cycle separator of Miller \cite{m-fsscs-86}
    (which is weighted).  The resulting cycle visits
    $M = O(\sqrt{\rho})$ vertices of $\DGraph'$, which corresponds to
    a set $\CSet$ of at most $M$ corridors of $\CDC$.  One can now
    track a closed curve $\CutPolygon''$ in the plane, corresponding
    to the cycle $\CutPolygon'$ in the primal, so that the curve stays
    inside the union of the corridors of $\CSet$, and all the vertices
    of $\DGraph'$ inside (resp. outside of) $\CutPolygon'$ correspond
    to the corridors that are strictly inside (resp. outside of) the
    curve $\CutPolygon''$. Now, $\CutPolygon''$ can be transformed to
    a curve $\CutPolygon$ using only the boundary of the corridors of
    $\CSet$ -- the easiest way to do so, is to take $\CutPolygon$ to
    be the outer boundary of the union of all the corridors of
    $\CSet$, see \figref{balanced:cut} (D).  As such, $\CutPolygon$
    consists of $O\pth{ \sqrt{\rho} }$ edges.  Here, an edge is either
    a spoke or a subchain of one of the polygons of $\PolySet$. Now,
    the total weight of polygons of $\PolySet$ that intersect a spoke%
    \footnote{Note, that by the disjointness of the polygons of
       $\Opt$, no such polygonal chain can intersect any of the
       polygons of $\Opt$.} %
    used in the $1/r$-cutting can be at most $\wopt /r$, it follows
    that the total weight of polygons in $\Opt$ intersecting
    $\CutPolygon$ is
    \begin{align*}
      O\pth{ \sqrt{\rho} \frac{\wopt}{r} }%
      =%
      O\pth{ \sqrt{r \log r} \cdot \frac{\wopt}{r} }%
      =%
      O\pth{ \wopt \sqrt{\frac{ \log r}{r}} }.
    \end{align*}
    
    We next show how to encode each edge of $\CutPolygon$ using
    $O( \log m)$ bits, which implies the claim.  Compute the set of
    $O\pth{m^4}$ loose corridors induced by any subset of polygons of
    $\PolySet$ that do not intersect, see \lemref{loose:tight}. Let
    $\XSet$ be this set of polygons. Every corridor of $\XSet$ induces
    $\leq 4$ vertices where its spokes touch its two adjacent
    polygons. In particular, let $\PntSetA$ be the set of all such
    vertices. Clearly, there are $O\pth{m^4}$ such vertices.
    
    Consider an edge $e$ of $\CutPolygon$. If it is a spoke we can
    encode it by specifying which spoke it is, which requires
    $O\pth{ \log m^4}$ bits, since there are $O\pth{m^4}$ possible
    spokes. Otherwise, the edge is a subchain of one of the polygons
    of $\PolySet$. We specify which one of the polygons it is on,
    which requires $O(\log m)$ bits, and then we specify that start
    and end vertices, which belong to $\PntSetA$, which requires
    $O(\log m )$ bits. We also need to specify which one of the two
    possible polygonal subchains we refer to, which requires an extra
    bit. Overall, the number of bits needed to encode $\CutPolygon$ is
    $O\pth{\sqrt{\rho} \log m}$, as claimed.
\end{proof}

\begin{remark}
    If one polygon in the optimal solution is heavier than a
    $\wopt/10$, then the cut can be the polygon itself -- the polygon
    defines the inner subproblem, and all polygons that do not
    intersect it are the outside subproblem. This degenerate case is
    implicit in \lemref{good:polygon}, and is not described
    explicitly, for the sake of simplicity.
\end{remark}

\begin{remark}
    \remlab{proof:cheap:balanced:cut}%
    The proof of \lemref{cheap:balanced:cut} follows from
    \lemref{good:polygon} by choosing
    $r:=\left(\frac{\log m}{\eps}\right)^{2+\mu}$ for any $\mu >0$.
\end{remark}

\begin{remark}
    \remlab{defining:set:good}%
    The separating cycle of \lemref{good:polygon} is defined by a
    random sample of the optimal solution. An interesting property of
    the construction, is that these defining polygons are added as
    their own corridors to the constructed arrangement.  These
    defining polygons are ``islands'' in the constructed arrangement,
    and the corridor decomposition tiles their complement. It thus
    follows that the constructed separating cycle does not intersect
    the interior of the defining polygons. This is crucial, as if
    there are a few heavy polygons (say, of weight $\geq \eps \wopt$),
    then they would be part of the defining set of the cycle, and they
    would get sent down to one of the two recursive subproblems.
\end{remark}

\begin{remark}
    \remlab{filtering}%
    While the separating polygon $\CutPolygon$ has a short encoding,
    it potentially can have a large number of edges -- $O(n)$ in the
    worst case, where $n$ is the total number of vertices in the input
    polygon. The separating polygon $\CutPolygon$ is a simple polygon
    (i.e., no holes or self intersections). Thus, $\CutPolygon$ can be
    preprocessed, in $O( n \log n)$ time, for ray shooting (from the
    inside and outside), where a ray shooting query can be answered in
    $O( \log n)$ time \cite{hs-parss-95}. Now, given another polygon
    $\Polygon$, one can decide if $\Polygon$ is intersecting,inside or
    outside by checking for each edge of $\Polygon$ whether or not it
    it intersects the boundary of $\CutPolygon$, and deciding (say
    using a point-location data-structure) if a vertex of $\Polygon$
    is inside $\CutPolygon$. As such, given a set of polygons with a
    total of $n$ vertices, one can partition them into
    the inside/outside/intersecting sets, in relation to
    $\CutPolygon$, in $O(n \log n)$ time.
\end{remark}

\subsection{Extension: %
   \QPTAS for sparse properties}
\seclab{extensions}

When we compute an independent set of polygons, we output a collection
of polygons with the property that their intersection graph consists
of only isolated vertices. In this section we extend our reasoning to
the setting where the output polygons may overlap, but where we
require that the intersection graph of the output polygons fulfills
some given sparsity condition, i.e., it is planar or it does not
contain a $K_{s,t}$ subgraph for some constants $s,t$.

Let $\PolySet$ be a set of polygons in the plane such that no input
polygon is contained in another input polygon.  We are interested in
the intersection graph $\Graph = (\PolySet, \Edges)$ induced by
$\PolySet$; that is,
$\Edges= \brc{ \Polygon \PolygonA \mid \Polygon, \PolygonA \in
   \PolySet, \Polygon \cap \PolygonA \neq \emptyset }$.  For a subset
$X\subseteq \PolySet$, let $\Graph_X = (X, \Edges_X)$ denote the
\emphi{induced subgraph} of $\Graph$ on $X$; that is,
$\Edges_X = \brc{ \Polygon \PolygonA \sep{ \Polygon, \PolygonA \in X
      \text{ and } \Polygon \PolygonA \in \Edges}}$.  We refer to two
subsets $X \subseteq \PolySet$ and $Y \subseteq \PolySet$ as
\emphi{separate}, if no polygon of $X$ intersects any polygon in $Y$.

Consider a property $\Property$ on graphs (e.g., a graph is
planar). We can naturally define the set system of all subsets of
$\PolySet$ that have this property. That is
$\PropertyF{\PolySet} = (\PolySet,\I)$, where
$\I = \brc{ X \subseteq \PolySet \sepw{ \bigl. \Graph_X \text{ has
         property } \Property}}\Bigl.$.

We are interested here in \emphi{hereditary} properties. Specifically,
if $X \in \PropertyF{\PolySet}$ then $Y \in \PropertyF{\PolySet}$, for
all $Y \subseteq X$.  We also require that the property would be
\emphi{mergeable}; that is, for any two separate subsets
$X, Y \subseteq \PolySet$, such that $X, Y \in \PropertyF{\PolySet}$
we have that $X\cup Y \in \PropertyF{\PolySet}$.  Notice, that the
combinatorial structure $\PropertyF{\PolySet}$ is similar to a
matroid, except that we do not require to have the augmentation
property (this is also known as an independence system).

Here, unlike the independent set case, we assume that the input
polygons are unweighted, see \remref{unweighted} below for more
details.
%
The purpose here is to compute (or approximate) the maximum
cardinality set $X \in \PropertyF{\PolySet}$.

As a concrete example, consider the property $\Property$ that a set
$X \subseteq \PolySet $ has no pair of intersecting polygons. In this
case, finding the maximum cardinality set in $\PropertyF{\PolySet}$
that has the desired property corresponds to finding the maximum
weight independent set in $\PolySet$.

\begin{defn}
    A property $\PropertyF{\PolySet}$ is \emphi{sparse} if there are
    constants $\delta, c > 0$, such that for any
    $X \in \PropertyF{\PolySet}$, we have that
    $\cardin{\EdgesX{\Graph_X}} \leq c \cardin{X}^{2-\delta}$.
\end{defn}

Informally, sparsity implies that in any set
$X \in \PropertyF{\PolySet}$ the number of pairs of intersecting
polygons is strictly subquadratic in the size of $X$.  Surprisingly,
for an intersection graph of curves where every pair of curves
intersects only a constant number of times, sparsity implies that the
number of edges in the intersection graph is linear
\cite{fp-stttr-08}.

\begin{lemma}[\cite{fp-stttr-08}]
    \lemlab{sparse}%
    Let $\PolySet$ be a set of polygons such that the boundaries of
    every pair of polygons have a constant number of intersections,
    and such that no polygon is contained in another polygon.  Let
    $\PropertyF{\PolySet}$ be a sparse property. Then, for any
    $X \in \PropertyF{\PolySet}$, we have that
    $\cardin{\EdgesX{\Graph_X}} = O\pth{\cardin{X}}$.

    If a pair of polygons in $\PolySet$ can have $\NInt$ intersections
    (i.e., of their boundaries), then
    $\cardin{\EdgesX{\Graph_X}} = O\pth{\NInt^{1/2} \cardin{X}}$.
\end{lemma}

\begin{proof}
    This result is known \cite{fp-stttr-08}. We include a sketch of
    the proof here for the sake of completeness.
    
    We think about the boundaries of the polygons of $\PolySet$ as
    curves in the plane, and let $m = \cardin{X}$. The intersection
    graph $\Graph_X$ has a subquadratic number of edges, and as such,
    the arrangement of the curves of $X$ has at most
    $m' = O(m^{2-\delta})$ vertices (there is a vertex for each
    intersection of two curves).  By the planar separator theorem,
    there is a set of $O\pth{\sqrt{m'}} = O(m^{1-\delta/2 })$
    vertices, that their removal disconnects this arrangement into a
    set of $m_1, m_2$ curves, where $m_1, m_2 \leq (2/3)m$ and
    $m_1 + m_2 \leq m + O\pth{\sqrt{m'}}$ (here we break the curves
    passing through a vertex of the separator into two curves, sent to
    the respective subproblems). Applying the argument now to both
    sets recursively, we get that the total number of vertices is
    $T(m) = O( m^{1-\delta/2}) + T(m_1) + T(m_2)$, and the solution of
    this recurrence is $T(m) = O(m)$.

    If there are $\NInt$ intersections between pairs of polygons, then
    the associated arrangement has $O( \NInt n^{1-\delta})$ vertices,
    and the recursion becomes
    $T(m) = ] O( \NInt^{1/2}m^{1-\delta/2}) + T(m_1) + T(m_2)$, and
    the solution is $T(m) = O(\NInt^{1/2} m)$.
\end{proof}

A property $\PropertyF{\PolySet}$ is \emphi{exponential time
   checkable}, if for any subset $X \subseteq \Vertices$, one can
decide if $X \in \PropertyF{\PolySet}$ in time
$2^{\cardin{X}^{O(1)}}$.

\begin{theorem}
    \thmlab{main:2}%
    Let $\PolySet$ be a set of $m$ unweighted polygons in the plane,
    with total of $n$ vertices, such that no input polygon is
    contained in another input polygon, and such that the boundaries
    of every pair of them intersect only a constant number of
    times. Let $\PropertyF{\PolySet}$ be a hereditary, sparse and
    mergeable property that is exponential time checkable.  Then, for
    a parameter $\eps > 0$, one can compute in quasi-polynomial time
    (i.e., $2^{O(\poly( \log m, 1/\eps))}n^{O(1)}$) a subset
    $X \subseteq \PolySet$, such that $X \in \PropertyF{\PolySet}$,
    and $\cardin{X} \geq (1-\eps)\cardin{\Opt}$, where $\Opt$ is the
    largest set in $\PropertyF{\PolySet}$.
\end{theorem}

\begin{proof}
    One need to verify that the algorithm of \secref{qptas} works also
    in this case. As before, we are going to argue that there exists a
    cheap separating polygon.
    
    So, let $\Opt$ be a maximum weight set in $\PropertyF{\PolySet}$,
    and consider any subset $\Sample \subseteq \Opt$.  Since
    $\PropertyF{\PolySet}$ is hereditary, we have that
    $\Sample \in \PropertyF{\PolySet}$.  By \lemref{sparse} the
    arrangement $\ArrX{\Sample}$ has $O\pth{ |\Sample| }$ intersection
    vertices.  By \lemref{corridor:decomp:ext}, the corridor
    decomposition $\CD{\Sample}$ consists of $O\pth{|\Sample|}$
    corridors.

    The existence of a good cycle separator now
    follow from the proof of \lemref{good:polygon} -- here the
    polygons are not necessarily disjoint, but since the corridor
    decomposition in this case still has linear complexity, it still
    works.     

    The resulting separating polygon intersects
    \begin{align*}
      O\pth{ \frac{ \cardin{\Opt} }{r} \sqrt{r \log r}}%
      =%
      O\pth{ \frac{\cardin{\Opt}}{r^{1/3}}}%
      \leq%
      \frac{\eps}{c \log m} \cardin{\Opt}
    \end{align*}
    polygons, for $r = \Omega\pth{ \pth{\log m /\eps}^{3}}$. It is
    easy to verify that such a polygon can be encoded using
    $O( \poly( \log m, \linebreak[0] 1/\eps))$ bits (vertices used by
    the cycle are either vertices rising out of intersection of
    polygons, and there are $O(m^2)$ such vertices, or one of the
    other $O(m^4)$ vertices). The rest of the algorithm now works as
    described in \secref{qptas}. Note that because of the mergeablity
    assumption the algorithm needs to verify that the generated sets
    have the desired property only in the bottom of the
    recursions. But such subsets have size
    $O( \poly( \log m, 1/\eps ))$, and thus they can be checked in
    $2^{O( \poly( \log m, 1/\eps ))}$ time, by the exponential time
    checkability assumption.
\end{proof}

\begin{remark}
    If a pair of polygons is allowed to intersect $\NInt$ times (and
    this number is no longer a constant), then encoding the separating
    cycle now requires $O\pth{ \log \pth{m^{12}\NInt^{8}} }$ bits per
    edge, see \lemref{loose:tight}. The separator now has size
    $O\pth{ \sqrt{\NInt r \log r} }$ by \lemref{sparse}. The resulting
    algorithm would still be a \QPTAS if for example
    $\NInt = O( \polylog m )$.

    There are known separator results for string graphs where the
    number of pairwise intersections is not necessarily a constant
    \cite{m-nossg-14}. Unfortunately, these separators do not
    necessarily form a cycle, which is crucial for the algorithm to
    work.
\end{remark}

Note, that without the assumption that no pair of input polygons is
contained inside each other, we have to deal with the non-trivial
technicality that the separating cycle might be fully contained inside
some input polygon\footnote{It seems that this technicality can be
   handled with some additional care, but the added complexity does
   not seem to be worth it.}.

Properties that comply with our conditions, and thus one can now use
\thmref{main:2} to get a \QPTAS for the largest subset $\Opt$ of
$\PolySet$ that have this property, include the following:
\begin{compactenum}[\qquad(A)]
    \item All the polygons of $\Opt$ are independent.
    
    \item The intersection graph of $\Opt$ is planar, or has low
    genus.
    
    \item The intersection graph of $\Opt$ does not contain $K_{s,t}$
    as a subgraph, for $s$ and $t$ constants.
    
    \item If the boundaries of every pair of polygons of $\PolySet$
    intersects at most twice, then they behave like pseudo-disks. In
    particular, the union complexity of $m$ pseudo-disks is linear,
    and the by the Clarkson-Shor technique, the complexity of the
    arrangement of depth $k$ of $m$ pseudo-disks is $O(km)$. This
    implies that if $\Opt$ is a set pseudo-disks with bounded depth,
    then the intersection graph has only $O( \cardin{\Opt})$ edges,
    and as such this is a sparse property, and it follows that one can
    $(1-\eps)$-approximate (in quasi-polynomial time) the heaviest
    subset of pseudo-disks where the maximum depth is
    bounded. Previously, only a constant approximation was known
    \cite{ehr-gpnuc-12}.
\end{compactenum}

\begin{remark}
    \remlab{unweighted}%
    The algorithm of \thmref{main:2} does not work for weighted
    polygons\footnote{In particular, the conference version of this
       paper \cite{h-qssp-14} incorrectly claimed that the algorithm
       works in this case.}. The main reason is that the separating
    cycle, when applied to the optimal solution, is defined by a
    collection of polygons. In the independent set case, these
    defining polygons were not lost (see \remref{defining:set:good}),
    but this can no longer be guaranteed. and unfortunately these
    polygons can be ``heavy'' and form a crucial subset of the optimal
    solution. A naive way to try and address this issue is to allow an
    additional set of ``special'' polygons sent down to the recursive
    subproblems as being present in the optimal solution. This
    requires a modification of the mergeablity property, which works
    in some cases, but fails for others (for example, if the
    intersection graphs of $X \cup Y, Y \cup Z \subseteq \PolySet$ are
    planar, this is not necessarily true for $X \cup Y \cup Z$). We
    leave the extension of the algorithm of \thmref{main:2} as an open
    problem for further research.

    Note, that the above is not an issue for the unweighted case --
    all the polygons intersecting the separating cycle can be thrown
    away -- the number of additional polygons in the set defining the
    separating cycle is small compared to the optimal solution, and
    has not impact on the approximation quality.
\end{remark}


\section{A \PTAS for $\delta$-large rectangles}
\seclab{PTAS-large-rectangles}

In this section we present a polynomial time approximation scheme for
the maximum weight independent set of $\delta$-large rectangles, i.e.,
for axis-parallel rectangles that have at least one edge which is long
with respect to the corresponding edge of the bounding box.
In order to achieve polynomial running time, we embed the recursion
into a dynamic program and show that the number of subproblems to be
considered is polynomially bounded (unlike the \QPTAS case where the
number of subproblems is larger). We first present our algorithm for
blocks, i.e., for large rectangles whose height or width is $1$, and
then we extend it to arbitrary $\delta$-large rectangles.

\subsection{Formal definition of the problem}

Let us fix constants $\delta > 0$ and $\eps > 0$.  Let
$\RSet = \brc{\rect_{1},...,\rect_{m}}$ be a set of $m$ axis-parallel
rectangles with integer coordinates in the plane, where the $i$\th
rectangle $\rect_{i}\in\RSet$ is defined as an open set
\begin{math}
    \rect_i = \OIntY{\rxL{i}}{\rxR{i}} \times \OIntY{\ryB{i}}{
       \ryT{i}},
\end{math}
where $\rxL{i} < \rxR{i}$, $\ryB{i} < \ryT{i}$.  Let $N$ be the
smallest integer s.t. the vertices of all rectangles in $\RSet$ are in
$\IntRange{N}^2$.  For each rectangle $\rect_{i}$, its \emphi{width}
is $\rW{i} = \rxR{i} - \rxL{i}$ and its \emphi{height} is
$\rH{i} = \ryT{i} - \ryB{i}$.  We denote the quantity
$\GWidth := \delta N$ as the \emphi{largeness threshold}.  A rectangle
$\rect_{i} \in \RSet$ is \emphi{$\delta$-large} (or just
\emphi{large}) if $\rH{i} = \ryT{i} - \ryB{i} > \GWidth$ or
$\rW{i} = \rxR{i} - \rxL{i} > \GWidth$.

\begin{problem}[Independent set of large rectangles]
    \problab{i:set:rects}%
    The input consists of a set $\RSet$ of weighted $\delta$-large
    rectangles, where the weight of a rectangle $\rect_i \in \RSet$ is
    a positive number $w_i$.  The task is to compute a maximum weight
    subset $\RSetA \subseteq \RSet$, such that the rectangles of
    $\RSetA$ are disjoint.
\end{problem}

\begin{defn}
    A rectangle $\rect_i \in \RSet$ is a \emphi{block} if (i)
    $\rect_i$ is $\delta$-large, and (ii) either $\rH{i}=1$ or
    $\rW{i}=1$.
\end{defn}

\begin{problem}[Independent set of blocks]
    \problab{i:set:blocks}%
    The input consists of a set
    $\BSet=\{\Block_1,\Block_2,...,\Block_m\}$ of weighted blocks.
    The task is to compute a maximum weight subset
    $\BSetA \subseteq \BSet$ such that the rectangles of $\BSetA$ are
    disjoint.
\end{problem}


To simplify the description, we assume the following:
\begin{inparaenum}[(i)]
    \item $m/\eps$ and $\eps m$ are both integers,
    \item $1/\delta\in\mathbb{N}$, and
    \item $\GWidth = \delta N \in\mathbb{N}$.
\end{inparaenum}
For any two points $\pnt,\pnt'$, let $\segCY{\pnt}{\pnt'}$ denote the
closed straight segment from $\pnt$ to $\pnt'$. Similarly, let
$\segOY{\pnt}{\pnt'}$ be the open segment
$\segCY{\pnt}{\pnt'} \setminus \brc{\pnt,\pnt'}$.

\subsection{Constructing the partition for blocks}
\seclab{partition:blocks}

\subsubsection{Overview and definitions}

As before, we assume that the maximum weight subset
$\Opt \subseteq \BSet$ of disjoint blocks is known to us, and we prove
that there is a ``cheap'' partition that enables one to compute a
near-optimal independent set of blocks using dynamic programming.

Specifically, we construct a partition of the bounding box $[0,N]^2$
using a set of at most $\constAmath$ rectilinear (i.e., horizontal and
vertical) line segments with integer endpoint coordinates, such that
the blocks of $\Opt$ intersected by these segments have a small total
weight compared with $w(\Opt)$. Furthermore, each face of the
partition is either a simple rectilinear polygon, or a rectilinear
polygon with a single rectilinear hole. In either case, the polygon
would have ``width'' of at most $\GWidth$, which is strictly smaller
than the length of any block.  In a sense, this partition sparsely
describes the topology of the (large) blocks while intersecting blocks
of negligible total weight.

From this point on, a \emphi{segment} refers to a horizontal or
vertical (closed) line segment that has integer coordinates with
endpoints in $\IntRange{N}^2$.

First, we construct a grid $\TGrid$ of large tiles, consisting of
\begin{math}
    {1}/{\delta}\times {1}/{\delta}
\end{math}
uniform grid cells in the input square $[0,N]^2$, i.e., for each
$i,j\in\brc{0,...,1/\delta-1}$ there is a grid cell with coordinates
\begin{math}
    \pbrc{ i\GWidth,(i+1)\GWidth \bigl. } \times \pbrc{j \GWidth
       ,(j+1)\GWidth}.
\end{math}

For a set $X \subseteq \Re^2$, we use $\clX{X}$ to denote the
\emphi{closure} of $X$.  We remind the reader that blocks (and
rectangles) are open sets, and therefore for a block
$\Block \in \BSet$, $\clX{\Block}$ is the closed version of $\Block$.

\begin{defn}%
    \deflab{s:cuts:r}%
    A segment (or a line) $\seg$ \emphi{cuts} a rectangle $\rect$ if
    $\rect \setminus \seg$ has two connected components.  A segment
    $\seg$ \emphi{hits} a rectangle $\rect$ if (i) $\seg$ intersects
    $\clX{\rect}$, (ii) $\seg$ does not intersect $\rect$, and (iii)
    $\LineSpanX{\seg}$ cuts $\rect$, where $\LineSpanX{\seg}$ denotes
    the line that spans $\seg$.
    Similarly, a segment $\seg$ \emphi{hits} a segment $\segA$, if
    \begin{inparaenum}[(i)]
        \item $\seg$ and $\segA$ are orthogonal to each other, and
        \item $\seg$ has an endpoint in the interior of $\segA$.
    \end{inparaenum}
    Two segments that intersect in their interior are
    \emphi{crossing}.
\end{defn}

Note that if a segment $\seg$ hits a segment $\segA$, then $\segA$
does not hit $\seg$.  By assumption, any block of $\BSet$ intersects
at least two grid cells.  A block $\Block \in \BSet$ \emphi{ends} in a
grid cell $\GCell$, if (i) $\GCell$ and $\Block$ intersect, and (ii)
$\GCell$ contains one of the corners of $\Block$.

Our construction has two steps. In the first step we construct an
initial set of segments $\LX$ that contains the boundary of the input
square and adds $O(1)$ segments per each grid cell. The segments of
$\LX$ do not intersect any blocks from the optimal solution $\Opt$ and
are pairwise non-crossing, but they might have \emph{loose ends}, that
is, endpoints of segments from $\LX$ that are not contained in the
interior of some other segment of $\LX$. Therefore, in our second
step, we add a set of segments $\LY$ that connect these loose ends
with other lines in $\LX \cup \LY$. Segments in $\LY$ might intersect
blocks from $\Opt$, however, the total weight of intersected blocks is
bounded by $\eps\cdot w(\Opt)$, and thus we can afford to lose them.

\subsubsection{Construction step I: The set $\LX$}
\seclab{Z:construction}

Initially, $\LX$ is a set containing the four segments forming the
boundary of the input square $[0,N] \times [0,N]$. Next, consider a
grid cell $\GCell$ and its (closed) bottom edge $\edge$.  A vertical
segment $\seg$ is \emphi{admissible} if
\begin{compactenum}[\quad(i)]
    \item $\seg$ intersects $\edge$ (i.e., either $\interX{\seg}$
    intersects $\edge$, or $\seg$ has an endpoint on $\edge$) and that
    $\interX{\seg}$ intersects $\GCell$,
    \item $\seg$ does not intersect any block of $\BSet$,
    \item $\seg$ does not cross any segment of $\LX$,
    \item $\lenX{\seg} > \GWidth$ (i.e., $\seg$ is long), and
    \item $\seg$ is a maximal (i.e., as long as possible) segment
    satisfying the properties above.
\end{compactenum}
We add to $\LX$ the following segments.
\begin{compactenum}[\quad(A)]
    \item \itemlab{l:y:A}%
    An admissible segment $\seg$ with the smallest $x$-coordinate
    (i.e., $\seg$ might lie on the left edge of $\GCell$)
    \item An admissible segment with the largest $x$-coordinate.
    \item \itemlab{l:y:C}%
    An admissible segment $\seg$ which maximizes the length of
    $\lenX{\seg \cap \GCell}$.  If there are several such segments, we
    add two of them: one with the smallest and one with the largest
    $x$-coordinate, respectively. Segments maximizing
    $\lenX{\seg\cap \GCell}$ are called \emphi{\reach} segments for
    $\edge$ in $\GCell$. The ones added to $\LX$ in this step are
    called \emphi{extremal} \reach segments.
\end{compactenum}

\begin{figure}
    \centerline{\IncludeGraphics{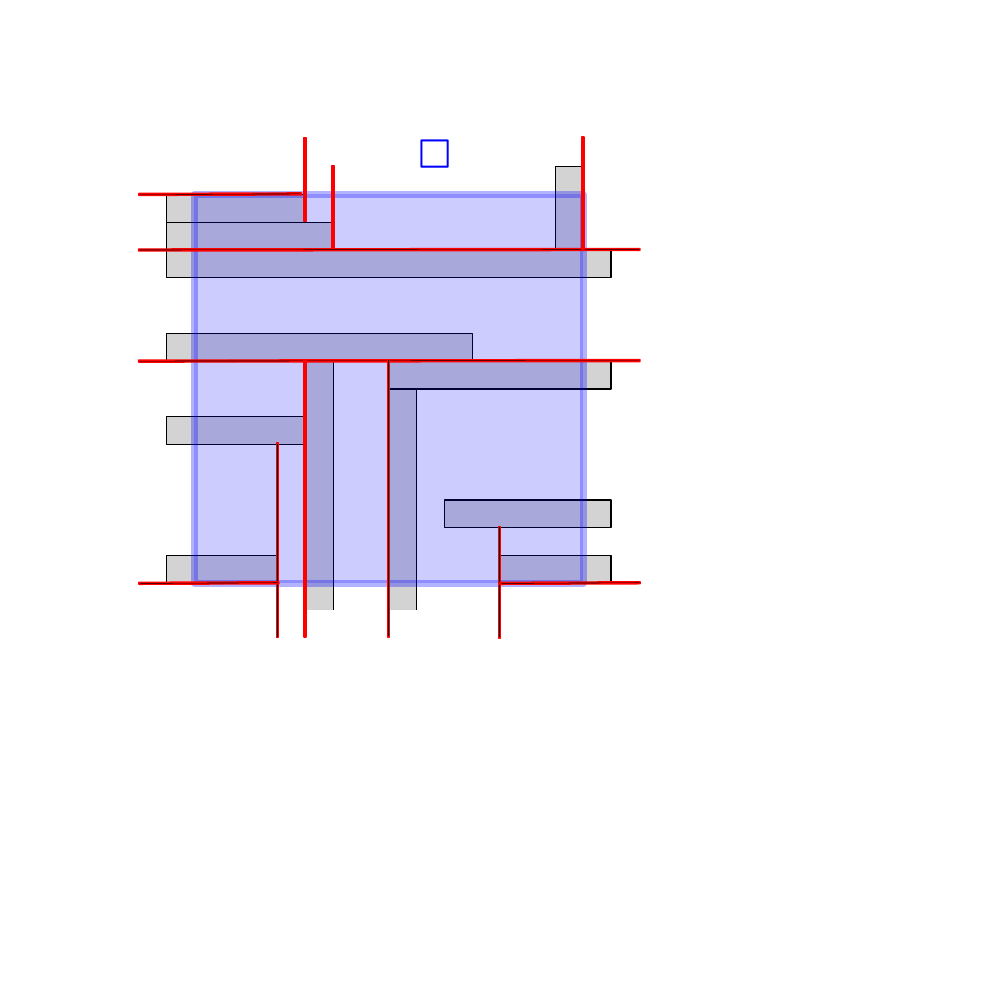}}
    \caption{The red lines denote the segments added to $\LX$ while
       processing the grid cell $\GCell$. The blocks of $\Opt$
       intersecting $\GCell$ are depicted in gray.}
    \figlab{construction-segments-L}
\end{figure}

\smallskip
\noindent%
The algorithm performs the same operation for the top, left and right
edges of $\GCell$, where for the left and right edges it considers
horizontal segments instead of vertical. This is done in a fixed
order, e.g., first all vertical segments, and then all horizontal
segments. See \figref{construction-segments-L} for an example.

By construction, the resulting segments of $\LX$ (excluding the four
frame segments) are all interior disjoint and maximal (i.e., they
cannot be extended without crossing other segments from $\LX$ or
intersecting blocks of $\Opt$).

\subsubsection{Construction step II: The set $\LY$}
\seclab{app:construction}

\paragraph*{Idea.}

The segments in $\LX$ might have loose ends as mentioned above. We
need to connect such endpoints up so that the resulting set of
segments partitions the input square into faces, where every face is
either a simple polygon, or a polygon with a single hole.  The idea is
to perform a walk in the square, looking for a way to connect such a
loose end with the segments already constructed. If the walk is too
long, it would be shortened by introducing a cheap shortcut segment.

\paragraph*{Setup.}

Initially, the set $\LY$ is empty.  For each endpoint $\pnt_0$ of a
segment $\seg \in \LX$, such that $\seg$ does not hit a perpendicular
segment in $\LX \cup \LY$ at $\pnt_0$, we create a path of segments
connecting $\seg$ with a segment in $\LX \cup \LY$, adding the new
segments of the path to $\LY$.

A segment $s$ is \emphi{maximal} if it does not cross any segment of
$\LX \cup \LY$ or intersect any of the blocks from $\Opt$, but any
segment $s'$ such that $s \subset s'$ violates this property. Note
that any endpoint of a maximal segment must lie in the interior of an
edge of a perpendicular block of $\Opt$, or in the interior of a
perpendicular segment of $\LX \cup \LY$. In the sequel, we use the
following technical lemma.

\bigskip%
\noindent%
\begin{minipage}{0.8\linewidth}
    \begin{lemma}
        \lemlab{line-hits-block}%
        Let $\GCell$ be a grid cell and $\pnt$ be a point in
        $\GCell$. Let $\seg$ be a maximal segment with one endpoint at
        $\pnt$, and the other endpoint outside of $\GCell$. If $\seg$
        hits a perpendicular block $\Block \in \Opt$ at $\pnt$, but it
        does not hit a perpendicular segment from $\LX$ at $\pnt$,
        then:
        \begin{compactenum}[\quad(i)]
            \item $\pnt \in \interX{\GCell}$, and
            \item one end of $\Block$ is in $\GCell$, and the other
            one is outside of $\GCell$.
        \end{compactenum}
    \end{lemma}
\end{minipage}%
\begin{minipage}{0.2\linewidth}
    \hfill%
    \IncludeGraphics[width=0.8\linewidth]{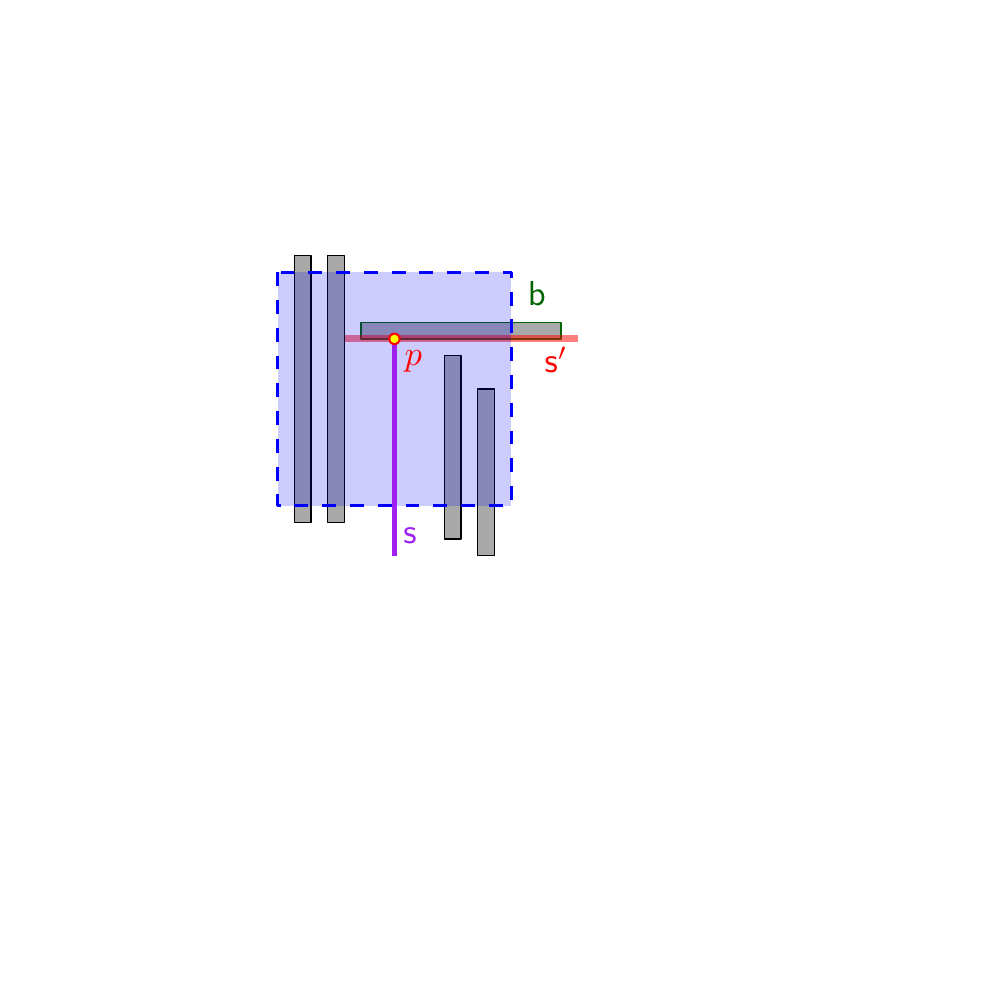}
\end{minipage}%

\begin{proof}
    Assume without loss of generality that $\seg$ is vertical, $\pnt$
    is at the top end of $\seg$, and $\Block$ crosses the boundary of
    the grid cell to the right of $\pnt$ (as, by assumption, $\pnt$
    lies on the long edge of $\Block$). Let $\seg'$ be the maximal
    segment which contains the bottom edge of $\Block$ and does not
    intersect any blocks or segments from $\LX$. As $\pnt \in \seg'$,
    by assumption, we have that $\seg' \notin \LX$.
    
    Assume that $\pnt$ is on the boundary of $\GCell$. If $\pnt$ lies
    on the bottom edge of $\GCell$, then $\seg'$ is the bottom-most
    long segment crossing the right edge of $\GCell$.  But then
    $\seg' \in\LX$, which is a contradiction. The cases that $\pnt$ is
    on the top, left or right edges of $\GCell$ are handled in a
    similar fashion.
    
    As such, $\pnt \in \interX{\GCell}$ and $\Block$ intersects the
    interior of $\GCell$.

    If $\Block$ does not have an end in $\GCell$, then $\seg'$ cuts
    $\GCell$. If $\seg'$ is the bottom-most \reach segment for the
    left edge of $\GCell$ then $\seg' \in\LX$, which gives a
    contradiction. Otherwise, the bottom-most \reach segment for the
    left edge of $\GCell$ is below $\seg'$ and cuts $\GCell$, so it
    intersects $\seg$, and again we get a contradiction, as $\seg$
    does not intersect edges from $\LX$. Block $\Block$ must end in
    $\GCell$.
\end{proof}

\paragraph*{Building the path for a single loose end.}%

We consider all loose endpoints of segments in $\LX$, one by one. Let
$\pnt_0$ be such a loose endpoint of a segment $\seg_0 \in \LX$, such
that $\seg_0$ does not hit a perpendicular segment in $\LX \cup \LY$
at $\pnt_0$, and set $i=1$.  In the following, let
\begin{align}
  M = 64/\pth{\eps \delta^{2} }. %
  \eqlab{M:value}
\end{align}
Next, we construct a path $\seg_1, \ldots , \seg_M$ starting from
$\pnt_0$, aiming to connect $\pnt_0$ with some existing segment in
$\LX \cup \LY$.


\begin{figure}[h]
    \begin{tabular}{ccc}
      \IncludeGraphics[page=1,scale=0.95]{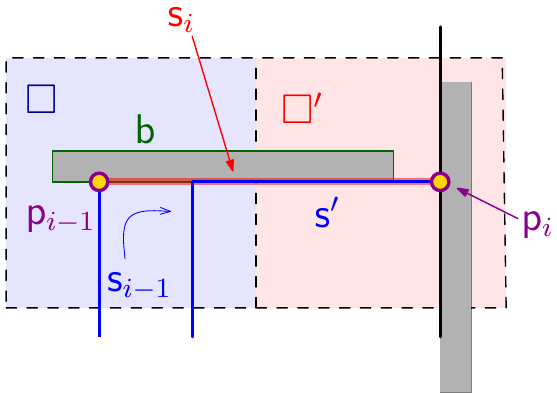}%
      &
        \quad%
        \IncludeGraphics[page=2,scale=0.95]{figs/16_new}%
      &%
                                                                                      \quad %
        \IncludeGraphics[page=3,scale=0.95]{figs/16_new}%
      \\%
      (A) & (B) & (C)
    \end{tabular}
    \caption{%
       Construction of the path. The segment $\seg_{i-1}$ hits a
       block $\Block$ at the point $\pnt_{i-1}$. The segment $\seg_i$
       starts at $\pnt_{i-1}$ and follows the edge of $\Block$
       until it hits a perpendicular segment which is already in
       $\LX \cup \LY$, or a perpendicular block. In case (A) the new
       segment $\seg_i$ overlaps an existing segment
       $\seg' \in \LY\cup\{\seg_1,...,\seg_{i-1}\}$ and $\seg'$ is replaced by $\seg_i$. 
       In case (B), $\seg_i$ is added to $\LY$. 
       In case (C), the walk continues from the new endpoint
       $\pnt_i$.}
    \figlab{creating-le-cases}%
\end{figure}

\NotHandled{%
   \andy{\figref{creating-le-cases}: (A) The caption says that
      $\seg_{i}$ is not added to $\LY$ but in the main text we say
      that $\seg'$ is replaced by $\seg_{i}$ (B) Where is $\pnt_0$?
      Could it be that we mean $\pnt_{i-1}$ here? The picture (B)
      looks as if the left boundary of the right block is already in
      $\LY$, is that the case?  If yes, we should state that in the
      caption and in case (C) explain why this case is different.}%
   \sariel{I believe what is there is correct.}%
}%

Let $\GCell$ be a grid cell such that $\pnt_{i-1} \in \GCell$ and
$\seg_{i-1}$ does not hit a segment from $\LX \cup \LY$ at
$\pnt_{i-1}$.  Let $\Block$ be the block hit by $\seg_{i-1}$ at
$\pnt_{i-1}$. As $\seg_{i-1}$ cannot be extended beyond $\pnt_{i-1}$,
such a block $\Block$ exists.
Applying \lemref{line-hits-block} to the point $\pnt_{i-1}$ together
with the segment $\seg_{i-1}$ (we can do that as $\seg_{i-1}$
intersects at least two cells), we obtain that
$\pnt_{i-1} \in \interX{\GCell}$, and $\Block$ has one end in $\GCell$
and the other end in some other grid cell $\GCell'$. Let $\segMax$ be
a maximal segment which contains the edge of $\Block$ containing
$\pnt_{i-1}$.
Let $\pnt_i$ be the endpoint of $\segMax$ such that
\begin{math}
    \segCY{\pnt_{i-1}}{\pnt_i} \cap \GCellA \neq \emptyset.
\end{math}
Set $\seg_i = \segCY{\pnt_{i-1}}{\pnt_i}$. Now,
$\pnt_i \notin \GCell$, so $\seg_i$ intersects at least two grid
cells.

We continue this walk, for $i=1,\ldots, M$ and consider the following
cases (see \figref{creating-le-cases}).
\begin{compactenum}[(A)]
    \item \itemlab{A} 
    There is a segment $\seg' \in \LY\cup\{\seg_1,...,\seg_{i-1} \}$,
    such that $\interX{\seg_i} \cap \interX{\seg'} \neq \emptyset$,
    see \figref{creating-le-cases} (A). This can happen, as the
    segments of $\LY$ are not necessarily maximal. As $\seg_{i-1}$
    does not hit $\seg'$ at $\pnt_{i-1}$, and $\seg_i$ cannot be
    extended beyond $\pnt_i$, we must have $\seg' \subset \seg_i$.  If
    $\seg'\in \LY$ then we replace $\seg'$ by $\seg_i$ in $\LY$, add
    the segments of the path constructed so far to $\LY$, and stop the
    path construction. If $\seg' \in \{\seg_1,...,\seg_{i-1} \}$ then
    we also stop the path construction and add the segments
    $\{\seg_1,...,\seg_{i} \}\setminus \{\seg'\}$ to $\LY$.
    \item \itemlab{B} %
    Case \itemref{A} does not happen, but $\seg_i$ hits a segment from
    $\LX \cup \LY \cup\{\seg_1,...,\seg_{i-1} \}$ at $\pnt_i$, see
    \figref{creating-le-cases} (B).  We add the segments
    $\{\seg_1,...,\seg_{i} \}$ to $\LY$, and the construction of the
    path is done.
    
    \item Cases \itemref{A} and \itemref{B} do not happen. In this
    case, $\seg_i$ hits some perpendicular block at $\pnt_i$ (see
    \figref{creating-le-cases} (C)). The algorithm proceeds as before,
    considering the segment $\seg_i$ and its endpoint $\pnt_i$ instead
    of $\seg_{i-1}$ and $\pnt_{i-1}$. The conditions of
    \lemref{line-hits-block} are satisfied, as $\seg_i$ intersects at
    least two grid cells. We continue extending the path until one of
    the cases \itemref{A} or \itemref{B} happens, or until the number
    of segments in the path reaches the upper bound of $M$.
\end{compactenum}


\begin{figure}
    \centerline{%
       \IncludeGraphics{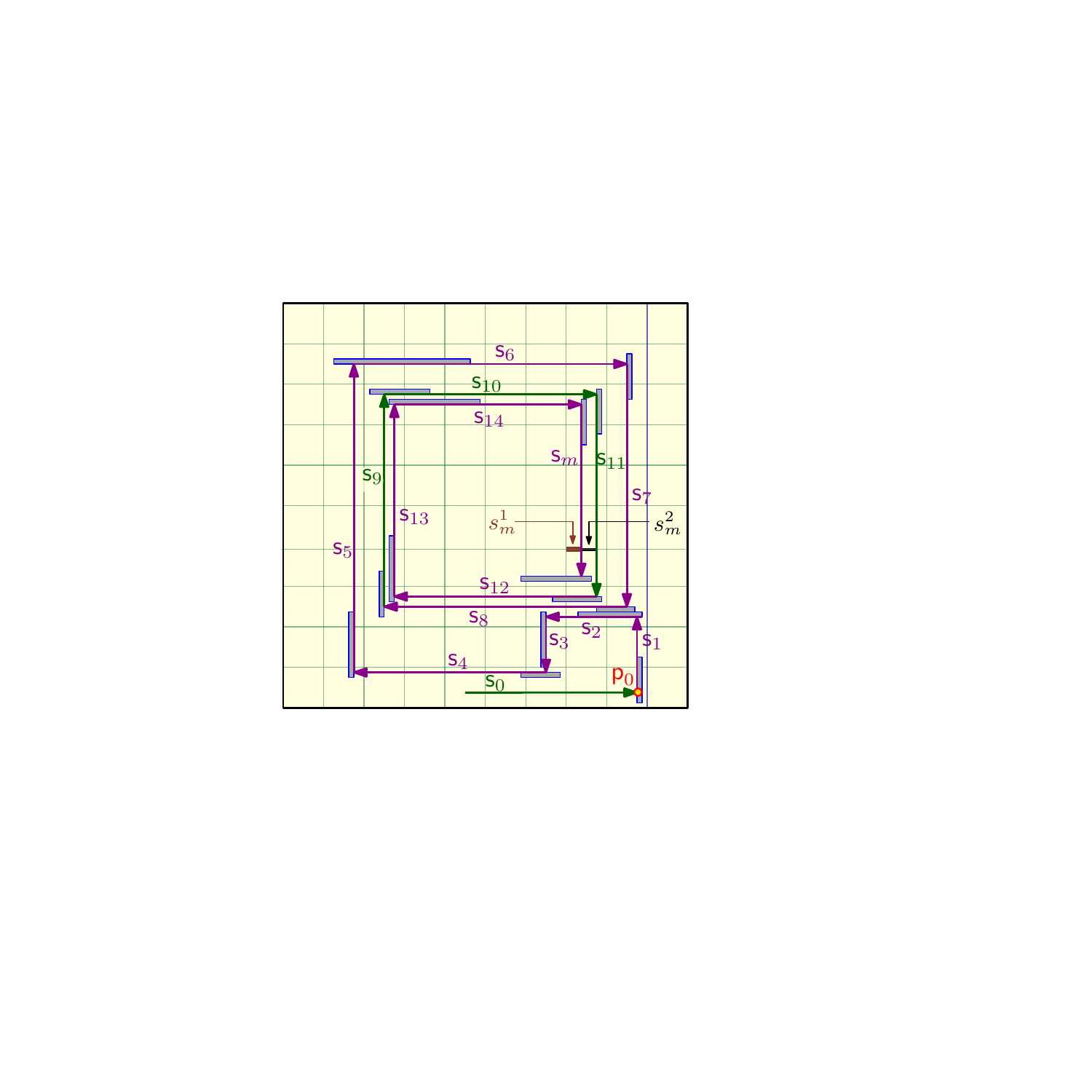}%
    }
    
    \caption{The construction of the segments $\LY$. The blocks
       of the considered instance are depicted in gray.}%
    \figlab{construction-segments-segp-case3}
\end{figure}

\paragraph{Shortcutting the path if it is too long.}
\seclab{compute:shortcut}

A more challenging situation occurs when after $M$ steps the path
$\seg_1, \ldots , \seg_M$ does not hit any segment of $\LX \cup \LY$
or itself.  To avoid creating an even longer path, the idea is to
shortcut the path by adding a single ``cheap'' segment which connects
it to a segment of $\LX \cup \LY$. Note that in this case some blocks
of $\Opt$ might be cut by the shortcut segment, but we ensure that the
total weight of cut blocks from $\Opt$ is negligible. The cut goes
along the boundary of some grid cell. An example can be seen in
\figref{construction-segments-segp-case3}.

The proof of \lemref{cheap:shortcut} below defines a set $\EdgeSet$ of
pairwise disjoint maximal vertical or horizontal segments. The
constructed set $\EdgeSet$ have the following properties:
\smallskip%
\begin{compactenum}[\;(A)]
    \item The segments of $\EdgeSet$ are contained in the grid edges of $\TGrid$.
    \item Each segment of $\EdgeSet$  has
    \begin{compactenum}[(a)]
        \item one of its endpoints on the segments of
        $\brc{\seg_1, \ldots, \seg_M}$, and
        \item its  other endpoint on one of the segments of
    \begin{math}
        \SetX = \LX \cup \LY \cup \brc{\seg_1, \ldots, \seg_M},
    \end{math}
    \end{compactenum}

    \item The segments of $\EdgeSet$ do not intersect any of the segments of $\SetX$ in
    their interior.
\end{compactenum}
\smallskip%
Every segment of $\EdgeSet$ has an associated weight, which is the
total weight of all the blocks that it cuts.  The minimum weight
segment from $\EdgeSet$ is the \emphi{shortcut} of the path, and is
denoted by $\shortcut$.

\begin{lemma}
    \lemlab{cheap:shortcut}%
    For any path there is always a shortcut of weight of at most
    $2\weightX{\Opt}/M$.
\end{lemma}

\begin{proof}
    Consider a path $\pi = \seg_0, \seg_1, \ldots, \seg_M$, and orient
    the segments of the path such that $\seg_i$ is oriented towards
    $\seg_{i+1}$, for all $i$. Let $\pnt_i = \seg_{i} \cap \seg_{i+1}$
    for $i=0,\ldots,M-1$, and let $\pnt_M$ be the endpoint of $\seg_M$
    different from $\pnt_{M-1}$.
    
    By the construction of the path, for any segment $\seg_i$, the two
    endpoints of $\seg_i$ are in two different cells of $\TGrid$. The
    $i$\th endpoint of the path, $\pnt_i$, is contained in some grid
    cell $\GCell_i$, where the path either takes a left turn, or a
    right turn. By \lemref{line-hits-block}, we have
    $\pnt_i \in \interX{\GCell_i}$.
    
    Consider any grid cell $\GCell$ of $\TGrid$ together with its grid
    edge $\edge$, and consider all points $\pnt_i$ such that
    $\seg_{i}$ crosses $\edge$, and $\pnt_i \in \GCell$. For each such
    a point $\pnt_i$ the path first crosses the edge $e$ of $\GCell$,
    and then performs a turn to the right or to the left inside
    $\GCell$. Let $k = n(\edge, \mathrm{left})$ and
    $k' = n(\edge, \text{right})$ denote the total number of such
    turns to the left and to the right, respectively. Then, $\pi$
    crosses $\edge$ at least $k+k'$ times, and all the intersection
    points are pairwise different. Denote by
    $\pnt'_1, \ldots, \pnt'_{k+k'}$ these intersection points, sorted
    by their position along the edge. Since any two fragments of $\pi$
    within $\GCell$ are disjoint, it must hold that the first $k$
    points $\pnt'_1, \ldots, \pnt'_{k}$ correspond to the left turns
    of $\pi$ within $\GCell$, and the last $k'$ points
    $\pnt'_{k+1}, \ldots, \pnt'_{k+k'}$ correspond to the right turns
    (see the figure).

    \parpic[r]{\IncludeGraphics{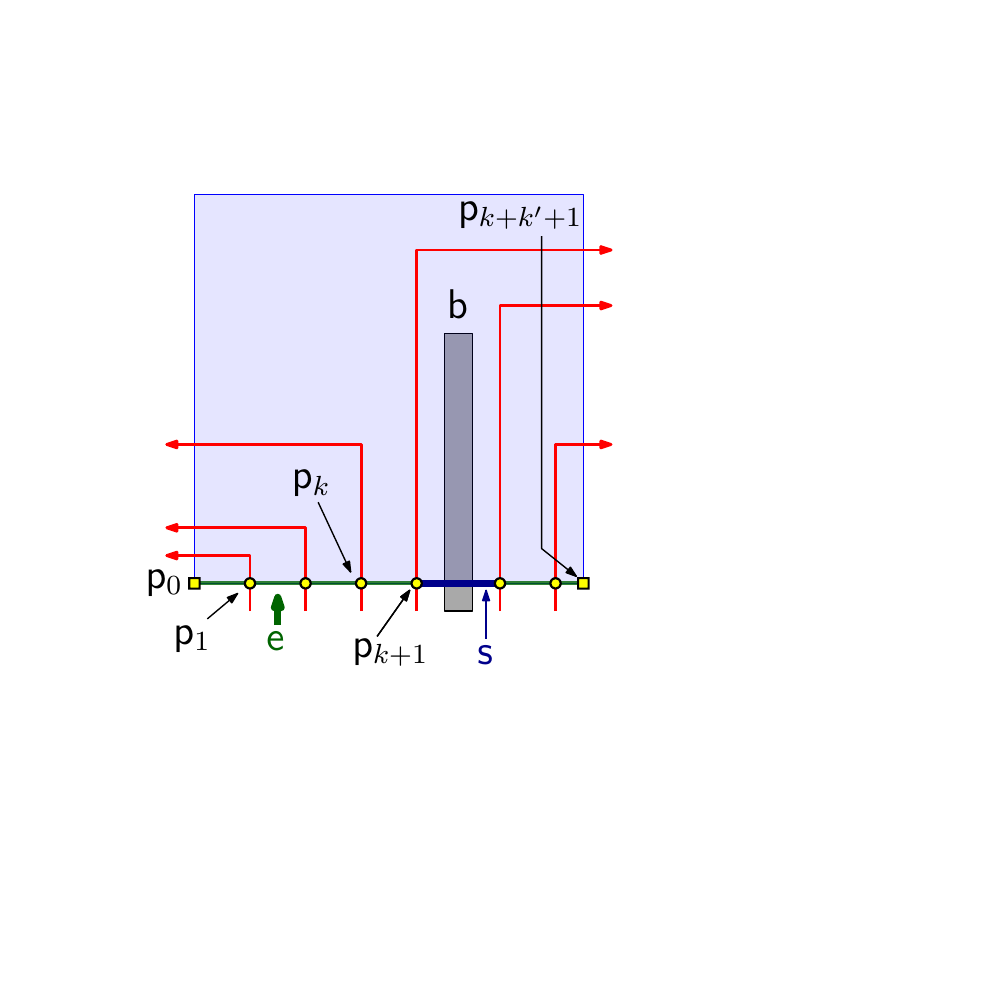}}

    Let $\pnt_0$ and $\pnt_{k +k'+1}$ be the two endpoints of
    $\edge$. We now break $\edge$ into a collection $\CandidX{\edge}$
    of $k + k'$ \emph{candidate segments}, which consists of the
    segments $\pnt_0 \pnt_1, \pnt_1 \pnt_2, \ldots, \pnt_{k-1} \pnt_k$
    and $\pnt_{k+1} \pnt_{k+2}, \ldots, \pnt_{k+k'}
    \pnt_{k+k'+1}$. Notice that we omit the segment
    $\pnt_k \pnt_{k+1}$, which bridges between the left and right
    turns. %
    \NotHandled{%
       \anna{Here we have to prove that all candidate segments satisfy
          some set of properties, in particular: they do not intersect
          any of the segments of $\SetX \cup \SetY$ in their interior,
          either their both endpoints are at $\pi$ or one at $\pi$ and
          the other one at $\SetX \cup \SetY$. For that we might have
          to modify the construction a bit.}\andy{Can't it also be
          that a segment $\pnt_i \pnt_{i+1}$ crosses a segment
          $\seg_j$ such that $\seg_j$ crosses $e$ downwards?}%
    }
 
    Let $\EdgeSet$ be the set of all candidate segments for all cells
    of $\TGrid$ and their corresponding edges. If a grid edge has no
    turns on it, then naturally it contributes no segments to
    $\EdgeSet$. %
    \NotHandled{%
       \anna{Should we discuss the "consistency" for the segments of
          an edge, as each edge is part of two cells?} %
       \andy{Good point. I believe the way it is it is not consistent
          because segments $\seg_i$ are treated differently depending
          on whether they cross $e$ upwards or downwards. I believe
          that these kind of issues could be fixed as follows: we do
          not define the candidate segments globally (i.e., a
          collection for each edge) but argue by the pigeon hole
          principle that if $\pi$ has too many turns then there must
          be one edge $e$ that is crossed very often. Then we do the
          above construction only for this edge $e$ and then ignore
          directions of the segments $\seg_i$ crossing $e$. I believe
          that then we have to make $M$ bigger in order to guarantee
          that there is a cheap shortcut, but it would simplify the
          proof IMO. What do you think?%
       }%
       \anna{I am going to fix the proof} %
       \sariel{I agree what is currently there is a bit hand-wavy, but
          I am too lazy to fix it...}
    }

    Consider any candidate segment $\seg$ obtained while considering a
    cell $\GCell$ together with its edge $\edge$, and observe that any
    block of $\Opt$ that is cut by $\seg$ must end in $\GCell$, as the
    turn of the path $\pi$ corresponding to the segment $\seg$
    prevents any block of $\Opt$ crossing $\seg$ to continue to the
    next cell. As such, every block of $\Opt$ is cut by at most two
    segments from the set $\EdgeSet$.
    
    The size of $\EdgeSet$ is equal to the number of turns in $\pi$
    (i.e., $M$).  We defined the \emph{weight} of a segment
    $\seg \in \CandidAll$, denoted by $\weightX{\seg}$, as the total
    weight of blocks of $\Opt$ crossed by $\seg$. By the above, we
    have that
    $\sum_{\seg \in \CandidAll} \weightX{\seg} \leq 2 \weightX{\Opt}$.
    As such, there is a segment in $\EdgeSet$ of weight at most
    $2\weightX{\Opt} /M$, as claimed.
\end{proof}

From the choice of $\shortcut$ we know that one endpoint of
$\shortcut$ lies on some segment
$\seg_i \in \brc{\seg_1, \ldots, \seg_M}$, and another one on some
segment
$\segA \in \LX \cup \LY \cup \brc{\seg_1, \ldots, \seg_{i-1}}$.
\begin{compactenum}[(A)]
    \item If $\segA \in \brc{\seg_1, \ldots, \seg_{i-1}}$, then adding
    $\shortcut$ to the path will create a cycle. We add the segments
    $\brc{\seg_1, \ldots, \seg_{i'}}$ such that $\seg_{i'}$ is the
    first segment that is part of the cycle. Then we add the portion
    of the segments that form the cycle itself (merging added segments
    that are on the same line and sharing endpoints if necessary). In
    particular, possibly only a part of $\seg_i$ will be added to the
    cycle.
    \item If $\segA \in \LX \cup \LY$, then we add the portion of the
    path till the shortcut, and the shortcut itself, to $\LY$. Again,
    possibly only a part of $\seg_i$ will be added to the cycle.
\end{compactenum}

\bigskip%
This completes the description of the algorithm for computing the
decomposition. The resulting set of segments is denoted by
$\LSet = \LX \cup \LY$.

\subsection{Analyzing the structure of the resulting %
   partition}

Here we prove that the construction above partitions the bounding box
into faces with a nice structure (i.e., \trail{}s and \ring{}s, see
below), the number of resulting faces is small, and each face has low
complexity. This requires quite a bit of care, and the result is
summarized in \lemrefpage{summary} -- the casual reader might want to
skip the details on the first reading.

\begin{defn}
    \deflab{trail:cycle}%
    A rectilinear polygon $\tpoly$ is an \emphi{$L$-shape} if its
    boundary has exactly six edges.

    A rectilinear polygon $\tpoly$ with coordinates in
    $\IntRange{N}^2$ is a \emphi{narrow polygon} if it does not
    contain any vertex of the grid $\TGrid$ in its interior, and if
    for any grid cell $\GCell$ of $\TGrid$, and any connected
    component $\CC$ of $\tpoly \cap \GCell$ (i) is either a rectangle
    or an $L$-shape, and (ii) $\CC$ intersects at most two edges of
    $\GCell$.

    A narrow polygon is a \emphi{\trail} if it has no holes (i.e., it
    is homotopic to a path). A narrow polygon is a \emphi{\ring} if it
    has a single hole.
\end{defn}



\subsubsection{Basic properties of the construction}

The construction immediately implies the following.

\begin{proposition}%
    \proplab{l-size}%
    The set $\LX$ consists of at most $16({1}/{\delta})^{2}+4$
    segments.
\end{proposition}

\begin{defn} 
    \deflab{nicely:connected}%
    A set of segments $\LSet$ is \emphi{nicely connected} if (i) no
    pair of segments of $\LSet$ intersect in their interior, and (ii)
    for any endpoint $\pnt$ of a segment $\seg \in \LSet$ there is a
    segment $\segA \in\LSet$ perpendicular to $\seg$, such that
    $\seg \cap \segA = \pnt$.
\end{defn}

\begin{lemma}
    \lemlab{l:e:properties}%
    The set of segments $\LSet = \LX \cup \LY$ satisfies the following
    properties.
    \begin{compactenum}[\quad(A)]
        \item%
        \itemlab{nicely}%
        $\LSet$ is nicely connected.

        \item We have $\cardin{\LY}\leq \constAmath$ and
        $\cardin{\LSet}\leq \constAmath$.
        
        \item Any segment $\seg \in \LY$ which cuts some blocks of
        $\Opt$ is contained in a single grid edge of $\TGrid$.
        
        \item The total weight of the blocks of $\Opt$ cut by segments
        of $\LSet$ is bounded by $\eps \weightX{\Opt}$.
        
        \item%
        \itemlab{y:z:cross}%
        Every segment of $\LX$ that intersect the interior of the
        square $[0,N]^2$ crosses some grid line of $\TGrid$.  Every
        segment of $ \LY$ intersects some grid line of $\TGrid$.
    \end{compactenum}
\end{lemma}

\begin{proof}
    (A) By construction, no two segments from $\LSet$ overlap or
    intersect in their interior.  For each endpoint $\pnt$ of a
    segment $\seg \in \LX$ which does not hit a perpendicular segment
    from $\LSet$ we added a perpendicular segment touching $\pnt$ to
    the set $\LY$. The path of segments connecting $\pnt$ with a
    segment from $\LSet$ is constructed in such a way, that each
    segment added to $\LY$ has both endpoints touching perpendicular
    segments from $\LSet$. If a segment from $\LY$ gets extended, it
    is extended in such a way that the new endpoint touches a
    perpendicular segment from $\LSet$. Thus, the set of segments
    $\LSet$ is nicely connected.

    \smallskip%
    (B) By \propref{l-size}, $\cardin{\LX} \leq 16(1 /\delta)^{2}+4$.
    For each endpoint of a segment of $\LX$, except for the four
    segments forming the boundary of the input square, we added at
    most $M+1 = O(1/(\eps \delta^{2}))$ segments to the set $\LY$,
    implying $\cardin{\LY} = O\pth{ 1/(\eps \delta^4)}$ and
    $\cardin{\LSet} = O\pth{ 1/(\eps \delta^4)}$.
    
    \smallskip%
    (C) The only segments of $\LSet$ which can cut blocks of $\Opt$
    are the shortcuts, and each such block is contained in some edge
    of $\TGrid$ (see construction of the shortcuts in the proof of
    \lemref{cheap:shortcut}).

    \smallskip%
    (D) By Property (C), we only need to bound the total weight of the
    shortcuts within $\LSet$. By \lemref{cheap:shortcut}, each
    shortcut has a weight of at most $2 \weightX{\Opt}/ M$. As each
    shortcut has been generated while creating a path from an endpoint
    of a segment from $\LX$, and the $4$ segments on the boundary do
    not get extended, there are at most $2 (\cardin{\LX} - 4)$
    shortcuts. The total weight of the blocks of $\Opt$ which have
    been cut is therefore bounded by
    \begin{math}
        (4 (\cardin{\LX}-4) /M) \weightX{\Opt}%
        \leq %
        \pth{\Bigl. 4 \pth{ \bigl. 16({1}/{\delta})^{2}} \cdot \eps
           \delta^{2} / 64} \weightX{\Opt} %
        \leq%
        \eps \weightX{\Opt},
    \end{math}
    by \propref{l-size} and \Eqrefpage{M:value}.

    \smallskip%
    (E) All the segments of $\LX$ are long, implying the first part of
    this claim. By the construction of the segments of $\LY$, all
    shortcuts lie on a single grid edge of $\TGrid$, and all other
    edges cross or intersect the boundary of the grid cell that
    contains one of their endpoints in its interior.
\end{proof}

\paragraph{On the structure of the connections between segments of
   $\LSet$}


\begin{lemma}
    \lemlab{b:horizon}%
    Consider a segment $\seg \in \LSet$, and let $\GCell$ be a grid
    cell such that $\seg$ is a \reach segment (not necessarily
    extremal) for an edge $\edge$ of $\GCell$, and $\seg$ does not
    cross $\GCell$, see \secrefpage{Z:construction}. Let
    $\pnt \in \GCell$ be an endpoint of $\seg$. Then, there exists a
    segment $\segA \in \LX$ perpendicular to $\seg$, such that
    \begin{inparaenum}[(i)]
        \item $\seg \cap \segA = \pnt$,
        \item $\segA$ does not end at $\pnt$, and
        \item $\segA$ is an extremal \reach segment for an edge
        $\edgeB$ of $\GCell$ perpendicular to $\edge$.
    \end{inparaenum}
\end{lemma}

\begin{proof}
    We assume that $\pnt \in \interX{\GCell}$, since the case where it
    is on the boundary of the grid cell can be handled in a similar
    fashion.

    \parpic[r]{%
       \begin{minipage}{0.23\linewidth}%
           \hfill
           \IncludeGraphics{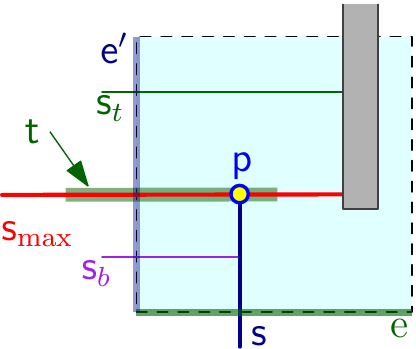}%
           \captionof{figure}{}%
           \figlab{sticking:in:line}
       \end{minipage}%
    }%

    By \lemref{l:e:properties}, $\LSet$ is nicely connected, which
    implies that there is a segment $\segA \in \LSet$ perpendicular to
    $\seg$, such that $\seg \cap \segA = \pnt$.  Without loss of
    generality, assume that $\edge$ is the bottom edge of $\GCell$.
    Let $\segMax$ be a maximal segment containing $\segA$ which does
    not intersect any blocks of $\Opt$ or segments of $\LX$.  We claim
    that $\segMax$ is the bottom-most \reach segment for an edge
    $\edgeB$ of $\GCell$, see \figref{sticking:in:line}. As $\seg$ is
    a \reach segment for $\edge$ and $\seg$ does not cross $\GCell$
    (i.e. $\seg$ does not touch the edge of $\GCell$ parallel to
    $\edge$), $\seg$ cannot be extended at $\pnt$. Either $\seg$ hits
    $\segA$ at $\pnt$, or $\seg$ hits a perpendicular block at
    $\pnt$. In either case $\segMax$ does not end at $\pnt$.
    

    If $\segA \in \LX$, then $\segA$ is long and so is $\segMax$
    (since all the segments of $\LX$ are long). If $\segA \in \LY$
    then $\segA$ cannot be a shortcut, since all shortcuts lie on the
    boundary of grid cells, and $\pnt$ is in the interior of
    $\GCell$. Now, by construction, $\segA$ lies on a long edge of a
    block, and $\segMax$ is at least as long as a block.
    
    The segment $\segMax$ is not contained in $\GCell$, i.e., it
    intersects an edge $\edgeB$ of $\GCell$ perpendicular to $\edge$
    (see \figref{sticking:in:line}).  If $\segMax$ crosses $\GCell$,
    then it is the bottom-most segment intersecting $\edgeB$ and
    maximizing the length of the intersection with $\GCell$. As such,
    by construction, it would be in $\LX$, thus implying the lemma.

    Otherwise, $\segMax$ does not cross $\GCell$.  The segment
    $\segMax$ ends in $\GCell$ by hitting a perpendicular segment from
    $\LX$ or a perpendicular block. This segment or block does not
    intersect the bottom boundary of $\GCell$, as it would yield a
    long segment crossing $\edge$ which reaches further than $\seg$,
    which gives a contradiction, as $\seg$ is a \reach segment for
    $\edge$, see step \itemrefpage{l:y:C} in \secref{Z:construction}.
    The segment or block hit by $\segMax$ crosses the top edge of
    $\GCell$ and does not intersect any segments from $\LSet$ inside
    $\GCell$.  Therefore any segment $\seg_t \in \LSet$ which
    intersects $\edgeB$ above $\segMax \cap \edge'$ satisfies
    \begin{math}
        \lenX{\seg_t \cap \GCell} \leq \lenX{\segMax \cap \GCell}.
    \end{math}
    Similarly, let $\seg_b \in \LSet$ be a segment which intersects
    $\edgeB$ below $\segMax \cap \edge'$.  As $\seg_b$ cannot cross
    $\seg$, and $\segMax$ extends beyond $\seg$, we get
    \begin{math}
        \lenX{\seg_b \cap \GCell} < \lenX{\segMax \cap \GCell}.
    \end{math}

    Thus, $\segMax$ is the bottom-most long segment maximizing the
    length of the intersection with $\GCell$, and so it is the
    bottom-most \reach segment for $\edgeB$. We get that
    $\segMax \in \LX$, and so $\segA = \segMax$.
\end{proof}


\begin{figure}[t]
    \centerline{%
       \begin{tabular}{cccc}
         \IncludeGraphics[page=1]{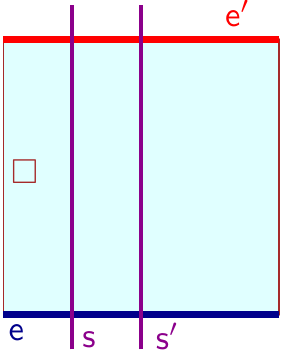}%
         \quad
         &%
           \quad
           \IncludeGraphics[page=2]{figs/entering}%
           \quad
         &%
           \quad
           \IncludeGraphics[page=3]{figs/entering}%
           \quad
         &%
           \quad
           \IncludeGraphics[page=5]{figs/entering}%
         \\
         (A)&(B)&(C) &(D)
       \end{tabular}%
    }
    \caption{Proof of \lemref{faces-line-continuation}.}
    \figlab{f:l:c}
\end{figure}

\begin{lemma}%
    \lemlab{faces-line-continuation}%
    Let $\GCell$ be a grid cell, and let $\edge$ be an edge of
    $\GCell$. Let $\seg, \seg' \in \LSet$ be two segments that cross
    the edge $\edge$, such that there is no segment $\seg' \in \LSet$
    which intersects $\edge$ between these two crossings. Then, there
    is an edge $\edge' \neq \edge$ of $\GCell$ and segments
    $\segA,\segA' \in \LSet$ intersecting $\edgeB$, such that
    $\seg \cap \segA \neq \emptyset$ and
    $\seg' \cap \segA' \neq \emptyset$.
\end{lemma}

\begin{proof}
    We assume w.l.o.g. that $\edge$ is the bottom edge of $\GCell$,
    $\seg$ is to the left of $\seg'$, and that
    \begin{math}
        \lenX{\seg \cap \GCell} \geq \lenX{\seg'\cap \GCell}.
    \end{math}
    There are now several possible cases.%
    \smallskip%
    \begin{compactenum}[(A)]
        \item \textbf{$\seg$ and $\seg'$ intersect the top edge of
           $\GCell$}: The claim holds for $\segA = \seg$ and
        $\segA' = \seg'$, see \figref{f:l:c} (A).
        
        \item \textbf{$\seg$ and $\seg'$ both have endpoints inside
           $\GCell$}: Let $\pnt$ and $\pnt'$ be the two endpoints of
        $\seg$ and $\seg'$ in $\GCell$, respectively. By
        \lemref{l:e:properties} \itemref{nicely}, there are two
        segments $\segA$ and $\segA'$ in $\LSet$ that are
        perpendicular to $\seg$ and $\seg'$, respectively, such that
        $\seg \cap \segA=\pnt$ and $\seg'\cap\segA' = \pnt'$.  The
        segments $\segA$ and $\segA'$ are not contained in
        $\interX{\GCell}$, by \lemref{l:e:properties}
        \itemref{y:z:cross}.  As such, each of them intersects an edge
        of $\GCell$.
        \begin{compactenum}
            \item If $\segA$ intersects the right edge of $\GCell$,
            then $\segA'$ also intersects the same edge, as otherwise,
            either $\segA'$ would cross $\seg$, or $\segA$ would cross
            $\seg'$. See \figref{f:l:c} (B).

            \item If $\segA'$ intersects the left edge of $\GCell$,
            then, by the assumption that
            \begin{math}
                \lenX{\seg \cap \GCell} \geq \lenX{\seg'\cap \GCell}
            \end{math}, we have that
            $\lenX{\seg \cap \GCell} = \lenX{\seg'\cap \GCell}$ and
            $\segA = \segA'$, implying the claim.  See \figref{f:l:c}
            (C).

            \item Otherwise, $\segA$ intersects only the left edge of
            $\GCell$, and $\segA'$ intersects only the right edge of
            $\GCell$, see \figref{f:l:c:D}.  As there are no edges in
            $\LSet$ intersecting $\edge$ between $\seg$ and $\seg'$,
            and by assumption
            \begin{math}
                \lenX{\seg \cap \GCell} \geq \lenX{\seg'\cap \GCell},
            \end{math}
            it follows that $\seg$ is a \reach segment of $\edge$.  By
            \lemref{b:horizon}, $\segA$ is a \reach segment for the
            left edge of $\GCell$ (specifically, the bottom-most
            \reach segment), and it does not end at $\pnt$.

            \smallskip%
            \parpic[r]{
                  \minipageW{\IncludeGraphics[page=4]{figs/entering}}%
                  {
                      \IncludeGraphics[page=4]{figs/entering}%
                      \captionof{figure}{}%
                      \figlab{f:l:c:D}%
                   }%
                }

            The segment $\segA$ does not touch the right edge of
            $\GCell$, as otherwise the claim holds. Now, by applying
            \lemref{b:horizon} to $\segA$, we have that $\segA$ hits a
            perpendicular \reach segment $\seg_T$ in $\GCell$. The
            segment $\seg_T$ does not intersect the bottom edge
            $\edge$ of $\GCell$, as otherwise
            $\lenX{\seg_T \cap \GCell} > \lenX{\seg \cap \GCell}$, and
            that would contradict $\seg$ being a \reach segment for
            $\edge$. Thus, $\seg_T$ must intersect the top edge of
            $\GCell$ and has an endpoint inside $\GCell$. Applying
            \lemref{b:horizon} to $\seg_T$ in turn, implies that it
            hits a perpendicular \reach segment $\seg_R$ that must
            intersect the right edge of $\GCell$.  If
            $\seg_R \neq \segA'$ then, by \lemref{b:horizon}, $\seg_R$
            hits another \reach segment that crosses $\edge$, but this
            \reach segment must be $\seg$, see \figref{f:l:c:D}.  The
            case that $\seg_R = \segA'$ follows verbatim by the same
            analysis, by observing that $\segA'$ does not hit $\seg'$
            (i.e., $\segA'$ hitting a segment $\bSeg$ implies that
            $\segA'$ endpoint is in the interior of $\bSeg$). For the
            pairs $\seg, \seg_R$ and $\seg', \segA'$ the claim now
            follows.
        \end{compactenum}
        
        \item \textbf{$\seg$ intersects the top edge of $\GCell$, and
           $\seg'$ does not}: By \lemref{l:e:properties}, there is a
        segment $\segA' \in \LSet$ that is perpendicular to $\seg'$ at
        its endpoint $\pnt' \in \GCell$, and furthermore $\segA'$ is
        not contained in $\interX{\GCell}$. As $\seg$ is to the left
        of $\seg'$, the segment $\segA'$ intersects the right edge
        $\edge'$ of $\GCell$ (see \figref{f:l:c} (D)). Let $\seg_R$ be
        a \reach segment for $\edge'$. Such a segment exists, as the
        maximal segment containing $\segA'$ is a candidate to be a
        \reach segment for $\edge'$.

        We claim that $\seg_R$ touches $\seg$. If $\seg_R$ crosses
        $\GCell$, then $\seg_R$ must touch $\seg$ (and either $\seg$
        or $\seg_R$ goes along an edge of $\GCell$). If $\seg_R$ does
        not cross $\GCell$, then, by \lemref{b:horizon}, $\seg_R$ hits
        a perpendicular \reach segment $\segB$ in $\GCell$.  The
        segment $\segB$ is a vertical \reach segment, and must be as
        long as $\seg$ inside $\GCell$; that is, it must cross
        $\GCell$. However, there is no segment in between $\seg$ and
        $\seg'$ crossing $\edge$, which implies that $\seg = \segB$.
        We conclude that $\seg_R$ intersects $\seg$, as desired.
    \end{compactenum}
\end{proof}

\subsubsection{Faces of the partition}

The segments of $\LSet$ subdivide the input square into a collection
of faces which are the connected components of
$\pbrc{0,N\bigr.}^2 \setminus \ULSet$ (as such, the faces are open
sets), where $\ULSet = \ds \cup^{}_{\seg \in \LSet} \seg$. Let
$\facesX{\LSet}$ denote the set of faces of this partition, and
$\facesPZ$ the set of faces that contain at least one block of $\Opt$.
Our purpose here is to prove that the number of resulting faces is
bounded by a constant and that each face is either a trail or a ring.

\paragraph{Inside a grid cell, faces are rectangles or $L$-shaped.}

\begin{observation}
    \obslab{narrow}%
    No face of $\facesPZ$ contains a vertex of $\TGrid$ in its
    interior.
\end{observation}

\begin{lemma}%
    \lemlab{corridor-shapes}%
    Consider a face $\face \in \facesPZ$ and let $\GCell$ be a grid
    cell with $\face \cap \GCell \neq \emptyset$.  Consider a
    connected component $\CC$ of $\face \cap \GCell$. Then
    $\interX{\CC}$ is the interior of a rectangle or the interior of
    an $L$-shape.  Also, $\CC$ has non-empty intersection with at most
    two edges of $\GCell$.
\end{lemma}

\begin{proof}
    Consider the case when $\CC$ has non-empty intersection with some
    block $\Block \in \Opt$ contained in $\face$.
    Let $\edge$ be an edge of $\GCell$ such that
    $\edge \cap \interX{\Block} \neq \emptyset$, and assume
    w.l.o.g. that $\edge$ is the bottom edge of $\GCell$ (see
    \figref{corridor-in-cell-2}). Let $\seg, \seg' \in \LSet$ be
$L$-    segments which intersect $\GCell$ and intersect $\edge$ at some
    points $\pnt_\seg$ and $\pnt_{\seg'}$, respectively, such that
    $\pnt_\seg$ is to the left of $\edge \cap \interX{\Block}$,
    $\pnt_{\seg'}$ is to the right of $\edge \cap \interX{\Block}$,
    and no segment from $\LSet$ which intersects $\GCell$ touches
    $\edge$ in between $\pnt_\seg$ and $\pnt_{\seg'}$. Such segments
    exist, as no segment from $\LSet$ intersects $\edge$ inside
    $\edge \cap \interX{\Block}$, the leftmost long segment
    intersecting $\GCell$ and touching $\edge$ (which belongs to
    $\LX$) either contains the left edge of $\Block$ or is to the left
    of it, and the rightmost long segment intersecting $\GCell$ and
    touching $\edge$ (which also belongs to $\LX$) either contains the
    right edge of $\Block$ or is to the right of it.
    
    \begin{figure}
        \begin{centering}
            
            \centerline{%
               \IncludeGraphics{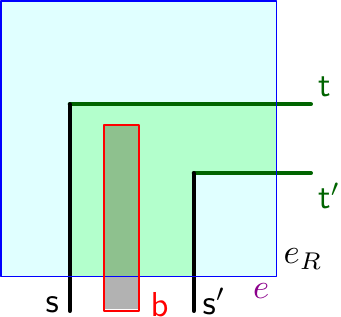}%
            }%
            
            \par\end{centering}
        \caption{A connected component of a face
           $\face \in \FC_+(\LSet)$ within a grid cell (denoted by a
           shaded area) must have a simple shape, i.e., it is either a
           rectangle or an $L$-shape.}%
        \figlab{corridor-in-cell-2}
    \end{figure}
    
    Parts of the segments $\seg, \seg'$ lie on the boundary of
    $\CC$. Denote by $\segA$ and $\segA'$ the lines given by applying
    \lemref{faces-line-continuation} to $\seg$ and $\seg'$. If $\segA$
    and $\segA'$ both intersect the top edge of $\GCell$ then the
    claim follows and $\CC$ is a rectangle.
    Otherwise, assume w.l.o.g. that both $\seg$ and $\seg'$ intersect
    the right edge $\edge_R$ of $\GCell$, and assume that $\segA$ is
    the bottom-most segment touching $\seg$ and $\edge_{R}$ and
    $\segA'$ is the topmost segment touching $\seg'$ and
    $\edge_{R}$. From \lemref{l:e:properties} the set of lines $\LSet$
    is nicely connected, and by construction, all lines in $\LSet$
    with non-empty intersection with $\interX{\GCell}$ for some grid
    cell $\GCell$ touch the boundary of $\GCell$.  Hence, there can be
    no segment of $\LSet$ within $\GCell$ intersecting $\edge_{R}$
    between $\edge_{R}\cap\segA$ and $\edge_{R}\cap\segA'$, and the
    claim follows.

    \medskip

    The above proves the claim for any $\CC$ which has non-empty
    intersection with some block $\Block \in \Opt$ contained in
    $\face$.  However, there can potentially be a grid cell $\GCell$
    and a face $\face$ such that a connected component $\CC$ of
    $\face \cap \GCell$ intersects no blocks of $\Opt$.  The proof in
    this case follows by propagating the property to the adjacent grid
    cells of $\GCell$ intersecting $\face$: Consider a connected
    component $\CC'$ of $\face \cap \GCellA$, such that
    $\CC \cap \CC' \neq \emptyset$ and assume that the claim holds for
    $\CC$ (since $\face \in \facesPZ$ there must be one cell $\GCell$
    such that $\GCell \cap \face \ne \emptyset$).

    \parpic[r]{\IncludeGraphics{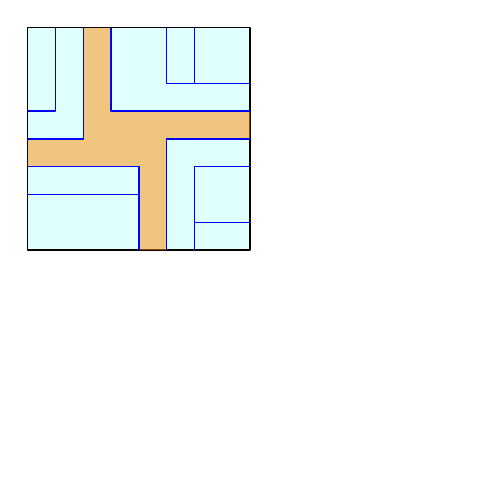}}

    If $\CC'$ intersects a block of $\Opt$, then the claim follows by
    the above.  So assume that $\CC'$ does not intersect any block of
    $\Opt$. Consider the segments of $\LX$ clipped to
    $\interX{\GCellA}$.  \NotHandled{%
       \anna{There is only one pathological case where there are no
          such segments.} %
    }%
    Each such segment either crosses $\GCellA$, or it has a loose end
    inside $\GCellA$. In the second stage of the construction every
    such loose end is connected up to a path of segments, importantly
    by a segment that leaves the interior of the cell. Then, during
    the second stage, a path might end up hitting an existing
    segment. Thus, the set of segments of $\LSet$ clipped to the
    interior of $\GCellA$ is formed by the union of segments that
    cross the cell, and $L$-shaped curves with their endpoints on the
    boundary of $\GCellA$.  None of these curves cross each other
    inside $\GCellA$, although they might have a non-empty
    intersection. See figure on the right for an example.

    Consider a partition of $\GCellA$ that might be formed by such a
    collection of $L$-shaped curves. Consider a face $\CC'$ in such a
    partition of $\GCellA$ which is not a rectangle. Then it must have
    a corner $\pnt$, such that the angle inside $\CC'$ is $270$
    degrees.  The only way such a corner can be formed is because one
    of the segments adjacent to $\pnt$ hits a block of $\Opt$. But
    that implies that $\interX{\CC'}$ intersects a block of $\Opt$,
    and as such, by the above it is $L$-shaped.
\end{proof}

\paragraph{Shortcuts are anchored at segments visiting both cells.}

\begin{lemma}%
    \lemlab{l-on-grid-boundary}%
    Let $\GCell$ and $\GCellA$ be two neighboring grid cells and let
    $\edge = \GCell \cap \GCellA$. Let $\segA \subseteq \edge$ be a
    maximal segment of $\LSet$ contained in $\edge$, and assume that
    $\segA$ does not contain an endpoint of $\edge$. Then $\segA$ is
    incident with segments $\seg, \seg' \in \LSet$ (where possibly
    $\seg = \seg'$) such that $\seg$ intersects $\interX{\GCell}$ and
    $\seg'$ intersects $\interX{\GCellA}$.
\end{lemma}
\begin{proof}
    From the construction of $\LX$ and $\LY$, the segment $\seg$
    consists of one or multiple shortcut segments, as any other
    segment from $\LSet$ would touch an endpoint of $\edge$.

    \parpic[r]{%
       \minipageW{
          \IncludeGraphics{\si{figs/zig_zag}}%
       }{%
          \IncludeGraphics{\si{figs/zig_zag}}%
       }
    }%

    Let $\segB \subseteq \segA$ be the first shortcut added to
    $\LY$. There are two segments $\seg, \seg' \in \LY \cup \LX$ such
    that the endpoints of $\segB$ are contained in these two
    segments. Assume $\seg$ intersects $\interX{\GCell}$.

    If $\seg$ also intersects $\interX{\GCellA}$, then the claim
    holds. Similarly, if $\seg'$ intersects $\interX{\GCellA}$, then
    we are done. So, it must be that $\seg'$ intersects
    $\interX{\GCell}$, and both $\seg$ and $\seg'$ have an endpoint on
    $\edge$, see figure on the right.  Let $\pnt$ (resp. $\pnt'$) be
    the endpoint of $\seg$ (resp. $\seg'$) on $\edge$. We have that
    $\seg$ does not hit a perpendicular block at $\pnt$, as otherwise
    this would induce a leftmost long segment crossing the top or the
    bottom edge of $\GCellA$. In turn, such a segment, by
    construction, is in $\LX$, see \itemrefpage{l:y:A} in
    \secref{Z:construction}.  This would contradict the assumption
    that $\segA$ does not contain an endpoint of $\edge$.  Similarly,
    $\seg$ cannot hit any vertical segment of $\LX$ at $\pnt$, since
    all such segments are long,

    So, it must be that $\seg$ is in $\LY$, and furthermore, it got
    shortened when the shortcut $\segB$ was created (because, this is
    the only way for $\seg$ to have an endpoint on $\edge$). We can
    apply verbatim the same logic to $\seg'$. However, by
    construction, it is not possible that when the shortcut $\segB$
    was introduced between $\seg$ and $\seg'$, both of them got
    clipped.%
    \footnote{%
       Underlying our argument here is the monotonicity of $\LY$: As
       the construction continues, the union of segments in this set
       only grows. }
\end{proof}

\paragraph{Faces of $\facesPZ$ do not fork.}

Now we study the structure of the faces in $\facesPZ$ at the boundary
of the grid cells. In the following lemma we show that multiple
connected components of a face inside one grid cell $\GCellA$ cannot
merge into one component in a neighboring grid cell $\GCell$.

\begin{lemma}%
    \lemlab{corridor-at-boundaries}%
    Let $\GCell$ and $\GCellA$ be two grid cells sharing a common edge
    $\edge$.  Consider a face $\face \in \facesPX{\LSet }$ such that
    $\face \cap \GCell \neq \emptyset$, and let $\CC$ be a connected
    component of $\face \cap \GCell$ such that
    $\CC \cap \edge \neq \emptyset$.  Then there is exactly one
    connected component $\CC'$ of $\face \cap \GCellA$ such that
    $\CC \cap \CC' \neq \emptyset$.
\end{lemma}

\begin{proof}
    We remind the reader that the grid cells are closed sets, but
    faces of $\facesPX{\LSet }$ are open sets.  Let
    $\pnt \in \CC \cap \edge$. Let $\CC'$ be a connected component of
    $\face \cap \GCellA$ containing $\pnt$. Clearly,
    $\pnt \in \CC \cap \CC' \neq \emptyset$.

    \smallskip%
    \parpic[r]{%
       \minipageW{\IncludeGraphics{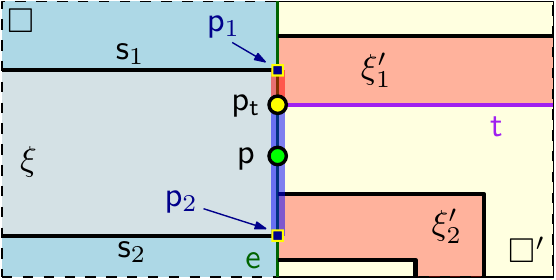}}%
       {%
          \IncludeGraphics{figs/19}%
          \captionof{figure}{}%
          \figlab{no-forks}%
       }
    }
    
    We claim that $\CC'$ is unique. Assume otherwise, i.e., that there
    are two connected components $\CC_1'$ and $\CC_2'$ of
    $\face \cap \GCellA$ which have non-empty intersection with $\CC$
    (and they are consecutive along $\edge$). Let $\segA \in \LSet$ be
    a segment intersecting $\interX{\GCellA}$, such that $\segA$
    intersection with $\edge$ is between $\CC_1' \cap \CC$ and
    $\CC_2' \cap \CC$ (see \figref{no-forks}). Such a segment exists,
    as $\CC_1' \cap \GCellA$ and $\CC_2' \cap \GCellA$ are not
    connected.

    Let $\seg_1, \seg_2 \in \LSet$ be the segments intersecting
    $\GCell$, touching $\edge$ and bounding $\CC \cap \GCell$. Then
    there is no segment $\seg_3 \in \LSet$ which intersects the
    interior of $\GCell$ and intersects $\edge$ between
    $\pnt_1 = \seg_1 \cap \edge$ and $\pnt_2 = \seg_2 \cap \edge$.  As
    $\pnt_\segA = \segA \cap \edge$ is between $\CC_1' \cap \CC$ and
    $\CC_2' \cap \CC$, it holds that $\pnt_\segA$ is between $\pnt_1$
    and $\pnt_2$. Implying that $\segA$ ends at
    $\pnt_\segA \in \edge$, and does not intersect the interior of
    $\GCell$.
    
    We claim that either $\pnt_\segA\pnt_1$ or $\pnt_\segA\pnt_2$ is
    contained in $\LSet$, which contradicts that either
    $\CC \cap \CC_1 \neq \emptyset $ or
    $\CC \cap \CC_2 \neq \emptyset$.  Hence, the component $\CC'$ is
    unique.
    
    As $\pnt_\segA$ is an endpoint of $\segA$, $\segA$ touches a
    perpendicular segment from $\LSet$ at $\pnt_\segA$.  Let $\seg$ be
    a maximal segment, such that
    \begin{inparaenum}[(i)]
        \item $\seg$ is contained in $\edge$,
        \item $\seg$ contains $\pnt_\segA$, and
        \item $\seg$ is contained in $\LSet$.
    \end{inparaenum}
    From \lemref{l-on-grid-boundary} segment $\seg$ contains an
    endpoint of $\edge$, or $\seg$ is incident with a segment
    intersecting $\interX{\GCell}$. In either case, one of the
    segments $\pnt_\segA \pnt_1 $, $\pnt_\segA\pnt_2$ is contained in
    $\seg$, i.e., it is contained in $\LSet$.
\end{proof}

\paragraph{Faces of $\facesPZ$ are either %
   \trail{}s or \ring{}s.}

\begin{lemma}
    \lemlab{summary}%
    A face $\face$ of $\facesPZ$ is either a \trail or a \ring, see
    \defrefpage{trail:cycle}. The number of faces of $\facesPZ$ is
    $O( 1/(\delta \eps^2 ))$, and each face has $\constAmath$
    vertices.
\end{lemma}

\begin{proof}
    Consider a face $\face$.  Consider the set of all connected
    components of $\face$ when clipped to the grid cells of $\TGrid$;
    that is,
    \begin{math}
        \CDC = \Set{\face \cap \GCell}{ \GCell \text{ cell of }
           \TGrid, \text{ and } \face \cap \GCell \neq \emptyset }.
    \end{math}
    Let $\Graph = \pth{\CDC, \Edges}$, where $\CC \CCA \in \Edges$, if
    $\CC, \CCA \in \CDC$ and $\CC \cap \CCA \neq \emptyset$.  By
    definition the graph $\Graph$ is connected.  By
    \lemref{corridor-shapes} and \lemref{corridor-at-boundaries}, all
    the vertices of $\Graph$ are either of degree one or two. That
    implies that $\Graph$ is either a cycle or a path.

    As for the number of faces, consider the construction of $\LSet$,
    just before $\LY$ is computed. At this stage, there are
    $O( \cardin{\LX})$ faces. Now, every endpoint of a segment of
    $\LX$, might give rise to one new face during the construction of
    $\LY$. As such, the total number of faces is bounded by
    $O( \cardin{\LX})$. The later part of the claim now follows by
    \lemrefpage{l:e:properties}.
\end{proof}

Note that we do not need to care about faces in
$\mathcal{F}\setminus\facesPZ$ since they do not contain any blocks
from the optimal solution.


\subsection{The approximation algorithm}

The basic idea of the algorithm is to start from the decomposition of
\secref{partition:blocks} of the input into \trail{}s and \ring{}s,
see \lemref{summary}.  We show that for the case of \trail{}s, one can
compute the independent set inside them optimally. We also show that
\ring{}s can be decomposed into collections of \trail{}s, with only a
small loss in the objective. Combining the two algorithms results in
the desired approximation algorithm. Finally, we show how to adapt the
resulting algorithm to handle large rectangles (and not only blocks).

\subsubsection{Computing the maximum weight %
   independent set of blocks inside a \trail}

\paragraph{Trails can be recursively divided without cutting any
   blocks.}

\begin{lemma}
    \lemlab{trail:blazers}%
    Consider a \trail $\tpoly$ whose boundary has $k$ vertices, for
    $k \geq 4$.  Let $\Opt$ be an independent set of blocks contained
    in $\tpoly$, such that $\cardin{\Opt} \geq 2$.  Then, there are
    two \trail{}s $\tpoly_1$ and $\tpoly_2$, with non-empty interior,
    each with at most $k$ vertices, such that
    $\tpoly = \tpoly_1 \cup \tpoly_2$, and every block of $\Opt$ is
    contained either in $\tpoly_1$ or $\tpoly_2$.
\end{lemma}

\begin{proof}
    Consider a segment $\seg$ with integer coordinates, such that (i)
    $\interX{\seg} \subseteq \interX{\tpoly}$, (ii) both endpoints of
    $\seg$ are on the boundary of $\tpoly$, and (iii) $\seg$ does not
    intersect the interior of any block of $\Opt$.  A segment $\seg$
    with the above properties is a \emphi{cut segment} for $\tpoly$.
    Clearly, if such a segment exists, we cut $\tpoly$ along $\seg$,
    and the claim holds.
    
    \parpic[r]{%
       \begin{minipage}{5.6cm}%
           \IncludeGraphics{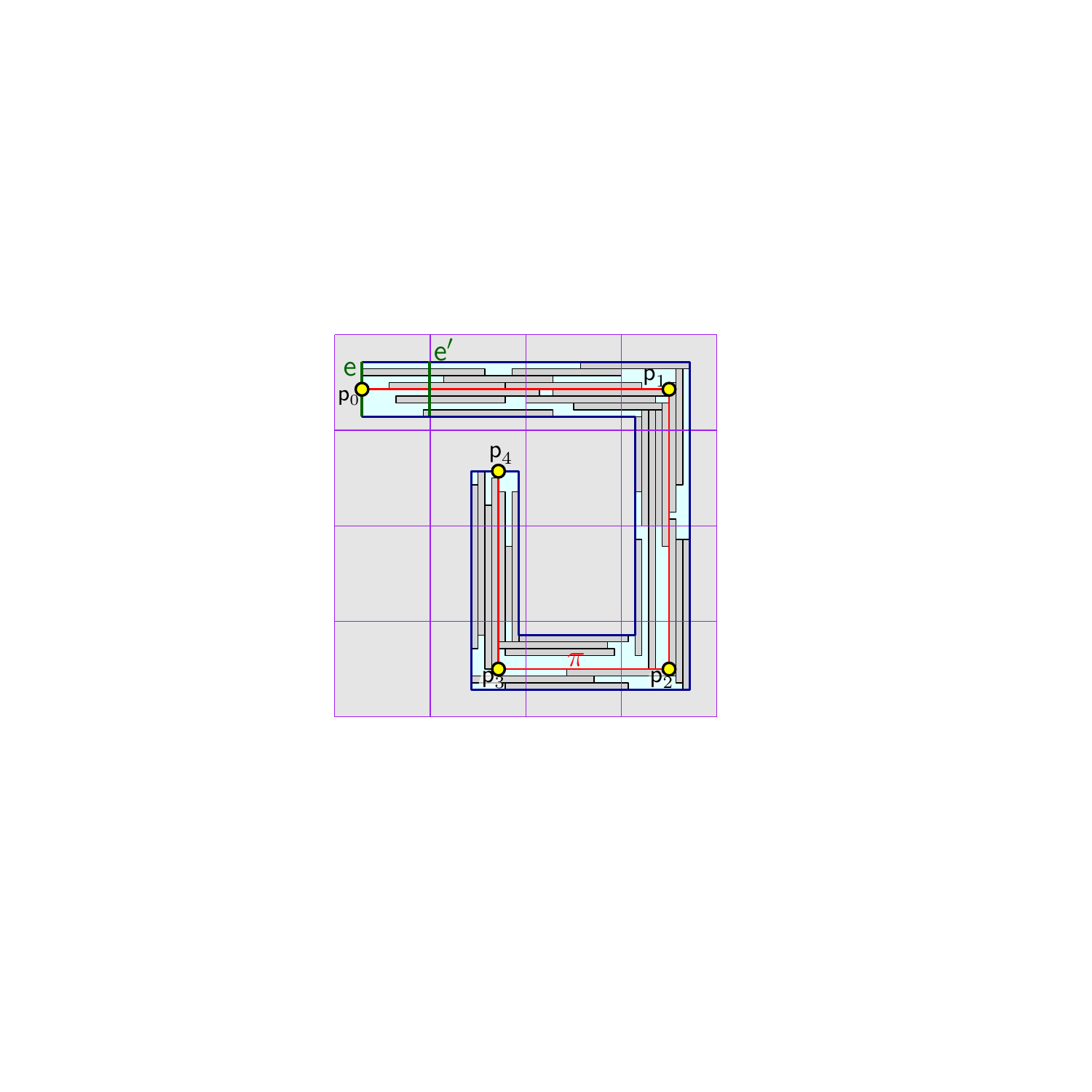}%
           \captionof{figure}{}
           \figlab{cut:trail}%
       \end{minipage}%
    }%
    We now assume that there is no cut segment, i.e., $\tpoly$ can not
    be shrank any further. Arguing as in \lemref{summary}, there must
    be a grid cell $\GCell$, and a connected component $\CC$ of
    $\GCell \cap \tpoly$, which is connected to the rest of the trail
    through a single segment, denoted by $\edge'$, on the boundary of
    $\GCell$, and assume without limiting generality that $\edge'$ is
    vertical, and on the right side of $\GCell$.  By
    \lemref{corridor-shapes}, the polygon $\CC$ is either a rectangle
    or an $L$-shaped polygon. If $\CC$ is $L$-shaped, then consider
    the maximal rectangle $\rect \subseteq \CC$ that has $\edge'$ as
    its right side. Any block $\Block \in \Opt$ that intersects the
    interior of $\CC$, must intersect $\edge'$, as it is the only way
    out of $\CC$ in this grid cell. As such, $\Block \cap \CC
    \subseteq \rect$. We conclude, that as far as the optimal
    solution, we can replace $\CC$ by $\rect$ in $\tpoly$. But this
    imply that there is a cut segment (i.e., the edge separating
    $\rect$ from the rest of $\CC$). 

    Thus, the polygon $\CC$ is a rectangle, its interface to the rest
    of $\tpoly$ is through the edge $\edge'$, and let $\edge$ be edge
    of $\CC$ parallel to $\edge'$.
    


    Assume for the time being that $\edge$ has length larger than one,
    and let $\pnt_0$ be a point with integer coordinates in the
    interior of $\edge$. Shoot a ray from $\pnt$ into the interior of
    $\tpoly$. This ray must hit a perpendicular block of $\Opt$ at a
    point $\pnt_1$, and as in construction of $\LY$ (see
    \secrefpage{app:construction}). We continue the ray shooting from
    $\pnt_1$ along the edge of $\Block$. Since $\tpoly$ is narrow, and
    inside a grid cell of $\TGrid$ each connected component of
    $\tpoly$ is either a rectangle or is $L$-shaped (see
    \lemref{corridor-shapes}), it follows that the generated path must
    end at an edge of $\tpoly$. The generated path
    $\Path$ does not intersect the interior of any block of $\Opt$,
    its two endpoints are on the boundary of $\tpoly$, and it cuts
    $\tpoly$ into two rectilinear polygons with at most $k$ vertices
    each.

    \NotHandled{\sariel{Add figure for the next paragraph.}}%
    
    The case that $\edge$ is of length one requires a special
    handling. If the trail is a single rectangle then the claim
    trivially holds, as it contains the blocks inside it in a linear
    order, and one can easily cut the trail after the first
    block. Otherwise, the trail must have a turn in it as it is being
    traversed from $\CC$. Assume, without limiting generality, that
    this is a right turn, and perform a ray shooting, as above from
    the bottom endpoint of $\edge$. If the ray shooting cut across
    $\tpoly$, without cutting any block, then we have a cut
    segment. Otherwise, it must have hit a vertical block $\Block$ at
    a point $\pnt_0$. If $\pnt_0$ is a vertex of $\tpoly$, then it is
    easy to verify that one can cut $\tpoly$ by a vertical segment
    through $\pnt_0$. As such, it must be that $\pnt_0$ is in the
    interior of $\tpoly$, and the argument above for the interior ray
    shooting applies, as can be easily verified, and it implies the
    claim.
\end{proof}

\paragraph{The approximation algorithm for trails.}

\begin{lemma}
    \lemlab{approx:trails}%
    Let $N$, $\eps$ and $\delta$ be parameters as above, let $\BSet$
    be a set of $m$ weighted ($\delta$-large) blocks with vertices in
    $\IntRange{N}^2$, and let $k$ be a parameter. Let $\PolySetA$ be
    the set of all possible \trail{}s within $\IntRange{N}^2$ with at
    most $k$ vertices each. For a \trail $\tpoly \in \PolySetA$, let
    $\woptX{\tpoly}$ be the weight of the maximum weight independent
    set of blocks of $\BSet$. Then, one can compute \emph{exactly} the
    value of $\woptX{\cdot}$ for all the \trail{}s of $\PolySetA$ in
    time $O \pth{ N^{4k} m k \log k }$.
\end{lemma}

\begin{proof}
    Let $\PolySet$ be the set of all rectilinear polygons with at most
    $k$ vertices from the set $\IntRange{N}^2$. An easy calculation
    shows that $\cardin{\PolySet} \leq 2N^{2k}$. We verify for every
    polygon in $\PolySet$ that it is a \trail, and if not, we reject
    it. Let $\PolySetA$ be the resulting set of \trail{}s, Clearly,
    verifying if a polygon of $\PolySet$ is a \trail can be done in
    $O( kN \log ( k N ) )$ time. As such, computing $\PolySetA$ takes
    $O\pth{ k N^{2k+1} \log (k N)}$ time.

    Now, for every \trail $\tpoly \in \PolySetA$, we compute the value
    of $\woptX{\tpoly}$. First, we check for the maximum weight single
    block contained in $\tpoly$. Next, we consider the possibility
    that the optimal solution within $\tpoly$ consists of more than
    one block. From \lemref{trail:blazers}, $\tpoly$ can then be
    broken into two smaller trails $\tpoly_1, \tpoly_2$, each with at
    most $k$ vertices, such that the optimal solution for $\tpoly$ is
    a union of the optimal solutions for $\tpoly_1$ and $\tpoly_2$. A
    such, we try all such partitions (naively, there are
    $\cardin{\PolySetA}$ such partitions), and for each such partition
    we verify that it is valid, and then we compute the best solution
    out of all such possibilities.

    Given $\tpoly$ and $\tpoly_1$, computing $\tpoly_2$ can be done in
    $O(k \log k)$ time using sweeping. As such, the resulting
    algorithm has running time
    $O\pth{ \cardin{\PolySetA}^2 k \log k}$.
\end{proof}

\subsubsection{Computing the maximum weight %
   independent set of blocks inside a \ring}

The situation here is somewhat more involved. As before, if one can
break the \ring into two smaller \ring{}s, or into a \ring and \trail,
without intersecting any blocks from the optimal solution, the
algorithm will try this partition. Alternatively, it will perform a
decomposition is into a \ring{} and a \trail, where there would be a
certain (small) loss in the objective.

\begin{figure}[t]
    \begin{tabular}{cccc}
      \IncludeGraphics[page=1,width=0.22\linewidth]{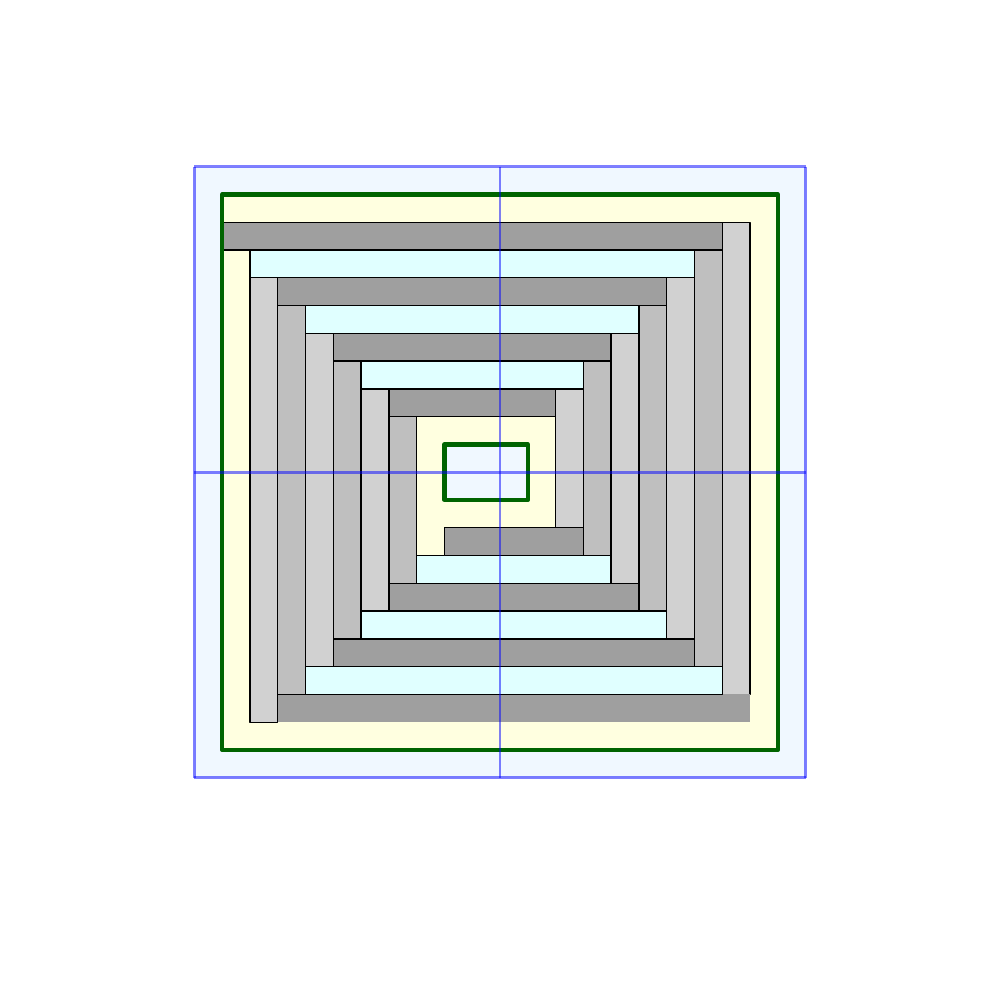}
      &
        \IncludeGraphics[page=2,width=0.22\linewidth]{figs/spiral}
      &
        \IncludeGraphics[page=3,width=0.22\linewidth]{figs/spiral}
      &
        \IncludeGraphics[page=4,width=0.22\linewidth]{figs/spiral}
      \\
      (A)&(B)&(C)&(D)
    \end{tabular}
    \caption{%
       (A) A \ring with the set of blocks $\Opt$. %
       (B) The resulting decomposition into a \trail and a smaller
       \ring. %
       (C) Breaking the long \trail into shorter \trail{}s. %
       (D) The union of a suffix of these \trail{}s forms a \ring (or
       a \trail). %
    }%
    \figlab{folded:ring}
\end{figure}

\paragraph{\Ring{}s can be divided without too much loss.}

\begin{lemma}
    \lemlab{ring:blazers}%
    Consider a \ring $\Polygon$ with $k$ vertices, for $k \geq 4$.
    Let $\Opt$ be an independent set of blocks contained in
    $\Polygon$, such that $\cardin{\Opt} \geq 2$.  Then, one the
    following holds.
    \begin{compactenum}[(A)]
        \item There is a \trail $\tpoly$ such that
        $\clX{\tpoly} = \clX{\Polygon}$ with $O(k/\eps)$ vertices,
        such that $\Opt \subseteq \tpoly$.

        \item There are two interior disjoint polygons $\Polygon_1$
        and $\Polygon_2$, with at most $k$ vertices each, such that
        $\Polygon = \Polygon_1 \cup \Polygon_2$, and such that
        $\Polygon_1$ and $\Polygon_2$ do not cut any block of
        $\Opt$. Furthermore, each of the two polygons is either a
        \trail or a \ring.

        \item There are interior disjoint \trail{}s
        $\tpoly_1, \ldots, \tpoly_{m}$, such that:
        \begin{compactenum}[(I)]
            \item $\bigcup_i \tpoly_i = \Polygon$.

            \item The blocks of $\Opt$ cut by the boundaries of
            $\tpoly_1, \ldots, \tpoly_{m}$ have a total weight of at
            most $(\eps/4)\weightX{\Opt}$.

            \item Every \trail $\tpoly_i$ has $O(1/\eps)$ vertices.

            \item For any $i$, $\cup_{j=i+1}^{m} \tpoly_j$ is a \ring
            (or a \trail) with $k+O(1)$ vertices.
        \end{compactenum}
    \end{compactenum}
\end{lemma}
\begin{proof}
    Consider any edge $\edge$ of the grid $\TGrid$, that intersects
    $\Polygon$, and the segments of $\Polygon \cap \edge$. If any of
    these segments does not intersect a block of $\Opt$, we cut
    $\Polygon$ along it, ending up with a \trail. We obtain case (A).
    

    Otherwise, let $x_0$ be the minimum $x$-coordinate of any point in
    $\Polygon$, and let $\pnt' = (x_0, y_0)$ be the bottom-most vertex
    of $\Polygon$ with $x$-coordinate equal to $x_0$. Let $\GCell$ be
    the grid cell of $\TGrid$ containing $\pnt'$, and let $\CC$ be the
    connected component of $\Polygon \cap \GCell$ that has $\pnt'$ as
    a vertex. The polygon $\CC$ is $L$-shaped, and let $\edge_t$ and
    $\edge_r$ be its top and right edges, respectively.
    
    Let $\BSet_{\CC}$ be the set of all blocks of $\Opt$ crossing
    either $\edge_t$ or $\edge_r$. Let $\Block \in \BSet_{\CC}$ be the
    block with minimum $L_1$-distance from $\pnt'$, and assume
    w.l.o.g. that $\Block$ intersects $\edge_r$. Let $\seg$ be the top
    edge of $\Block$.  Extend $\seg$ to the left till it hits the
    outer boundary of $\CC$, and denote this endpoint of $\seg$ by
    $\pnt_0$. Note, that after extending $\seg$ does not cut any
    blocks of $\Opt$ , as $\Block$ is the closet block to $\pnt'$.

    The point $\pnt_0$ lies in the interior of an edge of $\GCell$,
    and in the interior of an edge of $\Polygon$. As in
    \lemref{trail:blazers}, compute a path $\Path$ in $\Polygon$ by
    shooting a ray from $\pnt_0$ along $\seg$, and turning whenever
    the ray hits a block of $\Opt$. This process terminates when one
    of the following happens. %
    \smallskip
    \begin{compactenum}[\quad(I)]
        \item The path $\Path$ hits the outer boundary of
        $\Polygon$. In this case, it is easy to verify that this
        splits $\Polygon$ into a \trail and a \ring, without breaking
        any block in the process, and both the \trail and the \ring
        have at most $k$ vertices each. This corresponds to case (B).
        
        \item The path $\Path$ hits itself. Let $C$ be the portion of
        the loop formed by this path, and observe that $C$ has at most
        $k/2$ vertices, as can be easily verified, and $C$ breaks
        $\Polygon$ into two \ring{}s, as desired. This is again case
        (B).
    \end{compactenum}%
    \smallskip%
    The only remaining possibility is that the path $\Path$ hits the
    inner boundary of $\Polygon$.  Cutting $\Polygon$ along $\Path$,
    we obtain a new \trail polygon $\PolygonB$, which potentially can
    have a large number of vertices (i.e., $\PolygonB$ is a spiral
    folded over itself).  See \figref{folded:ring}.

    If $\PolygonB$ has $O(k/\eps)$ vertices, then we are done, as this
    is case (A). Otherwise, let $\CC_1, \ldots, \CC_m$ be the
    connected components of
    $\Set{\GCell \cap \PolygonB}{\GCell \in \TGrid}$ that are
    $L$-shaped, in their order along $\PolygonB$. Each such connected
    component has at most two edges intersecting the interior of
    $\PolygonB$, and let $\SegSet = \brc{\seg_1, \ldots, \seg_{2m}}$
    be the segments that corresponds to these edges, again, in their
    order along $\PolygonB$.

    Consider the ``ladder''
    $L_i = \brc{\seg_{i+t}, \seg_{i+2t}, \seg_{i+3t}, \ldots}$, for
    $t= \ceil{8/\eps}$.  Every block of $\Opt$ intersects at most two
    of segments of $\SegSet$. As such, there is a choice of $i$, such
    that the weight of $L_i$ is at most $(\eps/4) \weightX{\Opt}$. We
    cut $\PolygonB$ along these edges, creating the desired
    decomposition of $\Polygon$ into \trail{}s. This is case (C), and
    it is easy to verify that the other claim in this case holds.
\end{proof}

\paragraph{The approximation algorithm for rings.}

\begin{lemma}
    \lemlab{approx:rings}%
    Let $N$, $\eps$ and $\delta$ be parameters as above, $\BSet$ be a
    set of $m$ weighted ($\delta$-large) blocks with vertices in
    $\IntRange{N}^2$, and let $k$ be a parameter. Let $\PolySetA$ be
    the set of all possible \ring{}s with at most $k$ vertices within
    $\IntRange{N}^2$. For a \ring $\Polygon \in \PolySetA$, let
    $\woptX{\tpoly}$ be the weight of the maximum weight independent
    set of blocks of $\BSet$. Then, one can $(1-\eps)$-approximate
    $\woptX{\cdot}$ for all the \trail{}s of $\PolySetA$.  The running
    time of the algorithm is $O \pth{ N^{O(k/\eps)} m }$.
\end{lemma}

\begin{proof}
    Using the algorithm of \lemref{approx:rings}, we compute the
    maximum weight independent set for each possible trail with
    $O(k/\eps)$ vertices within $\IntRange{N}^2$. Now, we compute a
    set $\PolySetA$ of all possible rings within $\IntRange{N}^2$ with
    at most $k+O(1)$ vertices each. We try all possible partitions of
    the rings as described by \lemref{ring:blazers}. The only
    difference from the algorithm of \lemref{approx:rings} is
    observing that case (C) of \lemref{ring:blazers} corresponds to
    breaking a ring into a \trail/\ring, by taking only the first
    trail in the decomposition. Clearly, repeating this choice
    multiple times would yield the desired approximation.
\end{proof}

\subsection{The result}

\subsubsection{For blocks}

\begin{lemma}
    \lemlab{blocks:PTAS}%
    Given $N$, $\eps$ and $\delta$ as above, and a set $\BSet$ of $m$
    weighted blocks contained in the square $[0,N]^2$, one can
    $(1-\eps)$-approximate a maximum weight independent set of blocks
    in $\BSet$ in $O \pth{ N^{\constB} m }$ time.
\end{lemma}

\begin{proof}
    Let $\Opt$ be the maximum weight independent set of blocks of
    $\BSet$.  By \lemrefpage{summary}, there exists a partition of
    $[0,N]^2$ into a collection of $O( 1/(\delta \eps^2 ))$ faces
    $\facesPZ$ and possibly empty area, where each face in $\facesPZ$
    is either a \trail or a \ring, with $\constA$ vertices.
    Furthermore, there is subset $\Opt' \subseteq \Opt$, of weight
    $\geq (1-\eps/4)\weightX{\Opt}$, such that each block of $\Opt'$
    is fully contained in some face of $\facesPZ$.

    As such, we enumerate all such possible partitions. There are
    $N^{\constB}$ of them. For each such partition, and each face of
    the partition, we apply the $(1-\eps/4)$-approximation algorithm
    of \lemref{approx:trails} to each face that is a \trail, and the
    algorithm of \lemref{approx:rings} if it is a \ring. In both
    cases, the algorithm is run with the subset of blocks of $\BSet$
    contained in the face.

    Clearly, this yields the desired approximation with the desired
    running time.
\end{proof}

\subsubsection{For rectangles}

\begin{theorem}
    \thmlab{delta:large}%
    Given a positive integer $N >0$, parameters $\eps >0 $ and
    $\delta > 0$, and a set $\RectSet$ of $m$ weighted rectangles,
    such that the vertices of all rectangles belong to
    $\IntRange{N}^2$, where $\IntRange{N} =\brc{0,...,N}$.  Assume
    that each rectangle in $\RectSet$ is $\delta$-large; that is,
    either its height or its width is larger than $\delta N$.  For
    this input, there is a $(1-\eps)$-approximation algorithm for
    maximum weight independent set of rectangles with a running time
    of $O \bigl( (m N)^{\constB} \bigr)$.
\end{theorem}

\begin{proof}
    Let $\Opt \subseteq \RectSet$ be the optimal solution. For every
    rectangle $\rect \in \RectSet$, pick an arbitrary block
    $\Block_\rect \subseteq \rect$, such that $\Block_\rect$ is
    parallel to the longer edge of $\rect$, and set
    $\weightX{\Block} = \weightX{\rect}$.  Let $\BOpt$ be the
    resulting set of blocks (for the set $\Opt$).

    For a trail or a ring $\Polygon$, a rectangle $\rect \in \RectSet$
    \emphi{interacts} with $\Polygon$ if $\rect$ intersects the
    boundary of $\Polygon$ but $\Block_\rect$ does not. Let
    $\IRectSetX{\Polygon}$ denote the set of rectangles of $\Opt$ that
    interact with $\Polygon$. Note, that during the execution of the
    algorithm of \lemref{blocks:PTAS} (on the set $\BSet$ with the
    solution $\BOpt$), all the \ring{}s and \trail{}s considered have
    at most $\constB$ vertices.

    So, consider such a narrow polygon $\Polygon$, and observe that
    $\cardin{\IRectSetX{\Polygon}} \leq \constB$.  Indeed, either a
    vertex of $\Polygon$ is covered by a rectangle $\rect \in \Opt$,
    or alternatively, the boundary of $\Polygon$ enters $\rect$
    through one of its short edges, and leave through the other (as
    otherwise, $\Block_\rect$ would intersect the boundary of
    $\Polygon$). But then, at least the portion of the boundary of
    $\Polygon$ covered by $\rect$ is at least $\delta N$, and as
    $\Polygon$'s boundary has length at most $\constB N$, it follows
    that the number of such rectangles in $\Opt$, since all the
    rectangles in $\Opt$ are disjoint, is bounded by
    $\frac{\constB N}{\delta N} = \constB$.

    This suggest the following. Run the algorithm of
    \lemref{blocks:PTAS} on $\BSet$ -- it enumerates hierarchically
    over partitions of the input square. For every \ring or \trail
    $\Polygon$ considered by this algorithm, guess the associated set
    of rectangles (of the optimal solution) it interacts with. Now,
    when considering a partition of such a polygon into two
    subpolygons, we also need to keep track of these sets of
    rectangles for the two subpolygons, and make sure they are
    maintained correctly during the dynamic programming.

    By the correctness of \lemref{blocks:PTAS}, one of the considered
    partitions rejects rectangles with total weight
    $\leq \eps \weightX{\BOpt} = \eps \weightX{\Opt}$. Since there are
    $m^{\constB}$ possible interacting subsets, it follows that the
    new approximation algorithm has running time $(mN)^{\constB}$, and
    yields the desired approximation.
\end{proof}

\section{Conclusions}
\seclab{conclusions}

We presented a \QPTAS for the maximum independent set of polygons
problem.  Contrasting this, the best known approximation algorithm
with polynomial running time has a performance ratio of
$n^{\eps}$. Furthermore, even for the axis-parallel rectangles case
currently has no constant factor approximation algorithm in polynomial
time. Our \QPTAS suggests that such a better polynomial time
approximation algorithms are possible.  In particular, our \PTAS for
the case of large rectangles might well turn out to be a first step
towards a \PTAS for the general case.

For our results, we presented two new techniques: the recursive
partitioning that paved the way to our \QPTAS and the partition into
$O(1)$ thin corridors and cycles in our \PTAS.  Soon after first
publishing our results these techniques were used for other geometric
problems, see \secref{impact}.  We believe that there will be more
applications of them in other geometric settings.

Recently, Chuzhoy and Ene presented a $(1+\eps)$-approximation
algorithm for independent set of (unweighted) Rectangles with a
running time of
$n^{\mathrm{poly}(\log \log n/\eps)}$~\cite{ce-amisr-16} by building
on the methodology presented in this paper and significantly extending
it.


\section*{Acknowledgments}

The authors would like to thank Chandra Chekuri, J{\'a}nos Pach, and
Kasturi Varadarajan for useful discussions on the problems studied in
this paper.


 \providecommand{\CNFX}[1]{ {\em{\textrm{(#1)}}}}
  \providecommand{\tildegen}{{\protect\raisebox{-0.1cm}{\symbol{'176}\hspace{-0.03cm}}}}
  \providecommand{\SarielWWWPapersAddr}{http://sarielhp.org/p/}
  \providecommand{\SarielWWWPapers}{http://sarielhp.org/p/}
  \providecommand{\urlSarielPaper}[1]{\href{\SarielWWWPapersAddr/#1}{\SarielWWWPapers{}/#1}}
  \providecommand{\Badoiu}{B\u{a}doiu}
  \providecommand{\Barany}{B{\'a}r{\'a}ny}
  \providecommand{\Bronimman}{Br{\"o}nnimann}  \providecommand{\Erdos}{Erd{\H
  o}s}  \providecommand{\Gartner}{G{\"a}rtner}
  \providecommand{\Matousek}{Matou{\v s}ek}
  \providecommand{\Merigot}{M{\'{}e}rigot}
  \providecommand{\CNFSoCG}{\CNFX{SoCG}}
  \providecommand{\CNFCCCG}{\CNFX{CCCG}}
  \providecommand{\CNFFOCS}{\CNFX{FOCS}}
  \providecommand{\CNFSODA}{\CNFX{SODA}}
  \providecommand{\CNFSTOC}{\CNFX{STOC}}
  \providecommand{\CNFBROADNETS}{\CNFX{BROADNETS}}
  \providecommand{\CNFESA}{\CNFX{ESA}}
  \providecommand{\CNFFSTTCS}{\CNFX{FSTTCS}}
  \providecommand{\CNFIJCAI}{\CNFX{IJCAI}}
  \providecommand{\CNFINFOCOM}{\CNFX{INFOCOM}}
  \providecommand{\CNFIPCO}{\CNFX{IPCO}}
  \providecommand{\CNFISAAC}{\CNFX{ISAAC}}
  \providecommand{\CNFLICS}{\CNFX{LICS}}
  \providecommand{\CNFPODS}{\CNFX{PODS}}
  \providecommand{\CNFSWAT}{\CNFX{SWAT}}
  \providecommand{\CNFWADS}{\CNFX{WADS}}


\end{document}